\documentclass[11pt]{article}
\usepackage{float}
\usepackage{hyperref}
\usepackage{amsmath,url}
\usepackage{amsthm}
\usepackage{amssymb}
\usepackage{algorithm}
\usepackage{algorithmic}
\usepackage{epsfig}
\usepackage{multirow}
\usepackage{pdflscape}
\usepackage{scrextend}
\usepackage{verbatim}
\usepackage[english]{babel}
\usepackage{textcomp}
\usepackage{siunitx}

\usepackage{graphicx}
\usepackage{caption}
\usepackage{scalerel}
\usepackage[usenames, dvipsnames]{color}
\definecolor{mypink1}{rgb}{0.858, 0.188, 0.478}
\definecolor{mypink2}{RGB}{219, 48, 122}
\definecolor{mypink3}{cmyk}{0, 0.7808, 0.4429, 0.1412}
\definecolor{mygray}{gray}{0.6}

\theoremstyle{plain}
\newtheorem{theorem}{Theorem}[section]
\newtheorem{lemma}[theorem]{Lemma}
\newtheorem{corollary}[theorem]{Corollary}
\newtheorem{proposition}[theorem]{Proposition}

\theoremstyle{definition}
\newtheorem{definition}[theorem]{Definition}

\theoremstyle{remark}
\newtheorem{remark}[theorem]{Remark}

\newcommand{\tail}[3]{\widehat{#1}_{#2,#3}}
\numberwithin{equation}{section}

\def\cal{\mathcal}

\begin{document}


\title{Data-adaptive trimming of the Hill estimator and detection of outliers in the extremes of heavy-tailed data \thanks{MK and SB were supported by National Science Foundation grant CNS-1422078; SB and SS
	were partially supported by the NSF grant DMS-1462368.}}



\author{Shrijita Bhattacharya\thanks{Department of Statistics, University of Michigan, 311 West Hall, 1085 S. University Ann Arbor, MI 48109-1107, {\tt \{shrijita, sstoev\}@umich.edu}}\quad Michael Kallitsis\thanks{Merit Network, Inc., 1000 Oakbrook Drive, Suite 200, Ann Arbor, MI 48104, {\tt mgkallit@merit.edu}}\quad Stilian Stoev$^\dagger$}




\maketitle

\begin{abstract} We introduce a trimmed version of the Hill estimator for the index of a heavy-tailed distribution, 
	which is robust to perturbations in the extreme order statistics.  In the ideal Pareto setting, the estimator is 
	essentially finite-sample efficient among all unbiased estimators with a given strict upper break-down point.  For general heavy-tailed models, 
	we establish the asymptotic normality of the estimator under second order regular variation conditions and also show it is minimax rate-optimal 
	in the Hall class of distributions.  We also develop an automatic, data-driven method for the choice of the trimming parameter which yields a
	new type of  robust estimator that can {\em adapt} to the unknown level of contamination in the extremes.  This adaptive robustness property
	makes our estimator particularly appealing and superior to other robust estimators in the setting where the extremes of the data are contaminated.
	As an important application of the data-driven selection of the trimming parameters, we obtain a methodology for the principled identification of extreme outliers in heavy tailed data. Indeed, the method has been shown to correctly identify the number of outliers in the previously explored Condroz data set.
\end{abstract}

\section{Introduction}

The estimation of the tail index for heavy-tailed distributions is perhaps one of the most studied problems in extreme value theory. Since the seminal works of \cite{Hill:1975, pickands:1975, hall:1982} among many others, numerous aspects of this problem and 
its applications have been explored (see e.g., the monographs \cite{embrechts:kluppelberg:mikosch:1997} and 
\cite{beirlant:goegebeur:teugels:segers:2004}).  

Let $X_1,\cdots,X_n$ be an i.i.d. sample from a distribution $F$.  We shall say that $F$ has a heavy (right) tail if:
\begin{equation}
\label{e:heavy-tail}
\mathbb{P}(X_1> x) \equiv 1-F(x) \sim \ell(x) x^{-1/\xi},\ \ \mbox{ as }x\to\infty,
\end{equation}
for some $\xi>0$ and a slowly varying function $\ell:(0,\infty)\rightarrow (0,\infty)$, i.e.,  $\ell(\lambda x)/\ell(x)\to 1,\ x\to\infty,$ for all $\lambda>0$. The parameter $\xi$ is referred to as the {\em tail index} of $F$.  Its estimation is of fundamental importance to the applications of extreme value theory (see for example the monographs \cite{beirlant:goegebeur:teugels:segers:2004}, \cite{dehaan:ferreira:2006}, \cite{resnick:2007}, and the references therein).

\begin{figure}[H]
	\hspace{-1mm}	\includegraphics[width=0.36\textwidth]{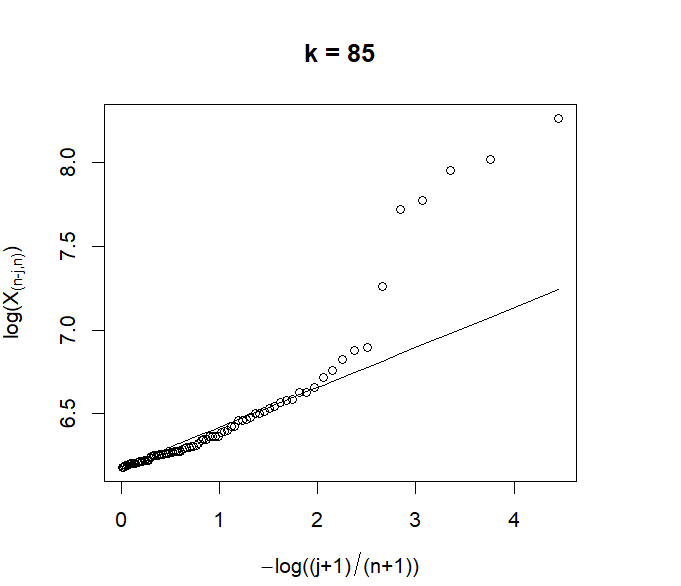}
	\hspace{-9mm}	\includegraphics[width=0.36\textwidth]{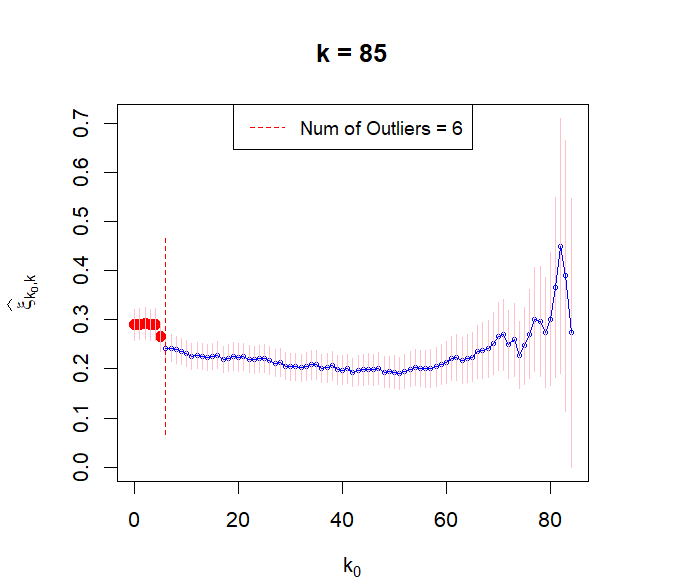}
	\hspace{-8mm}	\includegraphics[width=0.36\textwidth]{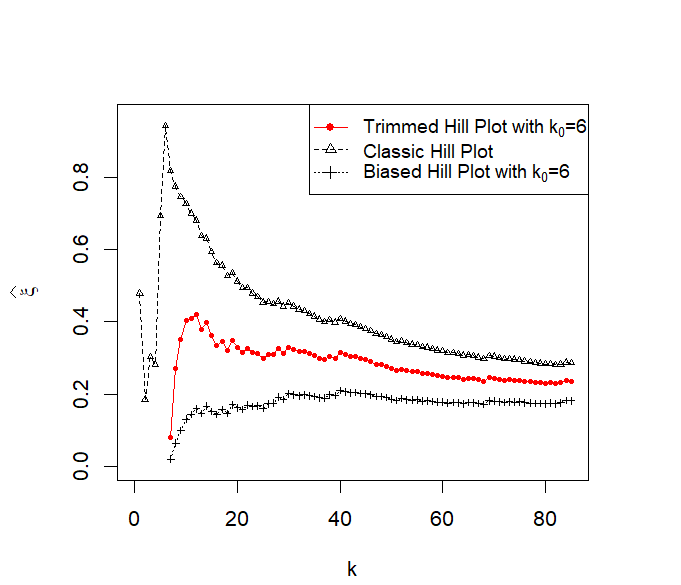}
	\caption{Exploratory plots of the Condroz data set. Left: Pareto quantile plot, Middle: Diagnostic Plot  and Right: Hill plots viz classic Hill plot, trimmed Hill plot and biased Hill plot.}
	\label{fig:condroz}
\end{figure}

The fact that the tail index  $\xi$ governs the asymptotic right tail-behavior of $F$ means that, in practice, one should estimate it by focusing on the most extreme values of the sample. In many applications, one may quickly run out of data since only the largest few order statistics are utilized.  Since every extreme data-point matters, the problem becomes even more challenging when a certain number of these large order statistics are {\em corrupted}. 
Contamination of the top order statistics, if not properly accounted for, can lead to severe bias in the estimation of the tail index. For example, the right 
panel of Figure \ref{fig:condroz} shows the classic Hill plot, its biased version and our new trimmed Hill plot for a data set which has
been previously identified to have 6 outliers (see \cite{goegebeur, VandewalleCa} and Section \ref{sec:real}, below, for more details).  We shall elaborate more on the construction
of these three plots\footnote{https://shrijita-apps.shinyapps.io/adaptive-trimmed-hill/} in the rest of the introduction but observe the drastic difference in the tail-index estimates produced by these methods. 

Recall the  {\em classic Hill} estimator of $\xi$:
\begin{equation}
\label{e:hill}
\widehat \xi_k(n):= \frac{1}{k} \sum_{i=1}^{k} \log \Bigg(\frac{X_{(n-i+1,n)}}{X_{(n-k,n)}} \Bigg), \hspace{5mm} 1\leq k\leq n-1.
\end{equation}
It is based on the top-$k$ of the order statistics: 
$$
X_{(n,n)} \geq X_{(n-1,n)}\geq  \cdots \geq  X_{(1,n)}
$$
of the sample $X_i,\ i=1,\cdots,n$. 

Naturally, one can trim a certain number of the largest order statistics in order to obtain a robust estimator of $\xi$. 
This idea has already been considered  in Brazauskas and Serfling \cite{MR1856199}, who (among other robust estimators)
defined a trimmed version of the Hill estimator:
\begin{equation}\label{e:xi-trim-gen}
\widehat{\xi}^{\rm trim}_{k_0,k}(n):= \sum_{i=k_0+1}^{k} c_{k_0,k}(i) \log \Bigg(\frac{X_{(n-i+1,n)}}{X_{(n-k,n)}} \Bigg),\hspace{5mm}  0\le k_0<k<n.
\end{equation}
where the weights $c_{k_0,k}(i)$ were chosen so that the estimator is asymptotically unbiased for $\xi$ (see Section 3.1 in \cite{MR1856199}). The weights used by Brazauskas and Serfling, however, are not optimal. In Section \ref{sec:tr-hill-plt}, we show that the asymptotically optimal trimmed 
Hill estimator has the form 
\begin{equation}\label{e:xi-trimmed}
\widehat{\xi}_{k_0,k}(n):=\frac{k_0}{k-k_0}\log \Bigg(\frac{X_{(n-k_0,n)}}{X_{(n-k,n)}} \Bigg)+\underbrace{\frac{1}{k-k_0} \sum_{i=k_0+1}^{k} \log\Bigg(\frac{X_{(n-i+1,n)}}{X_{(n-k,n)}}\Bigg)}_{\widehat{\xi}^0_{k_0,k}},\hspace{5mm}  0\le k_0<k<n.
\end{equation}
Note that if $k_0=0$ the  {\em trimmed Hill} estimator $\widehat{\xi}_{k_0,k}$ coincides with the classic Hill estimator. 

A number of authors have also considered trimming but of the {\em models} rather than the data.  Specifically,  the seminal works of  \cite{aban} and \cite{bier_truncated} studied the case where the distribution is {\em truncated} to a potentially unknown
large value.  In contrast, here we assume to have non-truncated heavy-tailed model and trim the data as a way of achieving robustness to outliers in the extremes.

Suppose now that somehow one has identified that the top-$k_0$ order statistics have been corrupted. Following \cite{trHill}, if one were to simply
ignore them and apply the classic Hill estimator to the observations $X_{(n-k_0)}\geq \cdots \geq X_{(n-k,n)}$, the estimator would be biased.  Indeed,
the second summand, $\widehat{\xi}^0_{k_0,k}(n)$ in \eqref{e:xi-trimmed} gives the expression for this {\em biased Hill} estimator.  
The recent work of Zou {\em et al} \cite{zou:davis:samorodnitsky:2017} uses this biased Hill estimator in a different inferential censoring--type  context, 
where {\em an unknown} number $k_0$ of the top order statistics is missing. 

Let us return to Figure \ref{fig:condroz} (right panel).  It shows the classic Hill plot, i.e., the plot of $\widehat{\xi}_k(n)$ as a function of $k$ as well as
the plots of $\widehat\xi_{k_0,k}(n)$ and $\widehat{\xi}_{k_0,k}^{0}(n)$ as a function of $k$.  We refer to the last two plots as to the trimmed Hill and biased Hill plots, respectively.  Since the data exhibits six outliers, the trimmed Hill and biased Hill plots are based on $k_0=6$.  The significant difference in the three
plots demonstrates the effect that outliers can have on the estimation of the tail index.

\medskip
{\em In this paper}, we introduce and study the trimmed Hill estimator $\widehat \xi_{k_0,k}(n)$ defined in \eqref{e:xi-trimmed}.  
We begin by establishing its finite sample optimality and robustness properties. Specifically, for ideal Pareto data, we establish in Theorem \ref{thm:umvue-p} that the trimmed Hill estimator is nearly minimum-variance among all unbiased estimators with given {\em strong upper break-down point} (see Definition \ref{def:ubp-est}).  Since the Pareto regime emerges asymptotically, it is not surprising that the trimmed Hill estimator is also minimax rate-optimal.  
This was shown in Theorem \ref{thm:opt-rate} for the Hall class of heavy-tailed distributions.  Furthermore, under technical second-order regular variation conditions, we establish the asymptotic normality of the trimmed Hill estimator in Section \ref{sec:asy-dist}.

The optimality and asymptotic properties of the trimmed Hill estimator although interesting are not practically useful unless one has a 
data-adaptive method for the choice of the trimming parameter $k_0$. This problem is addressed in Section \ref{sec:aut-trim}. We start by introducing \textit{diagnostic plot}\footnote{\url{https://shrijita-apps.shinyapps.io/adaptive-trimmed-hill/}\label{shiny}} to visually determine the number of outliers $k_0$. It is a plot of the trimmed Hill estimator as function of $k_0$ for a fixed $k$.  Figure \ref{fig:condroz} (middle panel) displays this 
plot for a real data set. A sudden change point at $k_0=6$ further corroborates the hypothesis of six plausible outliers in the data set.  This
value of $k_0$ was automatically identified by the method we introduce in Section \ref{sec:aut-trim}.  The methodology\footref{shiny} for the automatic selection of $k_0$
is based on a weighted sequential testing method, which exploits the elegant structure of the joint distribution of $\widehat{\xi}_{k_0,k}(n),\ k_0=0,1,\dots,k-1$ in the ideal Pareto setting. In Section \ref{sec:asy-dist}, we show that this test is asymptotically consistent in the general heavy-tailed regime
\eqref{e:heavy-tail} under second order conditions on the regularly varying function $\ell$ of \cite{bier}. In fact, the resulting estimator 
$\widehat{\xi}_{\widehat k_0,k}(n)$, where $\widehat k_0$ is automatically selected, has an excellent finite sample performance and it 
is {\em adaptively robust}. This novel adaptive robustness property is not present in other robust estimators of \cite{guillou1,guillou2,Knight_asimple,MR1856199,Peng2001,Brzezinski2016}, which involve hard to 
select tuning parameters. Also none of these estimators is able to identify outliers in the
extremes, a property inherent to the adaptive trimmed Hill estimator.  An R shiny app implementing the trimmed Hill estimator and the methodology for
selection of $k_0$ is available on \url{https://shrijita-apps.shinyapps.io/adaptive-trimmed-hill/}.

\medskip
{\em The paper is structured as follows.} In Section \ref{sec:trim-hill}, we study the benchmark Pareto setting.  We establish finite-sample optimality 
and robustness properties of the trimmed Hill estimator.  We also introduce a sequential testing method for the automatic selection of $k_0$.  
Section \ref{sec:general-heavy} deals with the asymptotic properties of the trimmed Hill estimator in the general heavy-tailed regime.  The consistency of the
sequential testing method is also studied. In Section \ref{sec:simulate}, the finite-sample performance of the trimmed Hill estimator is studied in the context of
various  heavy tailed models, tail indices, and contamination scenarios. In Sections \ref{sec:adap}, \ref{sec:L-xi} and \ref{sec:non-pareto-out}, we demonstrate 
the need for adaptive robustness and the advantages of our estimator in comparison with established robust estimators in the literature. In Section 
\ref{sec:real}, we demonstrate the application of the adaptive trimmed Hill methodology to the Condroz data set and French insurance claim settlements 
data set.

\section{Optimal and Adaptive Trimming: The Pareto Regime}
\label{sec:trim-hill}

In this section, we shall focus on the fundamental ${\rm Pareto}(\sigma,\xi)$ model and assume that
\begin{equation}\label{e:Pareto-model}
\mathbb{P}(X>x) = (x/\sigma)^{-1/\xi},\ x\ge \sigma,
\end{equation}
for some $\sigma>0$ and a tail index $\xi>0$.

 Motivated by the goal to provide a robust estimate of the tail index $\xi$, we consider trimmed versions of the {\em classical Hill} estimator in  Relation \eqref{e:hill} and thereby study the class of statistics,
$\widehat{\xi}^{\rm trim}_{k_0,k}(n)$ as in Relation \eqref{e:xi-trim-gen}. Proposition \ref{prop:xi-opt} below finds the optimal weights, $c_{k_0,k}(i)$ for which the estimator in
Relation \eqref{e:xi-trim-gen} is not only unbiased for $\xi$, but also has the minimum variance. This yields the {\em trimmed Hill}  estimator of Relation \eqref{e:xi-trimmed}. Its performance for general heavy-tailed models is discussed in Section \ref{sec:general-heavy}.
\vspace{-1mm}

\subsection{The Trimmed Hill estimator}
\label{sec:tr-hill-plt}
The following result gives the form of the trimmed Hill estimator, which is indeed the best linear unbiased estimator (BLUE) among the class of estimator in Relation \eqref{e:xi-trim-gen}
\begin{proposition} 
	\label{prop:xi-opt}
	Suppose $X_1, \cdots, X_n$ are i.i.d.\ ${\rm Pareto}(\sigma,\xi)$ random variables, as in Relation \eqref{e:Pareto-model}. Then,
	among the general class of estimators  given by Relation \eqref{e:xi-trim-gen}, the minimum variance linear unbiased estimator of $\xi$ is given by
	\vspace{-2mm}
	\begin{equation}
	\label{e:xi-opt}
	\widehat{\xi}_{k_0,k} (n)=\frac{k_0}{k-k_0} \log \Bigg(\frac{X_{(n-k_0,n)}}{X_{(n-k,n)}} \Bigg)+\underbrace{\frac{1}{k-k_0} \sum_{i=k_0+1}^{k} \log \Bigg(\frac{X_{(n-i+1,n)}}{X_{(n-k,n)}} \Bigg)}_{ \widehat{\xi}^0_{k_0,k}(n)},  \hspace{5mm} 0\leq k_0<k<n.
\vspace{-3mm}	\end{equation}
\end{proposition}
\noindent The proof is given in Section \ref{sec:proofs-sec2}.

\begin{remark}
	The second summand,  $\widehat{\xi}^0_{k_0,k}(n)$ in Relation \eqref{e:xi-opt} is nothing but the classic Hill estimator applied to the observations $X_{(n-k_0,n)} \geq \cdots \geq X_{(n-k,n)}$ which denote the top $k$ ordered statistics excluding the top $k_0$ ones. Note that, $\widehat{\xi}^0_{k_0,k}(n)$ which belongs to  the class of estimators in Relation \eqref{e:xi-trim-gen}, is not only suboptimal but also biased for the tail index $\xi$. We shall thus refer to it as the {\em biased Hill}  estimator. The biased Hill estimator has been previously used for robust analysis (see \cite{VandewalleCa}) and inference in truncated Pareto models (see \cite{trHill}, \cite{zou:davis:samorodnitsky:2017}). 
\end{remark}

\begin{remark}[\bf Classic, Biased and Trimmed Hill Plots] The {\em classic Hill plot} is a plot of the classic Hill estimator, $\widehat{\xi}_k(n)$ as function of $k$. Likewise, for a fixed $k_0$, a plot of the trimmed Hill estimator, $\widehat{\xi}_{k_0,k}(n)$ and the biased Hill estimator, $\widehat{\xi}^0_{k_0,k}(n)$ as function of $k$  will be referred to as the {\em trimmed Hill plot} and the {\em biased Hill plot}, respectively. Since  $\widehat{\xi}^0_{k_0,k}(n)\leq \widehat{\xi}_{k_0,k}(n)$, the biased Hill plot always lies below the trimmed Hill plot. Depending upon the nature of outliers in the extremes, the classic Hill plot can either lie above or below the trimmed Hill plot (see Figures \ref{fig:condroz} and \ref{fig:freclaim}).
\end{remark}

In the rest of the section, we discuss the robustness and finite-sample optimality properties of the trimmed Hill estimator. In this direction, inspired by \cite{MR1856199}, we define the notion of strict upper breakdown point.

\begin{definition}\label{def:ubp-est} A statistic $\widehat{\theta}$ is said to have a strict upper breakdown point $\beta$, $0\leq\beta<1$, if $\widehat{\theta}=T(X_{(n-[n\beta],n)},\cdots,X_{(1,n)})$ where $X_{(n,n)}\geq\cdots\geq X_{(1,n)}$ are the order statistics of the sample, i.e., $\widehat{\theta}$ is unaffected by the values of the top $[n\beta]$ order statistics.
\end{definition}
In Proposition \label{prop:xi-opt}, we showed that the trimmed Hill estimator is the BLUE for a large class of estimators with strict upper break down point of  $k_0/n$ (see Relation \eqref{e:xi-trim-gen}). We next prove a 
stronger result on the finite sample near-optimality of the trimmed Hill estimator. As stated in the next proposition, the trimmed Hill estimator is essentially the minimum variance unbiased estimator (MVUE) among the class of all 
tail index estimators with a given strict upper break down point.

\begin{theorem}
	\label{thm:umvue-p}
	Consider the class of statistics given by 
	\begin{equation*}
	{\cal{U}}_{k_0}=\left\{T=T(X_{(n-k_0,n)},\cdots,X_{(1,n)}): \: \mathbb{E}(T)=\xi, \: \mbox{ if } X_1, \cdots, X_n \stackrel{i.i.d.}{\sim} {\rm Pareto}(\sigma,{\xi})\right\}
	\end{equation*}
	which are all unbiased estimators of $\xi$ with strict upper breakdown point $\beta=k_0/n$. Then for $\widehat{\xi}_{k_0,n-1}(n)$ as in Relation \eqref{e:xi-opt}, we have
	\begin{equation}
	\label{e:opt-var}
	\frac{\xi^2}{n-k_0} \leq  \inf_{T\in	{\cal{U}}_{k_0}} {\rm Var}(T) \leq {\rm Var}(\widehat{\xi}_{k_0,n-1}(n))=\frac{\xi^2}{n-k_0-1}.
	\end{equation}
	In particular, $\widehat{\xi}_{k_0,n-1}(n)$ is asymptotically MVUE of $\xi$ among the class of estimators described by ${\cal{U}}_{k_0}$.
\end{theorem}
\noindent The proof is given in Section \ref{sec:proofs-sec3}.

\medskip
Though the trimmed Hill estimator has nice finite sample properties, it is of limited use in practice unless the value of trimming parameter $k_0$ is known.  In the 
following section, we will develop a data-driven method for the estimation of $k_0$.

\subsection{Automated Selection of the Trimming Parameter}
\label{sec:aut-trim} 

In this section, we introduce a methodology for the automated data-driven selection of the trimming parameter $k_0$. The trimmed Hill estimator with this estimated value of $k_0$ will be referred to as the {\em adaptive trimmed Hill} estimator. Its performance as a robust estimator of the tail index $\xi$ is discussed elaborately under Section \ref{sec:simulate}. In addition, the $k_0$-estimation methodology also provides a tool for the detection of outliers in the extremes of heavy tailed data. 

\medskip
We begin with a result on the joint distribution of the trimmed Hill statistics, which is a starting point towards the estimation of $k_0$.

\begin{proposition}
	\label{prop:xi-exp}
	The joint distribution of $\widehat{\xi}_{k_0,k}(n)$ can be expressed as follows:
	\begin{equation}
	\label{e:xi-jt-big}
	{\Big\{\widehat{\xi}_{k_0,k}(n),\ k_0=0,\ldots,k-1 \Big\}}
	\stackrel{d}{=} {\Big\{\xi \frac{\Gamma_{k-k_0}}{k-k_0},\ k_0=0,\ldots,k-1 \Big\}},
	\end{equation}
	where $\Gamma_i=E_1+\cdots+E_i$ with $E_1, E_2, \cdots$ i.i.d. standard exponential random variables. Consequently, as $k-k_0\rightarrow \infty$,
	\begin{equation}
	\label{e:xi-normal}\\
	\sqrt{k-k_0}(\widehat{\xi}_{k_0,k}(n)-\xi)\stackrel{d}{\implies}N(0,\xi^2)
	\end{equation}
\end{proposition}
\noindent The proof is given in Section \ref{sec:proofs-sec2}. This result motivates a simple visual device for the selection of $k_0$.
  
  \medskip
  \noindent {\bf Diagnostic Plot.}  For a fixed value of $k$, the plot of $\widehat{\xi}_{k_0,k}(n)$ as a function of of $k_0$ will be referred to as a trimmed Hill {\em diagnostic plot}.   Figure \ref{fig:knee1}, shows diagnostic plots 
  for simulated data in the cases of no outliers (left panel) and $k_0 = 5$ outliers (right panel).  The vertical lines correspond to $\widehat{\xi}_{k_0,k}(n)\underline{+}\widehat{\sigma}_{k_0,k}(n)$,
  where $\widehat{\sigma}_{k_0,k}(n)=\widehat{\xi}_{k_0,k}(n)/\sqrt{k-k_0}$ is the plug in estimate of the standard error of $\widehat{\xi}_{k_0,k}(n)$ (see Proposition \ref{prop:xi-exp}).

 \vspace{-3mm}
 \begin{figure}[H]
 	\centering
 	\includegraphics[width=0.4\textwidth]{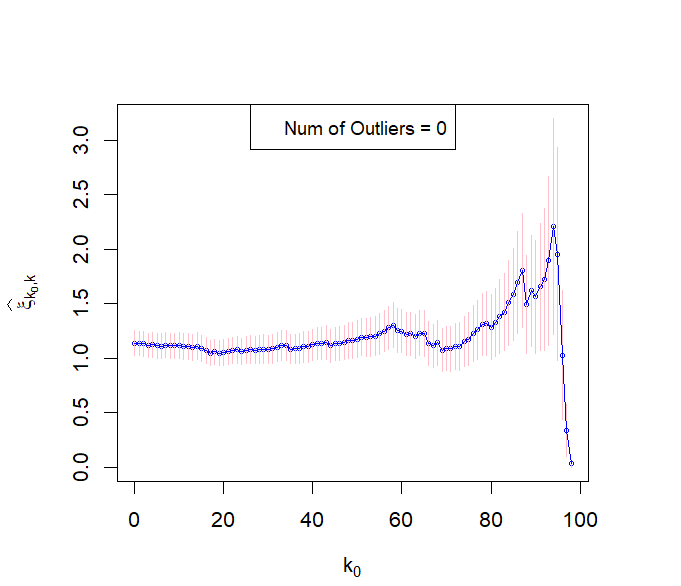}
 	\includegraphics[width=0.4\textwidth]{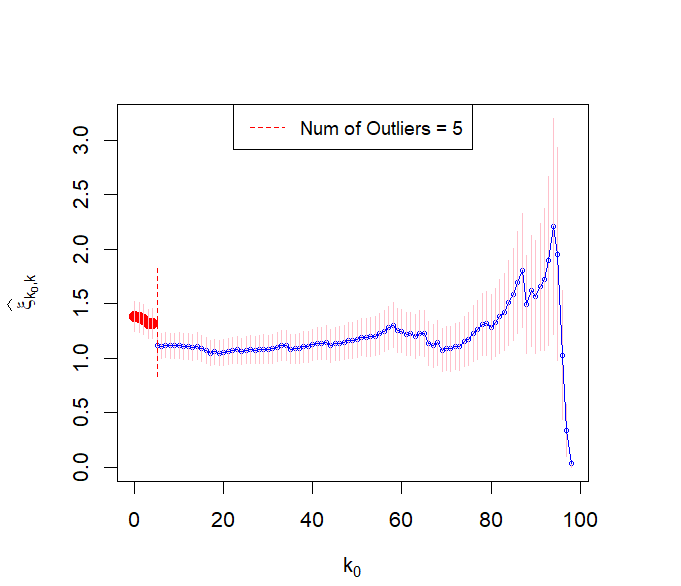}
 	\caption{Diagnostic Plot for Pareto(1,1) with $n=100, k=n-1$. {\em Left}: No outliers. {\em Right:} 5 outliers.}
 	\label{fig:knee1}
 \end{figure}
 
In the absence of outliers, modulo variability, the diagnostic plot should be constant in $k_0$ (see left panel in Figure \ref{fig:knee1}). The right panel in Figure \ref{fig:knee1} corresponds to a case where  extreme outliers have been introduced by raising the top $k_0=5$ order statistics to a power greater than 1. This resulted in a visible kink in the diagnostic plot near $k_0=5$.  Note that, in principle, the presence of outliers could lead to a kink/or change point with an upward or downward trend in the left part of the plot. The diagnostic plot, while useful, requires visual inspection of the data. In practice, an automated procedure is often desirable.

The crux of our methodology for automated selection of $k_0$ lies in the next result. The idea is to automatically detect a change point in the diagnostic plot by examining it sequentially from right to left. Formally, this will be achieved by a sequential testing algorithm involving the ratio statistics introduced next.

\begin{proposition}
	\label{T-def} 
	Suppose all the $X_i$'s are generated from ${\rm Pareto}(\sigma, \xi)$. Then, the statistics
	\begin{equation}
	\label{e:T-i-k}
	T_{k_0,k}(n):=\frac{(k-k_0-1)\widehat{\xi}_{k_0+1,k}(n)}{(k-k_0)\widehat{\xi}_{k_0,k}(n)}, \hspace{5mm} k_0=0,1,\cdots, k-2
	\end{equation} are independent and follow  ${\rm Beta}(k-k_0-1,1)$ distribution for $k_0=0,1,\cdots,k-2$. 
	\end{proposition}
\begin{proof}
	In view of Relations \eqref{e:xi-jt-big} and \eqref{e:T-i-k}, we have
	\begin{equation}
	\label{e:T-joint}
	\Big(T_{0,k}(n),\cdots,T_{k-2,k}(n)\Big)\stackrel{d}{=}\Big(\frac{\Gamma_{k-1}}{\Gamma_k},\cdots,\frac{\Gamma_1}{\Gamma_2}\Big),
	\end{equation}
	which implies
	$$T_{k_0,k}(n)\stackrel{d}{=}\frac{\Gamma_{k-k_0-1}}{\Gamma_{k-k_0}}\sim{\rm Beta}(k-k_0-1,1),\hspace{5mm}k_0=0,\cdots,k-2.$$ 
	To show the independence of the $T_{k_0,k}(n)$'s, note that, by Relation \eqref{e:gam-ind} in Lemma \ref{lem:u-ind} (below), $\Gamma_m$ and
	  $\{\Gamma_i/\Gamma_{m}, i=1,\cdots,m\}$  are independent for all $1\leq m \leq k-2$. This in turn implies that $$\Big(\frac{\Gamma_1}{\Gamma_{2}},\frac{\Gamma_2}{\Gamma_{3}}, \cdots,\frac{\Gamma_{m-1}}{\Gamma_m}\Big)\:\:\:\textmd{ and }\:\: \Gamma_m\textmd{ are independent.}$$
	Since $\Gamma_i$, $i=1, \cdots, m$ and $(E_{m+1}, \cdots, E_k)$ are independent, for all $m=1, \cdots, k-2$, we have
	\begin{equation}
	\label{e:T-i-k-joint}
	\Big(\frac{\Gamma_1}{\Gamma_2},\cdots,\frac{\Gamma_{m-1}}{\Gamma_m}\Big)
	\:\:\textmd{ and } \:\:(\Gamma_m, E_{m+1}, \cdots, E_k)\textmd{ are independent }.
	\end{equation}
Since $\Big(\Gamma_{m}/\Gamma_{m+1},\cdots,\Gamma_{k-1}/\Gamma_k\Big)$ is a function of $(\Gamma_m, E_{m+1}, \cdots, E_k)$ for all $1\leq m \leq k-2$, we have
	\begin{equation}
	\label{e:T-i-k-joint-2}
	\Big(\frac{\Gamma_1}{\Gamma_2},\cdots,\frac{\Gamma_{m-1}}{\Gamma_m}\Big)
	\:\:\textmd{ and } \:\:	\Big(\frac{\Gamma_m}{\Gamma_{m+1}},\cdots,\frac{\Gamma_{k-1}}{\Gamma_k}\Big)\textmd{ is independent for all }m\geq 1.
	\end{equation}
	In view of Relations \eqref{e:T-joint} and \eqref{e:T-i-k-joint-2}, the proof of independence of the $T_{k_0,k}(n)$'s follows.
\end{proof}

\begin{remark}
	\label{rem:T-dist}
Note that,  $T_{k_0,k}(n)$ depends only on $X_{(n-k_0,n)}, \cdots, X_{(n-k,n)}$. Therefore, the joint distribution of $T_{k_0,k}(n)$'s remains the same as long as $$(X_{(n-k_0,n)}, \cdots, X_{(n-k,n)})\stackrel{d}{=}(Y_{(n-k_0,n)}, \cdots, Y_{(n-k,n)})$$ where $Y_{(n,n)}>\cdots>Y_{(1,n)}$ are the order statistics of $n$ i.i.d. observations from ${\rm Pareto}(\sigma, \xi)$. In other words, Proposition \ref{T-def} holds even in the presence of outliers provided that they are confined only to the 
top-$k_0$ order statistics. This motivates the sequential testing methodology discussed next.
\end{remark}

\medskip
\noindent
{\bf Weighted Sequential Testing.} By Proposition \ref{T-def}, in the Pareto regime, the statistics
\begin{equation}
\label{e:U-i-k}
U_{k_0,k}(n):=2|(T_{k_0,k}(n))^{k-k_0-1}-0.5|, \hspace{5mm} k_0=0,1,\cdots, k-2.
\end{equation}
are i.i.d. $U(0,1)$. This follows from the simple observation that $T_{k_0,k}^{k-k_0-1}(n) \sim U(0,1)$. For simplicity, both in terms of notation and computation, we use the transformation in Relation \eqref{e:U-i-k} to switch from beta to uniformly distributed random variables.

Assuming that outliers affect only the top-$k_0$ order statistics, one can identify $k_0$ as the largest value $j$ for which $U_{j,k}(n)$ fails a test for uniformity. Specifically, we consider a sequential testing procedure, where starting with $j=k-2$, we test the null hypothesis $\mathcal{H}_0(j):U_{j,k}(n) \sim U(0,1)$ at level $\alpha_j$. If we fail to reject $\mathcal{H}_0(j)$, we set $j=j-1$ and repeat the process until we either encounter a rejection or $j=0$. The resulting value of $j$ is our estimate $\widehat{k}_{0}$. The methodology is formally described in the following algorithm.

\begin{algorithm}[H]
	\caption{Weighted Sequential Testing}
	\label{algo:ewst}
	\small
	\begin{algorithmic}[1]
	    \STATE Consider a set of $\alpha_j \in (0,1), j=0, 1, \cdots, k-2$.
	    \STATE Set $j=k-2$.
		\STATE Compute $U_{j,k}(n)$ as in Relation \eqref{e:U-i-k}.
		\STATE If $U_{j,k}(n)< 1-\alpha_j$, set $j=j-1$.
		\STATE If $j=0$ goto step 6 else goto step 3.
		\STATE Return $\widehat{k}_0=j$.
	\end{algorithmic}
\end{algorithm}

Since $\alpha_j$ varies as a function of $j$, we  refer to Algorithm 1 as the {\em weighted sequential testing} algorithm. The family wise error rate of the algorithm is well calibrated at level $q \in (0,1)$, provided
 \begin{equation}
\label{e:alpha-j-gen}
\prod_{j=0}^{k-2}(1-\alpha_j)=1-q.
\end{equation}
\begin{proposition} 
	\label{prop:type-I-seq}
	For i.i.d. observations from Pareto($\sigma,\xi$), let $\widehat{k}_0$ be the value from Algorithm \ref{algo:ewst} with $\alpha_j$ as in Relation \eqref{e:alpha-j-gen}. Then, under the null hypothesis $\mathcal{H}_0:k_0=0$, we have $\mathbb{P}_{\mathcal{H}_0}(\widehat{k}_0>0)=q$.
\end{proposition}

\begin{proof} We shall instead show, $\mathbb{P}_{\mathcal{H}_0}(\widehat{k}_0=0)=1-q$. Since the $U_{j,k}(n)$'s are independent $U(0,1)$,
	\begin{eqnarray} 
	\label{e:type-I-I}
	\mathbb{P}_{\mathcal{H}_0}(\widehat{k}_0=0)
	&=&\mathbb{P}_{\mathcal{H}_0}\Big(\bigcap_{j=0}^{k-2}(U_{j,k}(n)<1-\alpha_j)\Big)=\prod_{i=0}^{k-2}(1-\alpha_j)=1-q,
	\end{eqnarray}
	which completes the proof.
\end{proof}

\begin{remark}[\bf Choice of $\alpha_j$] For the purposes of this paper, the levels $\alpha_j$ in the above algorithm are chosen as follows:
	\begin{equation}
	\label{e:alpha-j}
	\alpha_j=1-(1-q)^{ca^{k-j-1}}, \hspace{5mm} j=0, \cdots, k-2
	\end{equation}
	with $a>1$ and  $c=1/\sum_{j=0}^{k-2}a^{k-j-1}$. This choice of $\alpha_j$ satisfies Relation \eqref{e:alpha-j-gen}, which in view of Proposition \ref{prop:type-I-seq}, ensures that the algorithm is well calibrated. In addition, this choice puts less weight on large values of $j$  and thereby allows for a larger type I error or fewer rejections for the hypothesis $\mathcal{H}_0(j):U_{j,k}(n) \sim U(0,1)$. This implies that large values of $j$ are less likely to be chosen over smaller ones. This guards against encountering spurious values of $\widehat{k}_0$ close to $k$, which can lead to highly variable estimates of $\widehat{\xi}_{k_0,k}(n)$. Our extensive analysis with a variety of sequential tests indicate that the choice of levels as in Relation \eqref{e:alpha-j} with $a=1.2$ works well in practice.
\end{remark}

\begin{remark} Proposition \ref{prop:type-I-seq} shows that in the Pareto case the weighted sequential testing algorithm is well calibrated and attains the exact level type I error. In the general heavy tailed regime, 
Theorem \ref{prop:U-conv} (below) establishes the asymptotic consistency of the algorithm. In Section \ref{sec:simulate}, we show that the algorithm can identify the true $k_0$  in the ideal Pareto regime as well as the challenging 
cases of Burr and T distributions (see Section \ref{sec:non-pareto-out}).
\end{remark}

\section{The General Heavy Tailed Regime}
\label{sec:general-heavy}

In this section, we study the asymptotic properties of the trimmed Hill statistics for a general class of heavy-tailed distributions $F$ as in Relation \eqref{e:heavy-tail}. 
Consider the  tail quantile function corresponding to $F$, defined as follows:
\begin{equation}\label{e:tail-qnt}Q(t)=\inf\{x:F(x)\geq 1-1/t\}=F^{-1}(1-1/t), \:\: t> 1.\end{equation}
Following \cite{bier}, for $F$ as in Relation \eqref{e:heavy-tail}, one can equivalently assume that
\begin{equation}
\label{e:L-def}
Q(t)=t^{\xi}L(t)
\end{equation}

\begin{remark}
	The relation between the slowly varying functions $\ell$ and $L$ in Relations \eqref{e:heavy-tail} and \eqref{e:L-def} is well known (see e.g., \cite{bier}, \cite{resnick:2007} and \cite{boucheron:thomas:2015}). 
	Specifically, one can show that 
	\begin{equation}
	\label{e:L-l}
	L(t^{1/\xi})\sim\ell^{\xi}(tL(t^{1/\xi})),\:\: {\rm as}\:\:\: t \rightarrow \infty.
	\end{equation}
	Thus, $\tilde{\ell}(x)=\ell^\xi(x)$ and $\tilde{L}(x)=L(x^{1/\xi})$ satisfy $\tilde{L}(x)\sim\tilde{\ell}(x\tilde{L}(x))$. This in view Theorem 1.5.13 of \cite{bingham1989regular} implies that  
	$1/\tilde{\ell}$ is the {\em de Bruijn conjugate}  of $\tilde{L}$ and hence unique up to asymptotic equivalence. 
\end{remark}

We start with a conceptually important derivation used in the rest of the section.  Using the tail-quantile function, one can express the trimmed Hill statistic for under 
the general heavy-tailed model \eqref{e:heavy-tail} as the sum of a trimmed Hill statistic based on ideal Pareto data plus a remainder term.  More precisely,  in view of \eqref{e:tail-qnt} and \eqref{e:L-def}, 
any i.i.d.\ sample $X_i,\ i=1,\dots,n$ from $F$ can be represented as:
$$
 X_i=Q(Y_i) \equiv Y_i^\xi L(Y_i),\ i=1,\dots,n,
$$
where the $Y_i$'s are i.i.d ${\rm Pareto}(1,1)$. Therefore, using $X_{(n-i,n)}=Q(Y_{(n-i,n)})$ in Relation  \eqref{e:xi-opt}, we obtain
\small
\begin{eqnarray}
\label{e:tail-3}
\hspace{-1mm}\tail \xi {k_0} k&=&\underbrace{\frac{k_0}{k-k_0}\log \frac{Y_{(n-k_0,n)}^\xi}{Y_{(n-k,n)}^\xi}+\frac{1}{k-k_0} \sum_{i=k_0}^{k-1}\log \frac{Y_{(n-i,n)}^\xi}{Y_{(n-k,n)}^\xi}}_{\large \widehat{\xi}^*_{k_0,k}(n)}+\underbrace{\frac{k_0}{k-k_0}\log \frac{L(Y_{(n-k_0,n)})}{L(Y_{(n-k,n)})} +\frac{1}{k-k_0}\sum_{i=k_0}^{k-1}\log \frac{L(Y_{(n-i,n)})}{L(Y_{(n-k,n)})}}_{R_{k_0,k}(n)}\nonumber\\
&&
\end{eqnarray}
\normalsize
where $Y_{(i,n)}$'s are the order statistics for the $Y_i$'s. Since  $Y_i^{\xi}$'s follow ${\rm Pareto}(1,\xi)$, the  statistic $\widehat{\xi}^{*}_{{k_0},k}(n)$ in \eqref{e:tail-3} is simply the trimmed Hill estimator for ideal Pareto data and $R_{k_0,k}(n)$ is a remainder term that encodes the effect of the slowly varying function $L$.

The nature of the function $L$ determines the rate at which the remainder term $R_{k_0,k}(n)$ converges to 0 in probability. We establish minimax rate optimality of the trimmed Hill estimator  under the Hall class of assumptions on the function $L$ (see  Section \ref{sec:minmax-opt}). To establish the asymptotic normality of the trimmed Hill estimator, we use second order regular variation conditions on the function $L$ (see Section \ref{sec:asy-dist}). Under the same set of conditions, the asymptotic consistency of the weighted sequential testing algorithm  is also established in Section \ref{sec:asy-seq}.

\subsection{Minimax Rate optimality of the Trimmed Hill Estimator}
\label{sec:minmax-opt}
Here, we  study the rate-optimality of the trimmed Hill estimator for the class of distributions in ${\cal D}: = {\cal D}_\xi(B,\rho)$, where Relation  \eqref{e:L-def} holds with tail index $\xi>0$ and $L$ of the form:
\begin{equation}
\label{e:D-new-def}
L(x) = 1+r(x),\ \quad\mbox{ with  }\quad |r(x)|\le B x^{-\rho},\ (x>0)
\end{equation}
for  constants $B>0$ and $\rho>0$ (see also Relation (2.7) in \cite{boucheron:thomas:2015}). This is known as the {\em Hall class} of distributions.

In \cite{HallWelsh}, Hall and Welsh showed that no estimator can be uniformly consistent over the  class of distributions in ${\cal D}$ at a rate faster than or equal to
$n^{\rho/(2\rho+1)}$.  Theorem 1 of \cite{HallWelsh} adapted to our setting and notation is as follows:

\begin{theorem}[optimal rate]
	\label{thm:opt-rate}
	Let $\widehat \xi_n$ be any estimator of $\xi$ based on an independent sample from a distribution $F\in {\cal D}_\xi(B,\rho)$. 
	If we have
	\begin{equation}
	\label{e:opt-rate}
	\liminf_{n \rightarrow \infty} \inf_{F \in {\cal{D}_\xi(B,\rho)}} \mathbb{P}_F(|\widehat \xi_n-\xi |\leq a(n))=1
	\end{equation}
	then $\liminf_{n \rightarrow \infty} n^{\rho/(2\rho+1)}a(n)=\infty$. Here by $\mathbb{P}_F$, we understand that $\widehat{\xi}_n$ was based on independent realizations from $F$.
\end{theorem}

In Theorem 3 of \cite{HallWelsh}, it is shown that for the case of no outliers,  the classic Hill estimator, $\widehat{\xi}_k$ with $k=k(n)\sim n^{2\rho/(1+2\rho)}$ is a uniformly consistent estimator of $\xi$ at a rate greater than or equal to any other uniformly consistent estimator. In other words, the classic Hill estimator is minimax rate optimal in view of  Theorem \ref{thm:opt-rate} wherein   $\widehat{\xi}_n=\widehat{\xi}_{k(n)}$ satisfies \eqref{e:opt-rate}  for every $a(n)$ with $a(n)n^{\rho/(2\rho+1)}\rightarrow \infty$. 

Note that, Theorem \ref{thm:opt-rate} also applies to the trimmed Hill estimator. We next show that in the presence of outliers, the trimmed Hill estimator  with $k=k(n)\sim n^{2\rho/(1+2\rho)}$ is minimax rate optimal   with the same rate as that of the classic Hill. In addition, the minimax rate optimality holds uniformly over all $k_0=[0,h(k)]$ for $h(k)=o(n^{2\rho/(1+2\rho)})$.

\begin{theorem}[uniform consistency] \label{t:uniform-consistency} Suppose that $k=k(n) \propto n^{2\rho/(2\rho+1)}$ and $h(k) = o(k)$, as $n\to\infty$.
	
	Then, for every sequence $a(n)\downarrow 0$, such that  $a(n) \sqrt{k(n)} \to \infty$, we have
	\begin{equation}\label{e:t:uniform-consistency}
	\liminf_{n\to\infty} \inf_{F\in {\cal D}_\xi(B,\rho)} \mathbb{P}_F\left( \max_{0\le k_0< h(k)} |\widehat \xi_{k_0,k}(n) - \xi| \le a(n) \right) =1.
	\end{equation}
\end{theorem}

The proof of this result is given in Section \ref{sec:proof-sec-minmax}. Observe that $\sqrt{k(n)} \propto n^{\rho/(1+2\rho)} $ is the optimal rate in Theorem \ref{thm:opt-rate}. Therefore,
Theorem \ref{t:uniform-consistency} implies that $\widehat \xi_{k_0,k}(n)$ is minimax rate-optimal in the sense 
of Hall and Welsh \cite{HallWelsh}. Also, note that the trimmed Hill estimator $\widehat \xi_{k_0,k}(n)$ is {\em uniformly consistent} with respect to  
both the family of possible distributions ${\cal D}$ as well as the trimming parameter $k_0$, provided $k_0 = o(k)$. 

\begin{remark} The above appealing result shows that trimming does not sacrifice the rate of estimation of $\xi$ so long as $k_0 = o( n^{2\rho/(2\rho+1)}),\ n\to\infty$.  In the regime where
the rate of contamination $k_0$ exceeds $n^{2\rho/(2\rho+1)}$, to achieve robustness and asymptotic consistency, one would have to choose 
$k(n)\gg n^{2\rho/(2\rho+1)}$, which naturally leads to rate-suboptimal estimators.  In this case, similar uniform consistency for the trimmed Hill estimators can be established along the 
lines of Theorem \ref{t:uniform-consistency}. 
\end{remark}

\subsection{Asymptotic Normality of the Trimmed Hill Estimator}

\label{sec:asy-dist}

Here, we shall establish the asymptotic normality of $\tail \xi {k_0} k$ under the general semi-parametric regime \eqref{e:heavy-tail} or equivalently \eqref{e:L-def}.
In Proposition \ref{prop:xi-exp}, we already established the asymptotic normality of the trimmed Hill estimator in the Pareto regime. Recalling Relation \eqref{e:tail-3}, we observe that
$\widehat{\xi}_{k_0,k}(n)$ differs from a tail index estimator based on Pareto data only by a remainder term $R_{k_0,k}(n)$. Thus, proving the asymptotic normality of $\widehat{\xi}_{k_0,k}(n)$ 
amounts to controlling the asymptotic behavior of the remainder term.

Indeed, we begin with a much stronger result which establishes the convergence rate of $R_{k_0,k}(n)$ {\em uniformly} for all $k_0\in [0,h(k)]$ where $h(k) \in o(k)$. To this end, 
following \cite{bier}, we adopt the following second order condition on the function $L$:
\begin{equation}\label{e:L-behav}
\sup_{t\ge t_\varepsilon}\Big| \log \frac{L(tx)}{L(t)}-cg(t)\int_{1}^x  \nu^{-\rho-1}d \nu  \Big| 
\leq \Bigg\{ 
\begin{array}{ll}
\varepsilon g(t) & \mbox{ if }\rho>0\\
\varepsilon g(t) x^\varepsilon  &  \mbox{ if }\rho=0.
\end{array}
\end{equation}
for all $\varepsilon>0$ and some $t_\varepsilon$ dependent on $\varepsilon$ and $g:(0,\infty) \rightarrow (0,\infty)$ is a $-\rho$ varying function with $\rho \geq 0$ (see Lemma A.2 in  \cite{bier} for more details.)


\vspace{2mm}

\begin{theorem}
	\label{prop:E-conv}
	Suppose the $X_i$'s are independent realizations with tail quantile function $Q$ as in Relation \eqref{e:L-def} with $L$ as in Relation \eqref{e:L-behav}. If, for some $\delta>0$ and constant $A>0$, 
	\begin{equation}\label{e:A-def}k^\delta g(n/k) \rightarrow A \:\:{\rm for}\:\: k/n \rightarrow 0 \:\:{\rm as}\:\:k,n \rightarrow \infty, \end{equation} then for $R_{k_0,k}(n)=\widehat{\xi}_{k_0,k}(n)-\widehat{\xi}^*_{k_0,k}(n)$ as in Relation \eqref{e:tail-3} and $h(k)=o(k)$, we have
	\begin{equation}
	\label{e:E-conv}
	k^\delta\max_{0\leq k_0 <h(k)}\Bigg|R_{k_0,k}(n)-\frac{cA k^{-\delta}}{1+\rho}\Bigg|\stackrel{P}{\longrightarrow}0
	\end{equation}
\end{theorem}
\noindent The proof is given in Section \ref{sec:proof-sec-normal}.

The asymptotic normality of $\widehat{\xi}_{k_0,k}(n)$ is a direct consequence of Theorem \ref{prop:E-conv} with  $\delta=1/2$ and Relation \eqref{e:xi-normal}. 
This is formalized in the following corollary.

\begin{corollary}
	\label{cor:xi-norm}
	If $k_0=o(k)$ and $\sqrt{k}g(n/k) \rightarrow A \in [0, \infty)$,
	\begin{equation}\label{e:cor:xi-norm}
	  \sqrt{k}(\widehat{\xi}_{k_0,k}(n)-\xi)\stackrel{d}{\implies}N\left(\frac{cA}{1+\rho}, \xi^2\right),\ \ \mbox{ as } n\to\infty.
	  \end{equation}
\end{corollary}

\begin{proof}
	By adding and subtracting the estimator $\widehat{\xi}^*_{k_0,k}(n)$ defined in Relation \eqref{e:tail-3}, we have
	\begin{eqnarray}
	\label{e:xi-conv}
	\sqrt{k}(\widehat{\xi}_{k_0,k}(n)-\xi)&=&\sqrt{k}(\widehat{\xi}_{k_0,k}(n)-\widehat{\xi}^*_{k_0,k}(n))+\sqrt{k}(\widehat{\xi}^*_{k_0,k}-\xi)\nonumber\\
		&=&\sqrt{k}R_{k_0,k}(n)+\sqrt{k}(\widehat{\xi}^*_{k_0,k}-\xi).
	\end{eqnarray}
	For the first term in \eqref{e:xi-conv}, we have  $\sqrt{k}R_{k_0,k} \stackrel{P}{\longrightarrow}cA/(1+\rho).$
	This follows from  Relation \eqref{e:E-conv} with $\delta=1/2$.  By Relation \eqref{e:xi-normal}, the second term in \eqref{e:xi-conv} satisfies $\sqrt{k}(\widehat{\xi}^*_{k_0,k}-\xi)\stackrel{d} {\implies } N(0,\xi^2)$, as $k\to\infty$,
	 and hence \eqref{e:cor:xi-norm} follows.
	 \end{proof}

\begin{remark}
	\label{rem:k-rho}
	Consider the asymptotic normality result of Corollary \ref{cor:xi-norm} for the Hall class of distributions in Relation \eqref{e:D-new-def}. In this case, we have $g(x) \propto x^{-\rho}$ and the convergence $\sqrt{k} g(n/k) \to A>0$ implies that $k=k(n) \propto n^{2\rho/(2\rho+1)}$, as $n\to\infty$.  This is the optimal rate, which as we know from Theorem \ref{t:uniform-consistency}, cannot be achieved by an asymptotically unbiased estimator of $\xi$.  Indeed, the limit distribution in 
	\eqref{e:cor:xi-norm} involves the bias term $cA/(\rho+1)$.  To eliminate the 
	bias term, one can pick $k =o (n^{2\rho/(2\rho+1)})$, which in this case implies that $\sqrt{k}g(n/k) \to A \equiv 0$.  That is, asymptotically unbiased estimators can be obtained but one needs to sacrifice the optimal rate.
\end{remark}

\subsection{Asymptotic behavior of the Weighted Sequential Testing}
\label{sec:asy-seq}

In this section, we establish the asymptotic consistency of the weighted sequential testing algorithm under the same set of second order regular variation conditions on the function $L$ as in Section \ref{sec:asy-dist}. We begin with a convergence result on the ratio statistics of Relation \eqref{e:T-i-k}.
\begin{theorem}
	\label{prop:T-conv}
	Assume that the conditions of Theorem \ref{prop:E-conv} hold. Then, for $\delta>0$ in Relation \eqref{e:A-def}, we have
	\begin{equation}
	\label{e:T-conv}
	k^\delta\max_{0\leq k_0 <h(k)}\Bigg|T_{k_0,k}(n)-T^*_{k_0,k}(n)\Bigg|\stackrel{P}{\longrightarrow}0,
	\end{equation}
	where $T_{k_0,k}(n)$ and $T^*_{k_0,k}(n)$ are based on $\widehat{\xi}_{k_0,k}(n)$ and $\widehat{\xi}^*_{k_0,k}(n)$, respectively as in Relation \eqref{e:T-i-k}.
\end{theorem}
The proof is described in Section \ref{sec:proof-sec-cons}.

\begin{remark}
	\label{rem:T-heavy}
	For the Hall class of distributions in \eqref{e:D-new-def} with $g(x)\sim x^{-\rho}$, if $\sqrt{k}g(n/k) \rightarrow A$ then with $k=k(n) \propto n^{2\rho/(2\rho+1)}$, the $T_{k_0,k}(n)$'s converge in distribution to ratio statistics of Pareto. Note that the order of $k$ is same as that needed for the asymptotic normality of the trimmed Hill estimator (see Remark \ref{rem:k-rho}).
\end{remark}

We next establish  that the weighted sequential testing algorithm is well calibrated and attains the significance level $q$ even for the general class of heavy tailed models in \eqref{e:L-behav}. 
\label{sec:cons-test}
\begin{theorem} 
	\label{prop:U-conv}
	Assume that the conditions of Theorem \ref{prop:E-conv} hold for some  $\delta\geq 1$. Then, for $\delta$ as in Relation \eqref{e:A-def}, we have
	\begin{equation}
	\label{e:U-conv}
	k^{(\delta-1)}\max_{0\leq k_0 <h(k)}\Bigg|U_{k_0,k}(n)-U^*_{k_0,k}(n)\Bigg|\stackrel{P}{\longrightarrow}0,
	\end{equation}
	with $U_{k_0,k}(n)$ and $U^*_{k_0,k}(n)$  based on  $T_{k_0,k}(n)$ and $T^*_{k_0,k}(n)$, respectively as in Relation \eqref{e:U-i-k}. Moreover, if the conditions of  Theorem \ref{prop:E-conv}  hold for some $\delta\geq 2$, then
	\begin{equation}\label{e:type-I-gen}
	\mathbb{P}_{\mathcal{H}_0}[\widehat{k}_0>0]\stackrel{P}{\longrightarrow}q.
	\end{equation}
\end{theorem}

The proof is given in the Section \ref{sec:proof-sec-cons}.

\begin{remark}
	To illustrate the above result, we consider the Hall class of distribution where $g(x)\sim x^{-\rho}$, as $x\rightarrow \infty$. In this case, for a given value of $\delta$, the order of $k$ which satisfies \eqref{e:A-def} for $A>0$ is given by $n^{\rho/(\rho+\delta)}$. The asymptotic normality and  minimax optimal rate for the trimmed Hill statistic, $\widehat{\xi}_{k_0,k}(n)$ is obtained for $\delta=1/2$ (see Remark \ref{rem:k-rho}). However, for the $U_{k_0,k}(n)$'s to converge and the type I error of the weighted sequential testing algorithm to be controlled, we need  $\delta\geq 1$ and $\delta\geq 2$, respectively. This would in turn produce suboptimal choices of $k$ in terms of rate. If $\rho$ is large, the difference between these suboptimal values of $k$ and the optimal value $n^{\rho/(\rho+1/2)}$ is negligible. For small values of $\rho$, the difference is greater and the consistency of the algorithm is compromised. However, in Section \ref{sec:non-pareto-out}, we show that even for smaller values of $\rho$, we do a reasonably good job in terms of determining the true number of outliers $k_0$.
\end{remark}



\section{Performance of the Adaptive Trimmed Hill Estimator.}
\label{sec:simulate}
\subsection{Simulation Set Up}
\label{sec:sim-setup}

In this section, we study the finite sample performance of the  adaptive trimmed Hill estimator, $\widehat{\xi}_{\widehat{k}_0,k}(n)$, which is the trimmed Hill statistic in Relation \eqref{e:xi-opt} with $k_0=\widehat{k}_0$ (see also {\bf https://shrijita-apps.shinyapps.io/adaptive-trimmed-hill/}). Here, the value of the trimming parameter $\widehat{k}_0$ is obtained from the weighted sequential testing algorithm in Section \ref{sec:aut-trim}. We also evaluate the accuracy of the algorithm\footnote{The parameters $a$ and $q$ are set at 1.2 and  0.05 respectively.} as an estimator of the number of outliers $k_0$.

{\bf \vspace{2mm} Measures of Performance:} The performance of an estimator $\widehat{\xi}$ of $\xi$ is evaluated in terms of its root mean squared error ($\sqrt{MSE}$), where
\begin{equation}
\label{e:mse}
{\rm MSE}(\widehat{\xi})=\mathbb{E}(\widehat{\xi}-\xi)^2.
\end{equation}
Usign criterion \eqref{e:mse}, we evaluate the performance of the adaptive trimmed Hill estimator and  several other competing estimators of the tail index $\xi$. The computation of the $\sqrt{MSE}$ is based on 2500 independent monte carlo simulations.

{\bf \vspace{2mm} Data generating models:} We generate $n$ i.i.d. observations from one of the following heavy-tailed distributions: 
\begin{eqnarray}
\label{e:heavy-dist}
{\rm Pareto}(\sigma,\xi)&:& 1-F(x)=\sigma^{1/\xi} x^{-1/\xi};\:\: x>1,\xi>0, \rho=\infty;\\\nonumber
{\rm Burr} (\eta,\lambda,\xi)&:&1-F(x)=1-\left(\frac{\eta}{\eta+x^{-1/\xi}}\right)^{-\lambda};\:\: x>0, \eta>0, \lambda>0, \xi>0,\rho=1\\\nonumber
{\rm |{\rm T}|}(\xi)&:&1-F(x)=\int_{x}^{\infty}\frac{2\Gamma(\frac{1/\xi+1}{2})}{\sqrt{n\pi}\Gamma(\frac{1}{2\xi})}\left(1+w^2\xi \right)^{-\frac{1+\xi}{2\xi}}dw; \:\: x>0,\xi>0, \rho=2\xi
\end{eqnarray}
Sections \ref{sec:adap} and \ref{sec:L-xi} deal with the performance of the weighted sequential testing algorithm and the adaptive trimmed Hill estimator for Pareto observation. Section \ref{sec:non-pareto-out} delve deeper into the performance under challenging cases of non Pareto scenarios like the  $|{\rm T}|$ and the Burr distributions.

{\bf Choice of $k$:} In \cite{HallWelsh-1}, Hall and Welsh proved that the asymptotic mean squared error of the classic Hill estimator $\widehat{\xi}_k(n)$ is minimal for \begin{equation}
\label{e:k-opt}k_n^{\rm opt} \sim{ \Big(\frac{C^{2\rho}(\rho+1)^2}{2D^2\rho^3}\Big)}^{1/(2\rho+1)}n^{2\rho/(2\rho+1)}
\end{equation}
In Theorem \ref{t:uniform-consistency}, we showed that the trimmed Hill estimator is also optimal at the same rate as the classic Hill estimator as long as the number of outliers, $k_0=o(k)$. Since for Pareto $\rho=\infty$, the optimal $k$ is $n-1$ where $n$ is the sample size. Sections  \ref{sec:adap} and \ref{sec:L-xi} which deal only with Pareto examples use this value of $k$. For Sections \ref{sec:no-out} and \ref{sec:non-pareto-out} which deal with non-Pareto examples,  $k$ is chosen approximately around the optimal value as in Relation \eqref{e:k-opt}. 

In Section \ref{sec:no-out}, we demonstrate the performance of the adaptive trimmed Hill estimators in the regime of no outliers. In this scenario,  the classic Hill   estimator (recall Relation \eqref{e:hill}) is an asymptotically optimal estimator of $\xi$ (see \cite{Hill:1975}) and is therefore used as the comparative baseline.

{\bf \vspace{2mm} Outlier Scenarios:} In Sections \ref{sec:adap}, \ref{sec:L-xi} and \ref{sec:non-pareto-out}, we demonstrate the performance of  the adaptive trimmed Hill estimator in the presence of outliers. We next discuss the mechanism of outlier injection which introduces outliers in the extreme observations of the data as follows:
\begin{enumerate}
	\item {\em Exponentiated Outliers:} The top $k_0$ order statistics are perturbed as follows:  
	\begin{equation}
	\label{e:exp-trans-0}
	X_{(n-i+1,n)}:= X_{(n-k_0,n)}+(X_{(n-i+1,n)}-X_{(n-k_0,n)}))^L, \hspace{5mm} i=1,\cdots,k_0, \hspace{5mm} 
	\end{equation}
	\item {\em Scaled Outliers:} The top $k_0$ order statistics are perturbed as
	\begin{equation}
	\label{e:scl-trans-0}
	X_{(n-i+1,n)}:= X_{(n-k_0,n)}+C(X_{(n-i+1,n)}-X_{(n-k_0,n)})), \hspace{5mm} i=1,\cdots,k_0, \hspace{5mm}
	\end{equation}
	\item{\em Mixed Outliers:} For $\mathcal{S}=\{s:X_s>\tau\}$, set \begin{equation}
	\label{e:mix-trans-0}{\{X_s\}}_{s \in \mathcal{S}}=M \tau, \hspace{5mm} M>1\end{equation}
	Thus, observations above a given threshold $\tau$ are perturbed. 
\end{enumerate}

Note that, all three nature of outliers preserve the order of the bottom $(n-k_0)$ order statistics. The exponentiated and scaled outliers preserve the order of the top $k_0$ order statistics as well. The case of mixed outliers is a challenging one because the trimming parameter $k_0$, though controlled by $\tau$, is random and not well defined. In contrast, $k_0$ is fixed and well defined for exponentiated and scaled outliers. Thus, for exponentiated and scaled outliers, we demonstrate the efficiency of the weighted sequential testing algorithm in determining $k_0$.

{\bf \vspace{2mm} Competing Robust Estimators:} In the presence of outliers, the adaptive trimmed Hill estimator is indeed a robust estimator of the tail index $\xi$. Thus, for a comparative baseline we use two other  robust estimators of the tail index in Sections \ref{sec:adap}, \ref{sec:L-xi} and \ref{sec:non-pareto-out}. These are the optimal B-robust estimator  of \cite{CJS:CJS247} and the generalized median estimator  of \cite{MR1856199}.  These estimators are indexed by two different ARE values viz 78\% and 94\% to allow for varying degrees of robustness. The constant $c$ which serves as a bound on the influence function (IF) controls for the degree of robustness for optimal B-robust estimator (see Relations (2) and (3) in \cite{CJS:CJS247}). The values $c=1.63$ and $c=2.73$ result in $78\%$ and $94\%$ asymptotic relative efficiency (ARE) values for the optimal B-robust estimator. Similarly, the parameter $\kappa$ which controls for the subset size in defining the generalized median statistic controls for the degree of robustness for generalized median estimator  (see Relation 2.2 in \cite{MR1856199}). Indeed, the values $\kappa=2$ and $\kappa=5$ produces  ARE values $78\%$ and $94\%$ respectively for the generalized median estimator. Other robust estimators of the tail index like the probability integral transform statistic estimator  of \cite{pitse} and the partial density component estimator  of  \cite{Vandewalle:2007:RET:1280299.1280640} were also considered but their results have been omitted for brevity.

\subsection{Case of No Outliers}
\label{sec:no-out}

For the three distribution models in Relation \eqref{e:heavy-dist}, we report the performance of the adaptive trimmed Hill estimator (ADAP) under the regime of no outliers. The classic Hill estimator (HILL) is used as the comparative baseline. Figure \ref{fig:h0} gives the $\sqrt{MSE}$ values for the ADAP and the HILL as a function of $k$  for the distributions for Pareto(1,$\xi$), Burr(1,0.5,$\xi$), $|$T$|(\xi)$ with tail index $\xi=2$. The value of $k$ for which the HILL has the smallest  $\sqrt{MSE}$ for a distribution is determined by Relation \eqref{e:k-opt}. This explains the observed trend in $\sqrt{MSE}$ values as a function of $k$.

\begin{figure}[H]
	\centering
	\includegraphics[width=0.36\textwidth]{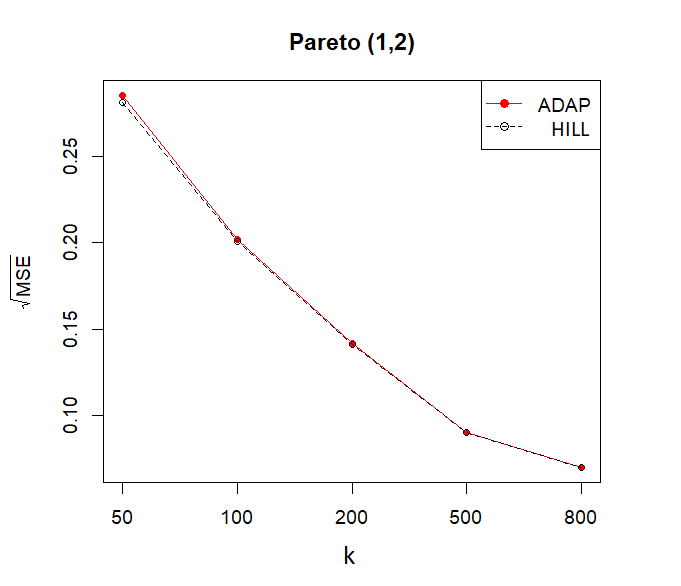}
	\hspace{-9mm}	\includegraphics[width=0.36\textwidth]{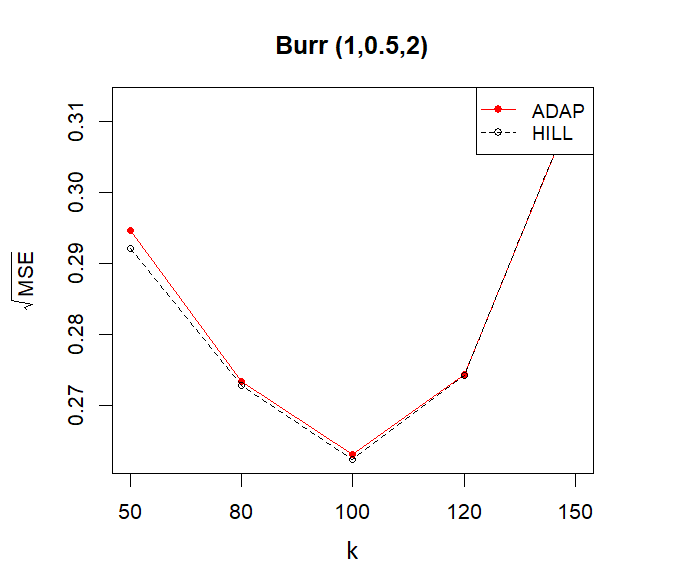}
	\hspace{-10mm}	\includegraphics[width=0.36\textwidth]{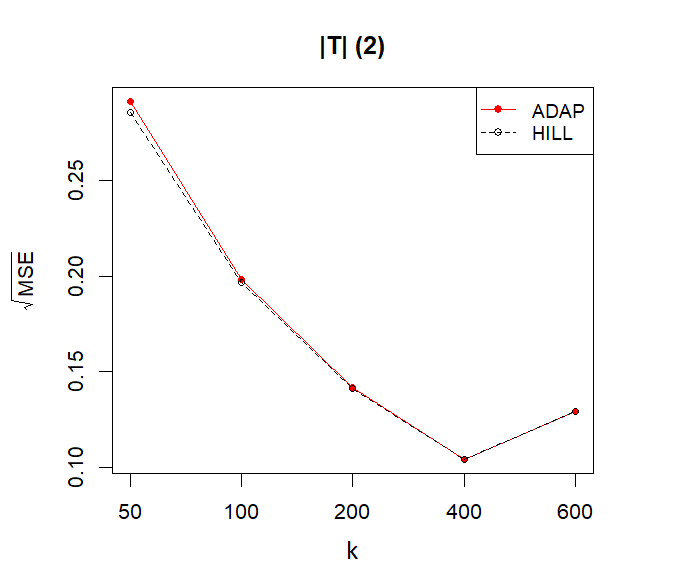}
	\vspace{-7mm}	\caption{$\sqrt{MSE}$ of the ADAP for $\xi=2$ and $k_0=0$.}
	\label{fig:h0}
\end{figure}
\vspace{-3mm}

\begin{table}[H]
	\centering
	\begin{tabular}{|c|c|c|c|c|c|}
		\hline
		Pareto(1,2) & 0.0500 (k=50) & 0.0472 (k=100) & 0.0496 (k=200) & 0.0448 (k=500) & 0.0556 (k=800)\\
		Burr(1,0.5,2) & 0.0476 (k=50)& 0.0496(k=80)& 0.0416 (k=100)& 0.0408 (k=150) & 0.0392 (k=200)\\
		$|$T$|$(2) &  0.0484 (k=50) & 0.0556 (k=100) & 0.0472 (k=200) & 0.0520 (k=400) & 0.0464 (k=600)\\\hline
	\end{tabular}
	\caption{Type 1 error of the weighted sequential testing algorithm for $k_0=0$.}
	\label{tab:type-I}
\end{table}

We observe that for a wide range of $k$, the ADAP  is virtually indistinguishable from the HILL irrespective of the distribution under study. This indicates  that the weighted sequential testing algorithm can precisely determine $k_0=0$ for the same wide range of $k$-values  as in Figure \ref{fig:h0}. Indeed, Table \ref{tab:type-I} shows that the algorithm attains the nominal significance level of $q=\mathbb{P}(\widehat{k}_0>0)=0.05$. This encouraging finite sample performance complements the theoretically established consistency of the algorithm in Theorem \ref{prop:U-conv}.

\subsection{Adaptive Robustness}
\label{sec:adap}

In this section, we study how the presence of outliers in the data influences the performance of the  adaptive trimmed Hill estimator (ADAP) and the weighted sequential testing algorithm. For clarity and simplicity, the data in this section are generated from Pareto as in Relation \eqref{e:heavy-dist} with $\sigma=1, \xi=2$ for varying sample sizes $n=100,300,500$.

The value of $k$ is fixed at $n-1$ which is indeed the optimal $k$ for the Pareto regime (see Relation \eqref{e:k-opt}). Section \ref{sec:non-pareto-out} illustrates the adaptive robustness phenomenon as explained in this section in the context of other heavy tailed models. Outliers are injected by Relations \eqref{e:exp-trans-0}, \eqref{e:scl-trans-0} and \eqref{e:mix-trans-0} with $L=3$, $C=200$ and $M=100$. Varying values of the parameter $k_0$ and $\tau$ are chosen to control for the number of outliers in the data.

Figures \ref{fig:ha-k0-exp}, \ref{fig:ha-k0-scl} and \ref{fig:ha-tau-mix} produce a plot of  the $\sqrt{MSE}$ for the ADAP for outlier generating mechanisms in Relations \eqref{e:exp-trans-0}, \eqref{e:scl-trans-0} and \eqref{e:mix-trans-0} respectively. For comparison, the performance of the optimal B-robust estimator (OBRE) and the generalized median estimator (GME) at 78\% and 94\% ARE levels have also been included. The figures clearly show that ADAP is uniformly the best estimator in terms $\sqrt{MSE}$. The figures also show an intriguing  adaptive robustness property of our estimator. Namely, its  $\sqrt{MSE}$ is nearly flat and grows slowly with increase in the degree of contamination (parametrized by either the number of outliers $k_0$ in Figures \ref{fig:ha-k0-exp} and \ref{fig:ha-k0-scl} or the threshold $\tau$ in Figure \ref{fig:ha-tau-mix}). On the other hand, the competing estimators break down completely with increase in the degree of contamination. This can be explained as: the competing estimators must be calibrated to a predefined level of robustness by setting their ARE level in advance.  To the best of our knowledge, none of the existing works in the literature provide a data-driven method for selecting this optimal ARE value. In contrast,  the trimming parameter $k_0$ involved in the ADAP is estimated from the data itself which allows it to adapt itself to  unknown degrees of contamination in the data.

Figures \ref{fig:ha-k0-exp} and \ref{fig:ha-k0-scl} show that whenever the target ARE value is greater than $(1-k_0/n)\times 100\%$, the performance of the ADAP is much superior to that of the competing estimators. For example, the   OBRE-94 and the GME-94 breakdown completely where $1-k_0/n\leq 0.9$ ($n=100, k_0\geq 15$ and $n=300,  k_0\geq 30$). Similarly, the performance of the OBRE-78 and the GME-78 is drastically poor where $1-k_0/n\leq 0.7$ ($n=100, k_0\geq 30$). An estimator indexed by a higher ARE value has greater efficiency provided $ARE\leq(1-k_0/n)\%$. This explains why the performance of the OBRE-78 and the GME-78 is quite poor in comparison to that of the OBRE-94 and the GME-94 where $1-k_0/n \leq 0.95$ ($n=100, k_0\leq 5$ and $n=300, k_0\leq  15$). By automatically estimating the number of outliers, ADAP not only produces an estimator of $\xi$ robust to varying levels of data contamination but also provides a methodology for outlier detection in the extremes of heavy tailed models. 

\vspace{-3mm}
\begin{figure}[H]
	\centering
	\hspace{-1mm}	\includegraphics[width=0.36\textwidth]{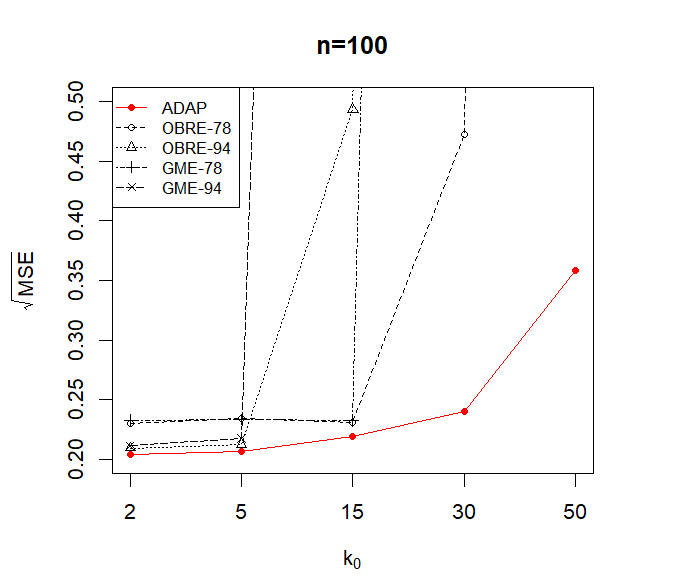}
	\hspace{-9.5mm}	\includegraphics[width=0.36\textwidth]{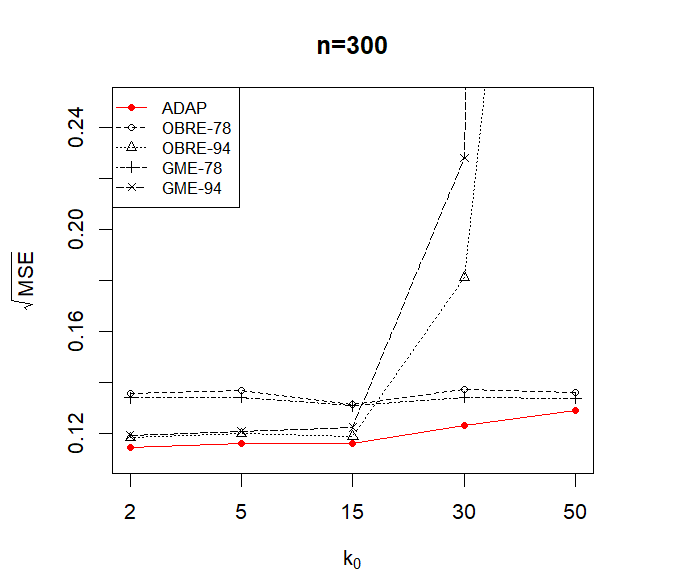}
	\hspace{-9.5mm}	\includegraphics[width=0.36\textwidth]{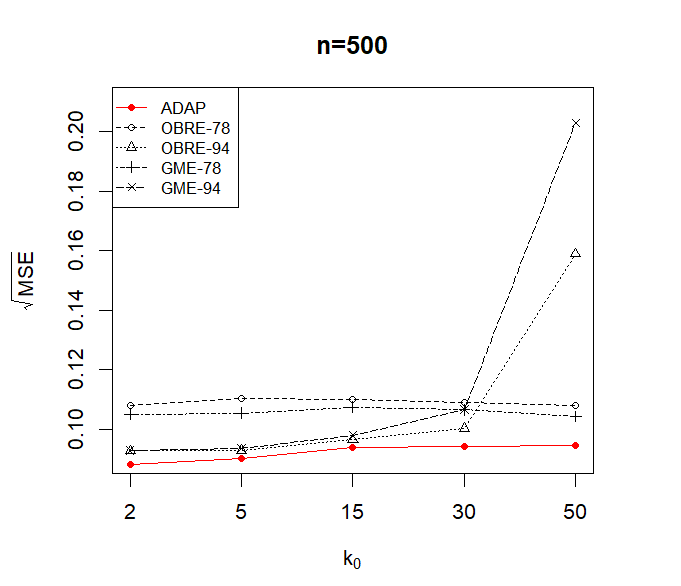}
	\caption{$\sqrt{MSE}$ of ADAP for Pareto(1,2) with exponentiated outliers: $L=3$, varying $k_0$.}
	\label{fig:ha-k0-exp}
\end{figure}

\vspace{-5mm}
\begin{figure}[H]
	\centering
	\hspace{-1mm}	\includegraphics[width=0.36\textwidth]{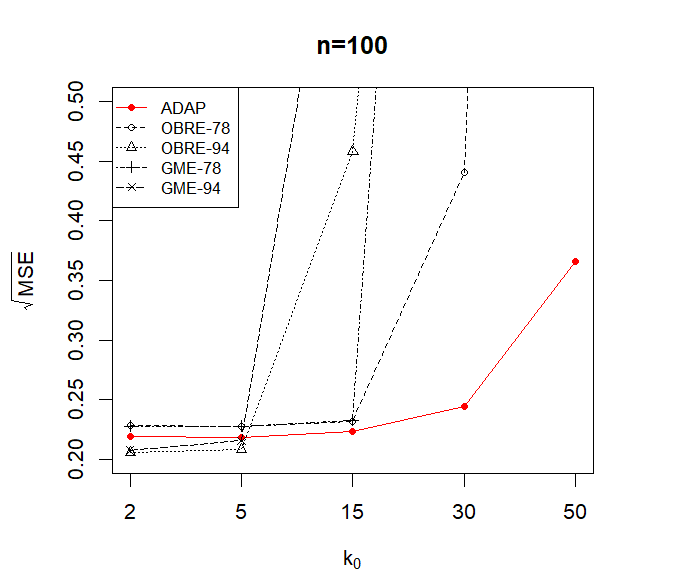}
	\hspace{-9.5mm}	\includegraphics[width=0.36\textwidth]{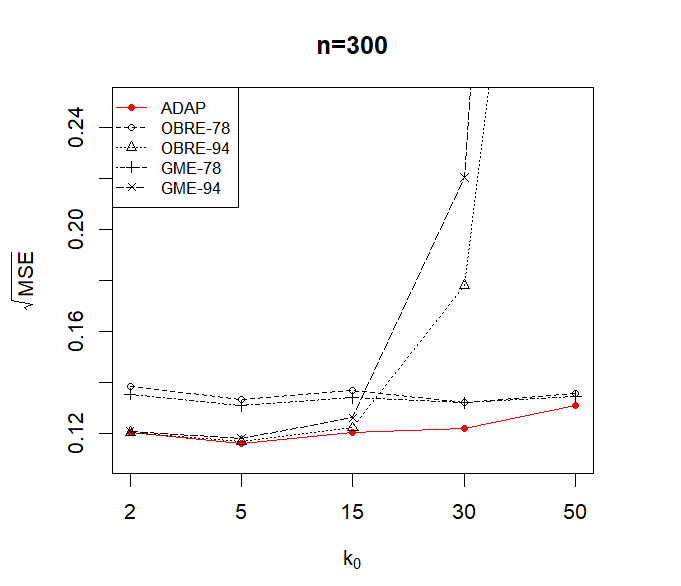}
	\hspace{-9.5mm}	\includegraphics[width=0.36\textwidth]{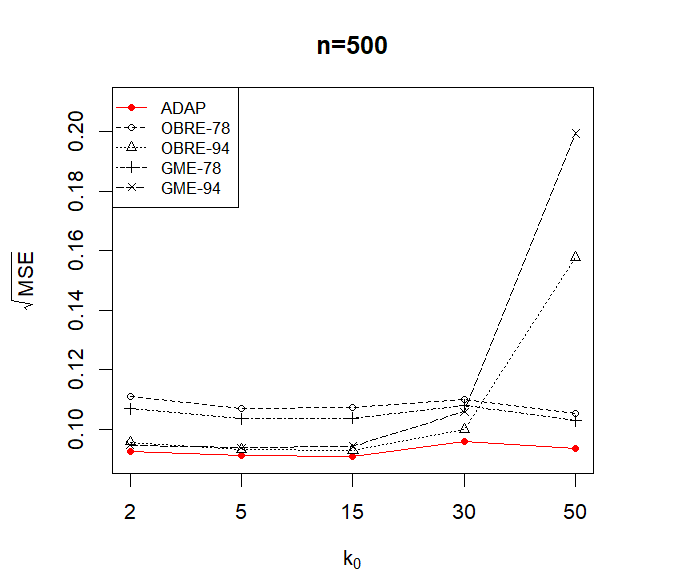}
	\caption{$\sqrt{MSE}$ of ADAP for Pareto(1,2) with scaled outliers: $C=200$, varying $k_0$.}
	\label{fig:ha-k0-scl}
\end{figure}

\vspace{-5mm}
\begin{figure}[H]
	\centering
	\hspace{-1mm}	\includegraphics[width=0.36\textwidth]{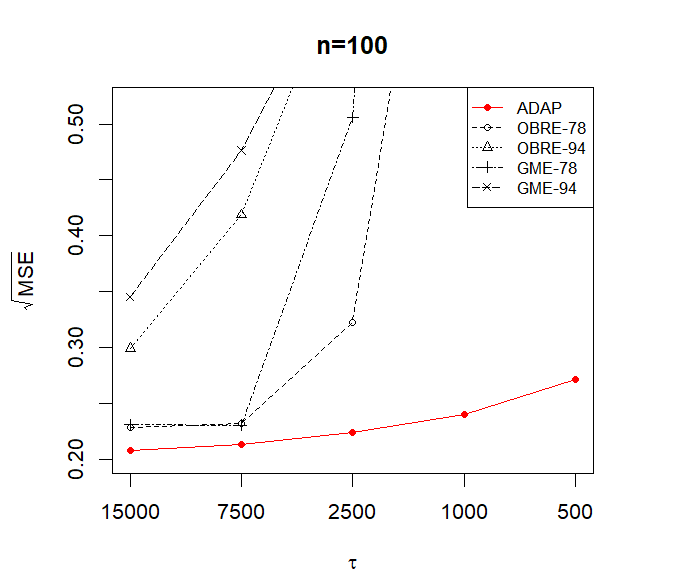}
	\hspace{-9.5mm}	\includegraphics[width=0.36\textwidth]{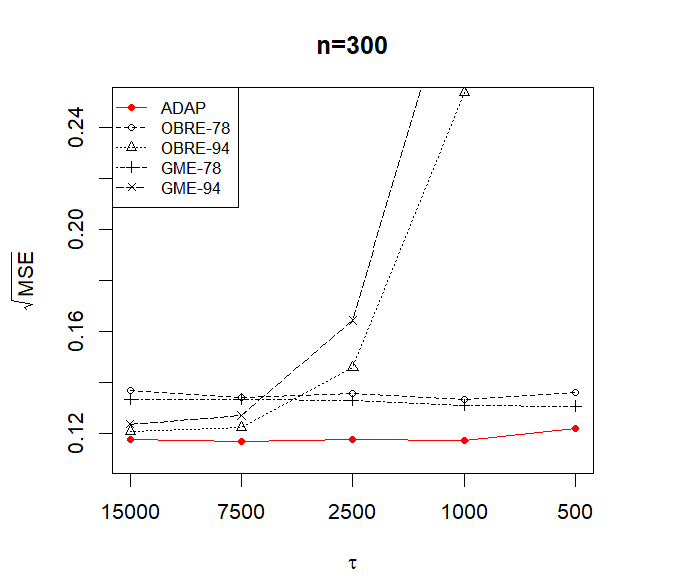}
	\hspace{-9.5mm}	\includegraphics[width=0.36\textwidth]{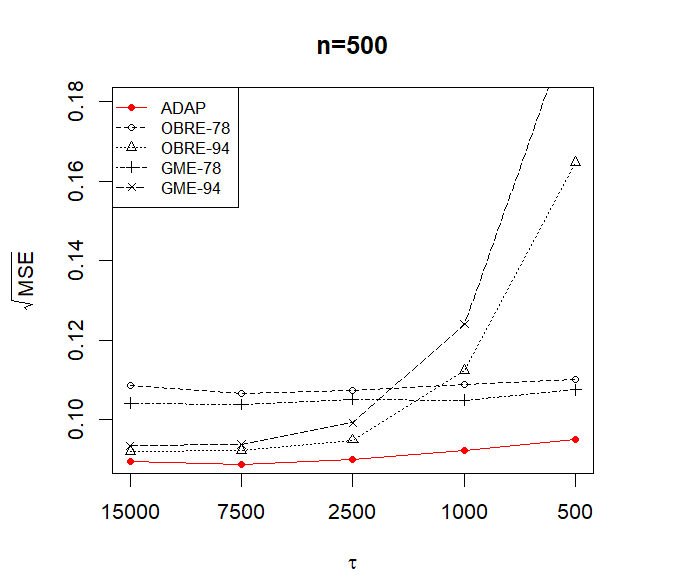}
	\caption{$\sqrt{MSE}$ of ADAP for Pareto(1,2) with mixed outliers: $M=100$, varying $\tau$.}
	\label{fig:ha-tau-mix}
\end{figure}

Indeed, Tables \ref{tab:ha-k0-exp} and \ref{tab:ha-k0-scl}  which produce the mean and standard errors of $\widehat{k}_0$ for outliers injected by mechanisms \eqref{e:exp-trans-0} and  \eqref{e:scl-trans-0}, show that for all values of $n$, the weighted sequential testing algorithm picks up the true number of outliers $k_0$  for almost all values $k_0$ (exception is $k_0=2$ for scaled outliers).

\begin{table}[H]
	\centering
	\begin{tabular}{|c|c|c|c|c|c|}
		\hline
		$n$ &  $k_0=2$ & $k_0=5$ & $k_0=15$ & $k_0=30$ & $k_0=50$\\\hline
		100 &$2.19\pm 1.42$&$5.10\pm 1.04$& $14.99\pm 0.51$  &$ 29.84\pm 0.41$& $49.47\pm  0.78$\\
		300 & $2.23\pm 1.61$ & $5.08\pm 0.95$ & $14.98\pm 0.44$ & $29.85\pm 0.44$ & $49.55\pm 0.70$\\
		500 &$ 2.17\pm 1.19$ & $5.20\pm 3.95$ & $ 14.98\pm  0.49$ & $ 29.85\pm 0.39$ & $49.55 \pm 0.70$\\\hline
	\end{tabular}
	\caption{{$\mathbb{E}(\widehat{k}_0)\pm {\rm Standard\:\:Error}(\widehat{k}_0)$ for Pareto(1,2) with exponentiated outliers, $L=3$.}}
	\label{tab:ha-k0-exp}
\end{table}

\vspace{-3mm}
\begin{table}[H]
	\centering
	\begin{tabular}{|c|c|c|c|c|c|}
		\hline
		$n$ & $k_0=2$ & $k_0=5$ & $k_0=15$ & $k_0=30$ & $k_0=50$
		\\\hline
		100 &$1.10\pm 2.09$& $4.66\pm 1.87$& $14.91\pm 0.90$ & $29.89\pm 0.70$ & $49.68\pm 3.01$\\
		300 &$1.06\pm 1.85$ & $4.68\pm 1.75$ & $14.94\pm 1.02$& $29.91\pm 0.84$ & $49.88 \pm 0.39$\\
		500 &$1.09\pm 1.96$& $4.69\pm 1.83$ & $14.91\pm 0.82$ & $29.97\pm2.81$ & $49.89\pm  0.37$
		\\\hline
	\end{tabular}
	\caption{{$\mathbb{E}(\widehat{k}_0)\pm {\rm Standard\:\:Error}(\widehat{k}_0)$ for Pareto(1,2) with scaled outliers, $C=200$.}}
	\label{tab:ha-k0-scl}
\end{table}

\subsection{\bf Impact of Outlier Severity and Tail Index.}
\label{sec:L-xi}

In this section, we study the influence of the magnitude of outliers and tail index on the performance of the adaptive trimmed Hill estimator (ADAP) for Pareto observations with sample size $n=500$. The conclusions were similar for other heavy tailed models explored.

\vspace{-2mm}
\begin{figure}[H]
	\centering
	\includegraphics[width=0.32\textwidth]{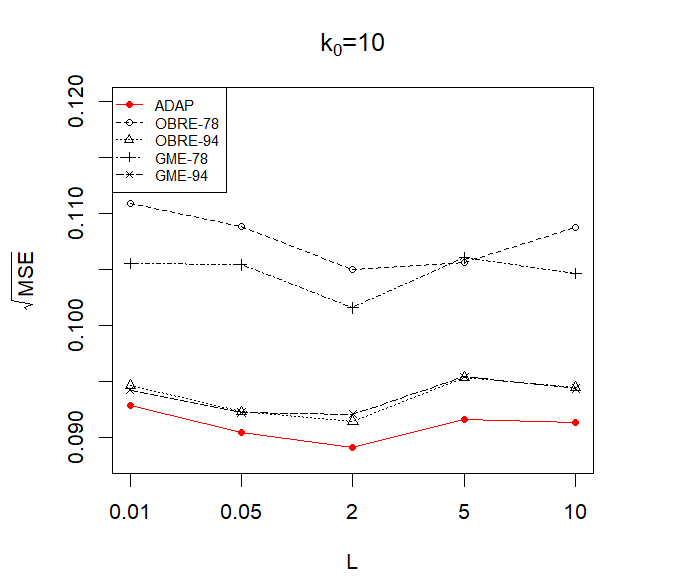}
	\includegraphics[width=0.32\textwidth]{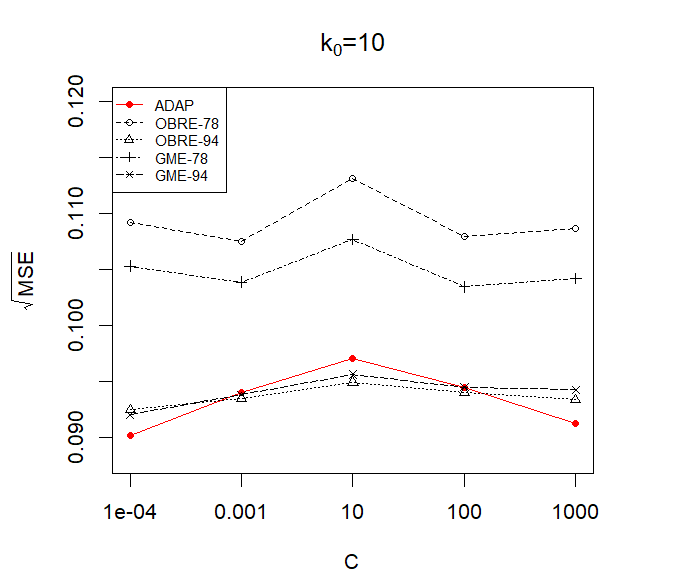}
	\includegraphics[width=0.32\textwidth]{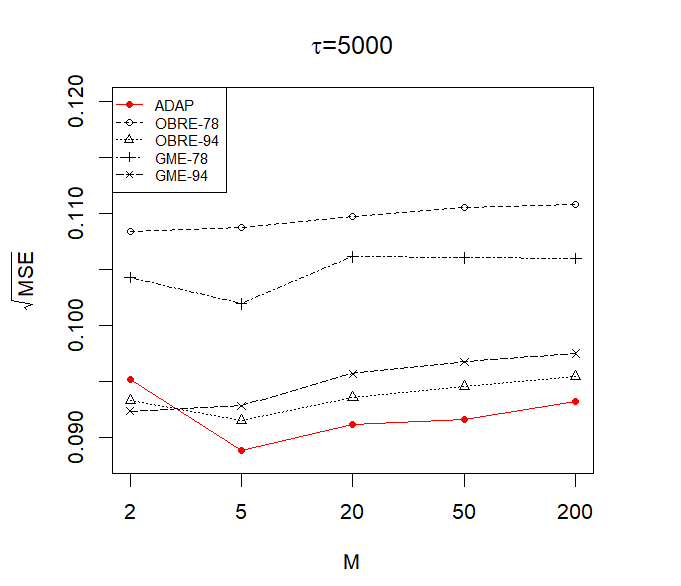}
	\caption{$\sqrt{MSE}$ of ADAP for Pareto(1,2). {\em Left:} Exponentiated  outliers with varying $L$. {\em Middle:} Scaled outliers with varying $C$. {\em Right:} Mixed outliers with varying $M$.}
	\label{fig:ha-L}
\end{figure}

We begin with the impact of outlier severity on the performance of the ADAP. For outlier generating mechanisms in Section \ref{sec:sim-setup}, the outlier severity is controlled by the parameters $L$, $C$ and $M$. The data generating model is Pareto as in Relation \eqref{e:heavy-dist} with $\sigma=1$ and $\xi=2$. Figure \ref{fig:ha-L} produces a plot of  the $\sqrt{MSE}$ for ADAP for outlier generating mechanisms in Relations \eqref{e:exp-trans-0}, \eqref{e:scl-trans-0} and \eqref{e:mix-trans-0} with $k_0=10$, $\tau=5000$ and varying $L,C$ and $M$. For comparison, $\sqrt{MSE}$ values for the optimal B-robust estimator (OBRE) and the generalized median (GME) at 78\% and 94\% ARE levels have also been included. The ADAP  performs better than both the OBRE and the GME for almost all values of $L$, $C$ and $M$ no matter what their ARE levels is. The only exception is $C=10$ for the case scaled outliers (see Relation \eqref{e:scl-trans-0}). Though more robust, the estimators  the OBRE-78 and the GME-78 perform poorly at lower levels of contamination in the data. This explains their inferior behavior at $n=500, k_0=10$ where the degree of contamination is only $5\%$.

\begin{table}[H]
	\centering
	\begin{tabular}{|c|c|c|c|c|c|}
		\hline
		Exponentiated outliers	 &$L=0.01$ & $L=0.05$ & $L=2$ & $L=5$ & $L=10$\\
		&$9.74 \pm 1.06$ & $9.59 \pm  1.11$& $9.91 \pm  0.51$& $10.05 \pm 0.87 $& $10.05 \pm 0.71$\\\hline
		Scaled outliers &	$C=0.0001$ & $C=0.001$ & $C=10$ & $C=100$ & $C=1000$\\
		&	$10.05\pm 0.70$ &$8.83\pm 2.17$ &$3.86\pm 4.73$ &$9.57\pm 1.83 $&$10.03\pm 0.61$\\\hline
	\end{tabular}
	\caption{$\mathbb{E}(\widehat{k}_0)\pm {\rm Standard\:\:Error}(\widehat{k}_0)$ for Pareto (1,2) with $k_0=10$ outliers.}
	\label{tab:ha-L}
\end{table}

The superiority of the ADAP grows with increase in the magnitude of the outliers. For exponentiated and scaled outliers, the increase in magnitude is manifested through increasing values of $|\log(L)|$ and $|\log(C)|$, respectively\footnote{ $|\log x|$ is an increasing function of $x$ for $x>1$ and a decreasing function of $x$ for $x<1$.}.  For mixed outliers, the increase in magnitude occurs with the  increase in the value of $M$.  With an increase in magnitude, the weighted sequential testing algorithm can correctly detect the true number of outliers $k_0$ (see Table \ref{tab:ha-L}) and hence the greater efficiency of ADAP.

We next study the impact of the tail index $\xi$ on the performance of ADAP. The data generating model is Pareto as in Relation \eqref{e:heavy-dist} with $\sigma=1$ and varying values of $\xi$. Outliers are injected according to Relations \eqref{e:exp-trans-0}, \eqref{e:scl-trans-0} and \eqref{e:mix-trans-0}   with $k_0=10$, $\tau=5000$, $L=3$, $C=200$ and $M$.  Figure \ref{fig:ha-xi} produces a plot of the $\sqrt{MSE}$ values for the ADAP along with those of the OBRE and the GME at 78\% and 94\% ARE levels. The performance of the ADAP is superior to that of the remaining estimators.
For exponentiated  and mixed outliers, the improvement is even more prominent at larger values of $\xi$. This is because for the same values of $L$ and $M$, the severity of outliers is greater for heavier tails ($\xi=2.5$) than lighter ones ($\xi=0.5$).  In contrast, for scaled outliers, the improvement is more prominent at smaller $\xi$ values. This is because for the same value of $C$, the severity of outliers is greater for lighter tails than heavier ones. This is in consensus with the findings of  Table \ref{tab:ha-xi} where the  accuracy of  the weighted sequential testing algorithm in correctly estimating the true number of outliers improves  with increase in $\xi$ for exponentiated and mixed outliers and deteriorates with increase in $\xi$ for scaled outliers.

\vspace{-2mm}
\begin{figure}[H]
	\centering
	\hspace{-1mm}	\includegraphics[width=0.36\textwidth]{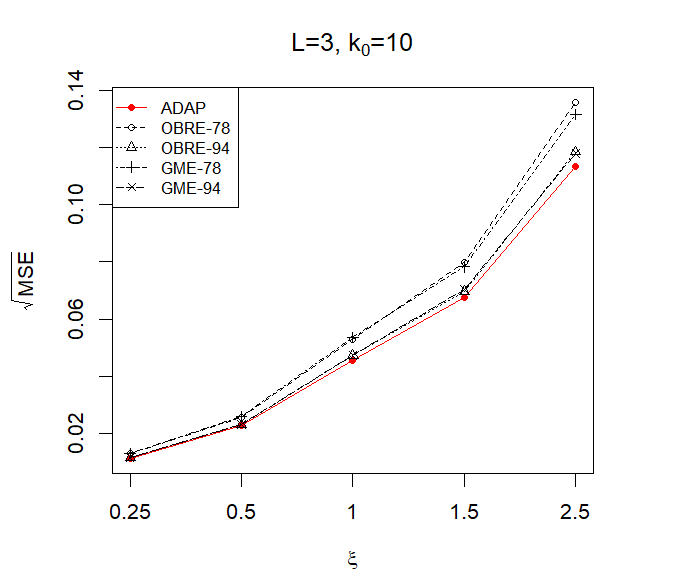}
	\hspace{-10mm}	\includegraphics[width=0.36\textwidth]{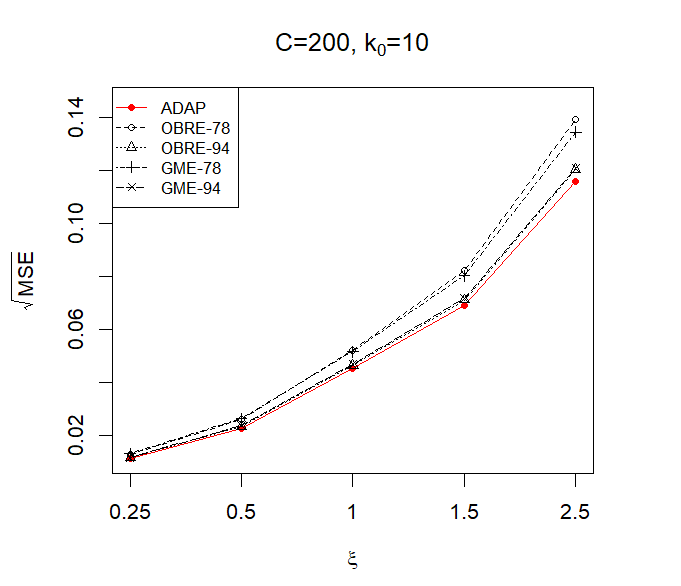}
	\hspace{-10mm}	\includegraphics[width=0.36\textwidth]{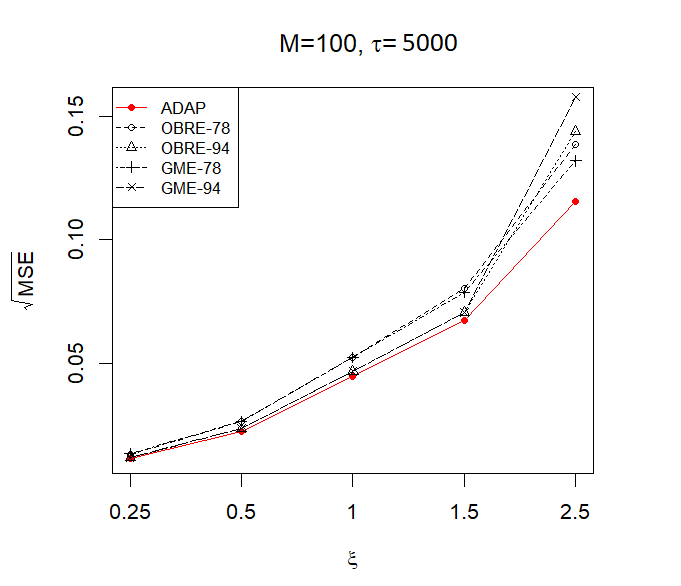}
	\caption{$\sqrt{MSE}$ of ADAP for Pareto(1,$\xi$) for varying $\xi$. {\em Left:} Exponentiated  outliers. {\em Middle:} Scaled outliers. {\em Right:} Mixed outliers.}
	\label{fig:ha-xi}
\end{figure}

\begin{table}[H]
	\centering
	\begin{tabular}{|c|c|c|c|c|c|}
		\hline
		& $\xi=0.25$ & $\xi=0.5$ & $\xi=1$ & $\xi=1.5$ & $\xi=2.5$\\\hline
		Exponentiated outliers &	$5.68\pm 4.23$&$ 7.33\pm 2.32$&$ 9.57\pm 1.03$&$  9.95\pm 0.73 $&$ 10.03\pm  0.51$\\\hline
		Scaled outliers	&$10.00\pm 0.59$&$ 10.01\pm 1.18$&$  9.98\pm 0.72$&$ 9.99\pm 1.05$&$ 9.79\pm  1.47$\\\hline
	\end{tabular}
	\caption{$\mathbb{E}(\widehat{k}_0)\pm {\rm Standard\:\:Error}(\widehat{k}_0)$ for Pareto (1,$\xi$) with $k_0=10$ outliers for $L=3$ and $C=200$.}
	\label{tab:ha-xi}
\end{table}

\subsection{Outliers in Non Pareto distributions}
\label{sec:non-pareto-out}

In this section,  $n=1000$  sample points are generated from non-Pareto distributions as in Relation \eqref{e:heavy-dist}. These include be the $|{\rm T}|$($\xi$) and the Burr($\eta$,$\lambda$,$\xi$) distribution with $\xi=2$, $\eta=1$ and $\lambda=0.5$.  Outliers are injected by mechanisms \eqref{e:exp-trans-0}, \eqref{e:scl-trans-0} and \eqref{e:mix-trans-0} for by $L=3$, $C=200$, $M=100$, $k_0=10$ and $\tau=5000$. The adaptive trimmed Hill estimator (ADAP) is constructed  for $k$ in the neighborhood of its optimal\footnote{Optimal $k$ is the one which produces the asymptotically minimum variance for the classic Hill estimator. The optimal $k$ for the trimmed Hill estimator is also of the same order as that of the classic Hill estimator (see Theorem \ref{e:t:uniform-consistency}). For a sample of size $n=1000$, the optimal $k$ is $464$ and $97$ for $|T|$ and Burr distributions respectively (see Relation \eqref{e:k-opt}).} value.

\vspace{-3mm}
\begin{figure}[H]
	\centering
	\hspace{-1mm} \includegraphics[width=0.36\textwidth]{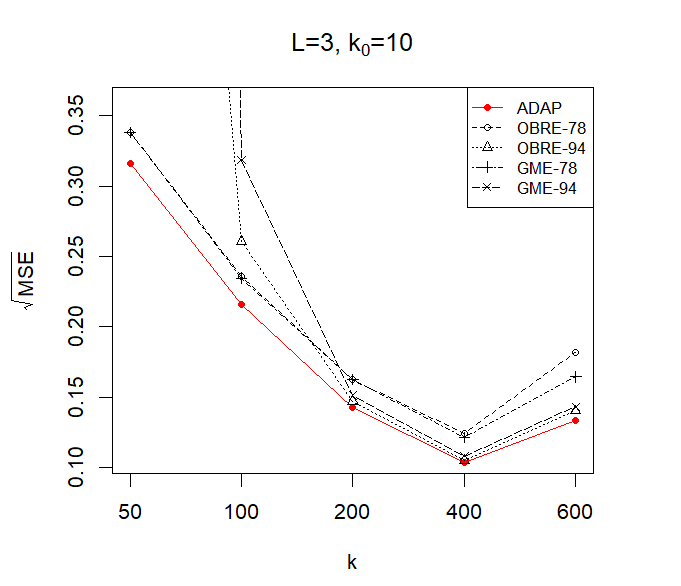}
	\hspace{-8.5mm}\includegraphics[width=0.36\textwidth]{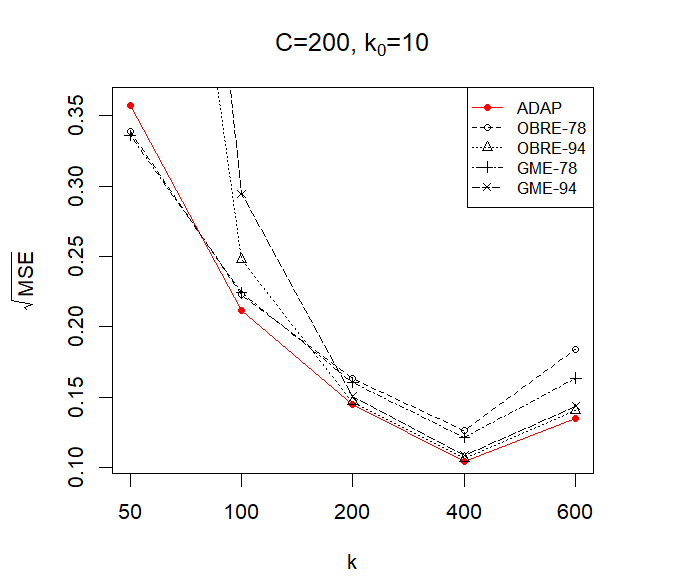}
	\hspace{-8.5mm}\includegraphics[width=0.36\textwidth]{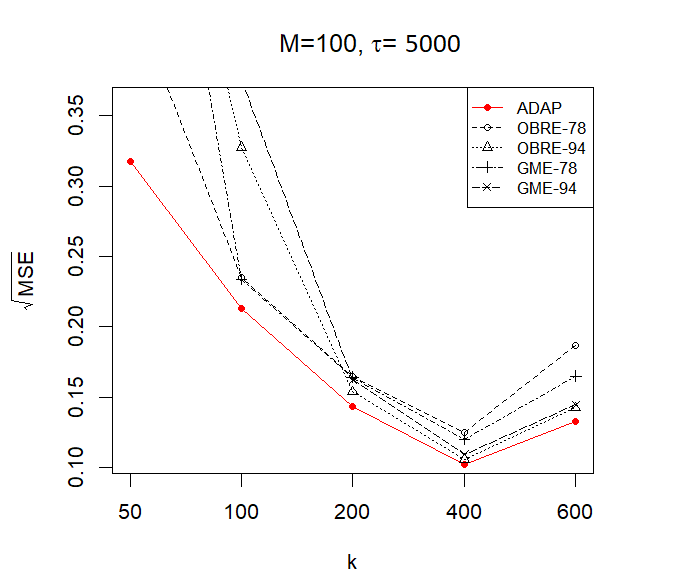}
	\caption{$\sqrt{MSE}$ of ADAP for $|$T$|$(2) as a function of $k$. {\em Left}: Exponentiated Outliers. {\em Middle}: Scaled Outliers. {\em Right}: Mixed Outliers.}
	\label{fig:ha-t}
\end{figure}

\vspace{-5mm}
\begin{figure}[H]
	\centering
	\hspace{-2mm} 	\includegraphics[width=0.36\textwidth]{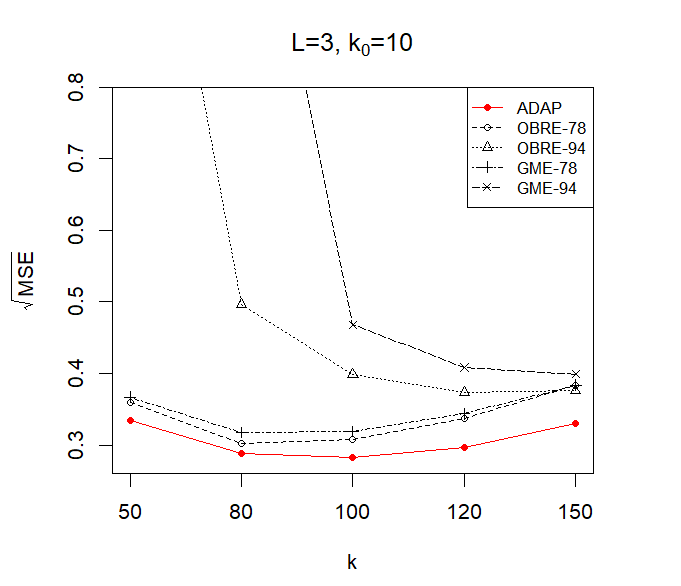}
	\hspace{-10mm}	\includegraphics[width=0.36\textwidth]{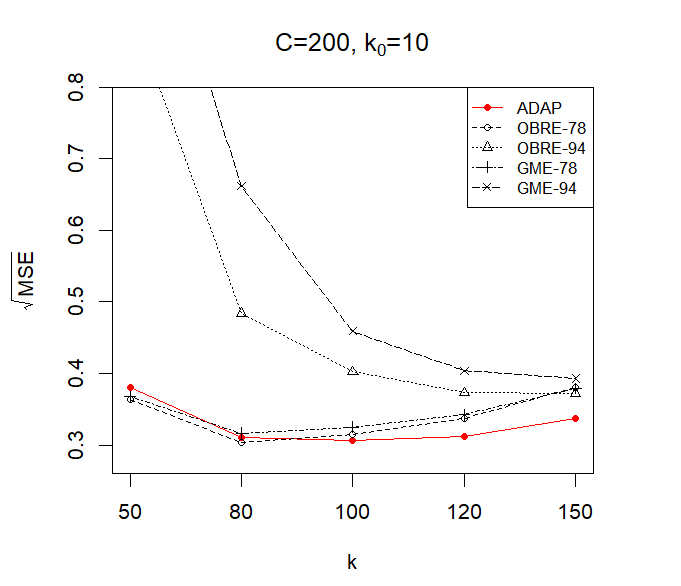}
	\hspace{-10mm}	\includegraphics[width=0.36\textwidth]{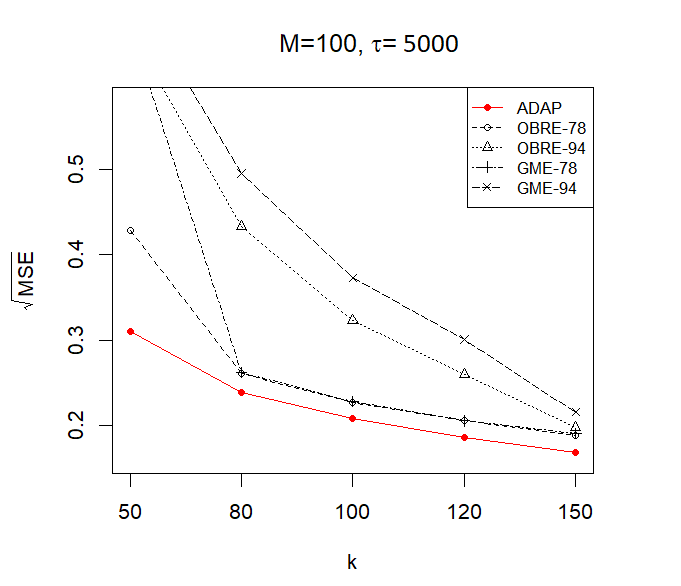}
	\caption{$\sqrt{MSE}$ of adaptive trimmed Hill for Burr(1,0.5,2) as a function of $k$. {\em Left}: Exponentiated Outliers. {\em Middle}: Scaled Outliers. {\em Right}: Mixed Outliers.}
	\label{fig:ha-bur}
\end{figure}

Figures \ref{fig:ha-t} and \ref{fig:ha-bur} display the performance of the adaptive trimmed Hill estimator (ADAP) for $|{\rm T}|(2)$ and Burr(1,0.5,2), distributions, respectively together with that of the optimal B-robust estimator (OBRE) and the generalized median estimator (GME).  Overall, the  ADAP is uniformly better than the OBRE and the GME. Exceptions include small values of $k$ for the scaled outliers. For $k<100$, the OBRE-94 and the GME-94 break down completely irrespective of the nature of the outliers and distribution under study. The OBRE-78 and the GME-78, though more robust than the OBRE-94 and the GME-94, cannot surpass the efficiency of the ADAP. Also for mixed outliers, even the OBRE-78 and the GME-78 break down for $k=50$. This is because the OBRE and the GME are immune to outliers only if their target ARE  value is less than the ratio $1-k_0/k$. This is another manifestation of the fact that the OBRE and the GME, unlike ADAP are not adaptive to the unknown levels of contamination in the extremes (see also Figures \ref{fig:ha-k0-exp}, \ref{fig:ha-k0-scl} and \ref{fig:ha-tau-mix} in Section \ref{sec:adap}).

\vspace{-3mm}
\begin{table}[H]
	\centering
	\begin{tabular}{|c|c|c|c|c|c|}
		\hline
		$k$ &  $k=50$ & $k=100$ & $k=200$ & $k=400$ & $k=600$\\\hline
		Exponentiated outliers & $10.04\pm 0.91$ & $ 10.01\pm 0.66$ & $ 10.02\pm 0.72$ & $ 10.02\pm 0.82$ & $ 10.02\pm  0.78$\\
		Scaled outliers &$9.74\pm 1.71$ & $  9.86\pm 1.44$ & $  9.91\pm 1.00$ & $  9.91\pm 0.95 $ & $ 9.87\pm  0.98$\\\hline
	\end{tabular}
	\caption{$\mathbb{E}(\widehat{k}_0)\pm {\rm Standard\:\:Error}(\widehat{k}_0)$ for $|$T$|(2)$ for $L=3$, $C=200$ and $k_0=10$.}
	\label{tab:ha-t}
\end{table}

\vspace{-3mm}

\begin{table}[H]
	\centering
	\begin{tabular}{|c|c|c|c|c|c|}
		\hline
		$k$ &  $k=50$ & $k=80$ & $k=80$ & $k=100$ & $k=120$\\\hline
		Exponentiated outliers &$10.01\pm0.71$ & $ 10.00\pm 0.67$ & $ 10.01\pm 0.88$ & $ 9.98\pm 0.43$ & $ 9.99\pm  0.49$\\
		Scaled outliers &$ 9.69\pm 1.72$ & $  9.76\pm 1.58$ & $  9.76\pm 1.36$ & $ 9.78\pm 1.46$ & $ 9.77\pm  1.29$\\\hline
	\end{tabular}
	\caption{$\mathbb{E}(\widehat{k}_0)\pm {\rm Standard\:\:Error}(\widehat{k}_0)$ for Burr(1,0.5,2) distribution for for $L=3$, $C=200$ and $k_0=10$.}
	\label{tab:ha-bur}
\end{table}

Due to their slow rate convergence to Pareto tails, both Burr and $|{\rm T}|$ are difficult cases to analyze. For the Burr distribution with $\rho=1$, the rate of convergence is further slower than that of the $|{\rm T}|$ with $\rho=2$. However, the ADAP performs well even in this challenging regime. This can be attributed to the accuracy of the weighted sequential testing algorithm which correctly identifies true number of outliers $k_0$ irrespective of the distribution under study for a wide range of $k$-values (see Tables \ref{tab:ha-t} and \ref{tab:ha-bur}).

\section{Application}
\label{sec:real}

In this section, we apply our weighted sequential testing algorithm and adaptive trimmed Hill estimator to real data. Two data sets have been explored in this context. The first one provides the calcium content in the Condroz region of Belgium \cite{goegebeur} (also analyzed in {\bf https://shrijita-apps.shinyapps.io/adaptive-trimmed-hill/})). The data is indeed heavy tailed and has already been explored in the works of \cite{beirqq} and \cite{VandewalleCa}. The second data set involves insurance claim settlements \cite{freclaim}. Both these data sets on analysis revealed the presence of outliers in the extremes and are therefore suitable for the application of our methodology.

\subsection{Condroz Data set}
\label{sec:condroz}

Figure \ref{fig:condroz} produces exploratory plots for the  {\em Condroz data set} of \cite{goegebeur}  which measures the calcium content of soil samples together with their pH levels in the Condroz  region of Belgium. As in \cite{VandewalleCa}, the conditional distribution of the calcium content for pH levels lying between 7-7.5 have been considered.  The left and middle panels use the value of $k=85$ based on the $k_{\rm opt}$ value from \cite{VandewalleCa}. The left panel displays a pareto quantile plot  \cite{beirqq} of the data where an apparent linear trend indicates Pareto distributed observations. Nearly six data points show up as outliers in the pareto quantile plot. This has already been observed in \cite{goegebeur}  but no principled methodology for the identification of such outliers has been proposed. Our trimmed Hill estimator  (recall Relation \eqref{e:xi-opt}) diagnostic plot in the middle panel also shows a change point in the values of the trimmed Hill statistics at $k_0=6$.  On applying the weighted sequential testing algorithm with type I error $q=0.05$, we formally identify exactly $k_0=6$ outliers for this data set\footnote{ The ties in the data are broken using a  suitable dithering technique like adding a small perturbation $\epsilon \sim  U(0,0.1)$ to the data or considering unique values in the data \label{dither}}.  This is in consensus with the findings of  \cite{goegebeur} and \cite{VandewalleCa}.

 The right panel in Figure \ref{fig:condroz} displays the values trimmed Hill estimator as a function of $k$ for $k_0=\widehat{k}_0=6$. Also displayed as a function of $k$ are the values  of the estimators, classic Hill and  biased Hill with $k_0=6$ (recall Relations \eqref{e:hill} and \eqref{e:xi-trimmed}). The robust estimator of $\xi$ as reported in the analysis of \cite{VandewalleCa} is same as that of the biased Hill. When compared with the trimmed Hill, the classic Hill plot produces much larger estimates and the biased Hill plot produces much smaller estimates of the tail index $\xi$.  This can be explained by the apparent upward trend in the outliers as shown in left and middle panels of Figure \ref{fig:condroz}. Thus, ignoring the presence of outliers by either using the classic Hill estimator or by naively truncating them and using the biased Hill statistics can lead to large discrepancies in the tail index values. The trimmed Hill estimator with $\widehat{k}_0=6$, which is in fact our adaptive trimmed Hill estimator discussed in Section \ref{sec:sim-setup}, produces more credible estimates of the tail index $\xi$. 

\subsection{French Claims Data Set}

Next, we consider a data set of claim settlements issued by a private insurer in France for the time period 1996-2006 from \cite{freclaim}. We investigate the payments of claim settlements  for the year 2006. Figure \ref{fig:freclaim} produces exploratory plots of this data where the left and middle panels use the value of $k=130$. The left panel displays a pareto quantile plot  \cite{beirqq} of the data where an apparent linear trend  indicates Pareto distributed observations as well as a large number of outliers. Nearly thirty three data points show up as outliers in the pareto quantile plot. This is further confirmed by the diagnostic plot in the middle panel where a change point in the values of trimmed Hill  statistics is evident at $k_0\approx 33$. On applying the weighted sequential testing with $q=0.05$, we identify $k_0\approx 33$ outliers for this data set\footref{dither}.

In contrast to the case of Condroz data set (Figure \ref{fig:condroz} right panel), now the both classic and biased Hill plots lie under the trimmed Hill plot (see the right panel of Figure \ref{fig:freclaim} constructed with $k_0=33$ and varying $k$). This can be explained by the apparent downward trend in the outliers as shown in left and middle panels of Figure \ref{fig:freclaim}. 

Observe that the trimmed Hill plot in Figure \ref{fig:freclaim} (right panel) has a rather high peak for $k$ close to $k_0$, but then it 
quickly stabilizes around the value of $2$, when $k$ grows. It is well-known that except in the ideal Pareto setting, the classic Hill plot 
can be quite volatile for small values of $k$ (see Figure 4.2 in \cite{resnick:2007}).  The same holds for the trimmed Hill plots, but ultimately, in 
Figure \ref{fig:freclaim} for a wide range of $k$'s the trimmed Hill plot is relatively stable and it provides more reliable estimates of $\xi$ 
than the classic and biased Hill plots therein.  This simple analysis shows that ignoring or not adequately treating extreme outliers can lead to significant underestimation of the tail index $\xi$. This in turn can result in severe underestimation of the tail of loss distribution with detrimental effects to the insurance industry.

\begin{figure}[H]
	\hspace{-2mm}	\includegraphics[width=0.37\textwidth]{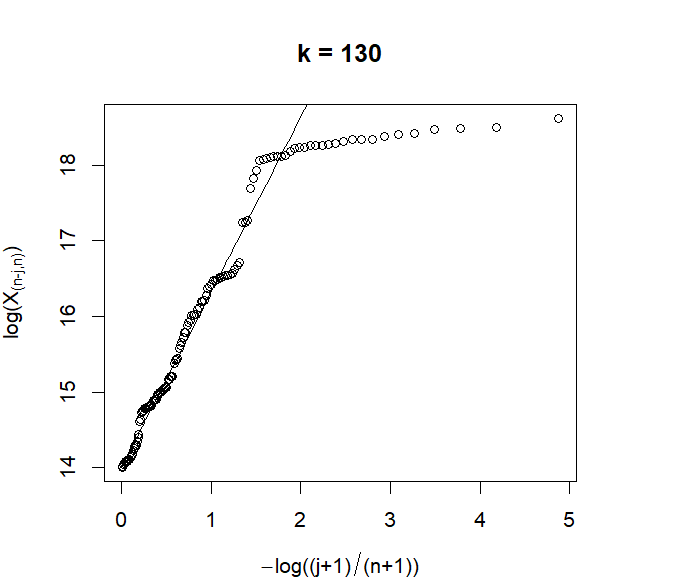}
	\hspace{-11mm}	\includegraphics[width=0.37\textwidth]{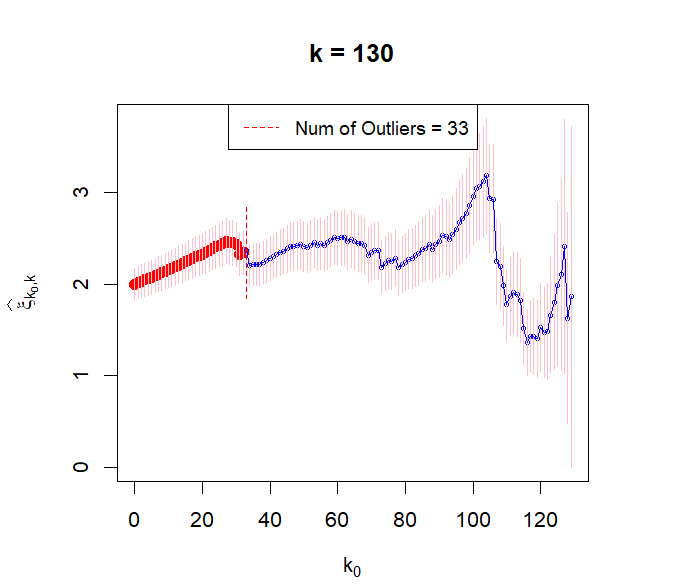}
	\hspace{-10mm}	\includegraphics[width=0.37\textwidth]{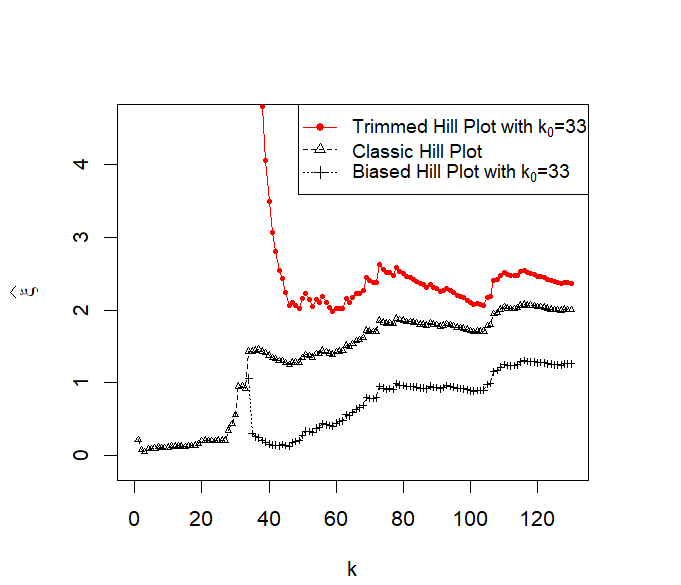}
	\caption{Exploratory plots of the French claim settlements. Left: Pareto quantile plot, Middle: Diagnostic Plot  and Right: Hill plots viz classic Hill plot, trimmed Hill plot and biased Hill plot.}
	\label{fig:freclaim}
\end{figure}

\section{Appendix}
\label{sec:append}

\subsection{Auxiliary Lemmas}

\begin{lemma} 
	\label{lem:u-ind}
	Let $E_j \stackrel{i.i.d}{\sim}{\rm Exp}(1)$, $j=1,2, \cdots, n+1$ be standard exponential random variables. Then, the ${\rm Gamma}(i,1)$ random variables defined as
	\begin{equation}
	\label{e:gam-dist}
	\Gamma_i = \sum_{j=1}^iE_j \hspace{5mm} i=1, \cdots, n+1,
	\end{equation} 
	satisfy 
	\begin{equation}
	\label{e:gam-ind}
	\Big( \frac{\Gamma_1}{\Gamma_{n+1}},\cdots, \frac{\Gamma_n}{\Gamma_{n+1}}\Big)\:\textmd{  and  }\:\Gamma_{n+1}\textmd{ are independent.}\\
	\end{equation}and
	\begin{equation}
	\label{e:gam-u}
	\Big( \frac{\Gamma_1}{\Gamma_{n+1}},\cdots, \frac{\Gamma_n}{\Gamma_{n+1}}\Big)\stackrel{d}{=}( U_{(1,n)},\cdots,U_{(n,n)}) 
	\end{equation}
	
	where $U_{(1,n)}< \cdots< U_{(n,n)}$ are the order statistics of $n$ i.i.d. U(0,1) random variables.
\end{lemma}

For details on the proof see Example 4.6 on page 44 in \cite{MR3025012}. The next result, quoted from page 37 in \cite{Haan}, shall be used throughout the course of the paper to switch 
between order statistics of exponentials and i.i.d.\ exponential random variables.  

\begin{lemma}[R\'enyi, 1953]
	\label{lem:renyi}
	Let $E_1, E_2, \cdots, E_n$ be a sample of $n$ i.i.d. exponential random variables with mean $\xi$ (denoted by {\rm Exp}($\xi$)) and $E_{(1,n)} \leq E_{(2,n)} \leq E_{(n,n)} $ be the order statistics. By R\'enyi's (1953) representation, we have for fixed $k\leq n$,
	\begin{equation}
	\label{e:renyi}
	(E_{(1,n)},\cdots, E_{(i,n)},\cdots,E_{(k,n)})\stackrel{d}{=}\Big(\frac{E_1^*}{n},\cdots,\sum_{j=1}^i \frac{E^*_j}{n-j+1}, \cdots, \sum_{j=1}^k \frac{E^*_j}{n-j+1}\Big) 
	\end{equation}
	where $E_1^*, \cdots, E_k^*$ are also i.i.d. Exp($\xi$).
\end{lemma}

\begin{lemma}
	\label{lem:ratio-conv}
	For $\Gamma_m=E_1+E_2+\cdots+E_m$ where the $E_i'$s are i.i.d. standard exponential random variables, for any $\rho$ 
	\begin{eqnarray}
	\label{e:gam-conv-0}
	\sup_{m\geq M}\Big|\Big(\frac{\Gamma_{m}}{m}\Big)^{-\rho}-1\Big| &\stackrel{a.s.}{\longrightarrow}& 0, \hspace{2mm}M \rightarrow \infty\\
	\label{e:gam-conv-1}
	\sup_{m,n\geq M}\Big|\Big(\frac{\Gamma_{m}/m}{\Gamma_{n}/n}\Big)^{-\rho}-1\Big| &\stackrel{a.s.}{\longrightarrow}& 0, \hspace{2mm}M \rightarrow \infty
	\end{eqnarray}
\end{lemma}

\begin{lemma}
	For all $\rho>0$, we have 
	\label{lem:gam-int}
	$$
	\sup_{m\geq M}\Big|\frac{1}{m}\sum_{i=1}^{m}{\Big(\frac{\Gamma_{i+1}}{\Gamma_{m+1}}\Big)}^{\rho}-\frac{1}{1-\rho}\Big|\stackrel{a.s.}{\longrightarrow}0, \hspace{2mm} M \rightarrow \infty
	$$
\end{lemma}
\begin{proof} It is equivalent to show that, as $m\to\infty$,
	\label{lem:sum-conv}
	\begin{equation}
	\label{e:sum-rho-conv}
	\Big|\frac{1}{m}\sum_{i=1}^{m}{\Big(\frac{\Gamma_{i+1}}{\Gamma_{m+1}}\Big)}^{\rho}-\frac{1}{1-\rho}\Big|\stackrel{a.s.}{\longrightarrow}0.
	\end{equation}
	For a fixed $\omega \in \Omega$, let us define the following sequence of functions
	$$f_m(x)=\sum_{i=1}^m {(\Gamma_{i+1}/\Gamma_{m+1})}^{\rho}(\omega){\mathbf{1}}_{(\frac{i-1}{m},\frac{i}{m}]}(x), \hspace{5mm}x>0$$
	Suppose $x \in ((i-1)/m,i/m]$, then 
	\begin{equation}
	\label{e:dct}
	f_m(x)={(\Gamma_{[mx]+1}/\Gamma_{m+1})}^{-\rho}(\omega)=\Big(\frac{[mx]+1}{m}\Big)^{-\rho}\Big(\frac{\Gamma_{[mx]+1}/([mx]+1)}{\Gamma_{m}/m}\Big)^{\rho}(\omega) \rightarrow x^{-\rho}
	\end{equation}
	where the convergence follows from \eqref{e:gam-conv-1}. Moreover since $\Gamma_{[mx]+1}<\Gamma_m$ and $\rho<0$, therefore $|f_m(x)|\leq 1$, for all $x>0$. Thus by dominated convergence theorem,
	\begin{equation}
	\label{e:dct-1}
	\int_{0}^{1}f_m(x)dx=\frac{1}{m}\sum_{i=1}^m {(\Gamma_{i+1}/\Gamma_{m+1})}^{-\rho}(\omega) \rightarrow \int_0^1 x^{-\rho}dx=\frac{1}{1-\rho}
	\end{equation}
	Since \eqref{e:dct} holds for all $\omega \in \Omega$ with $\mathbb{P}(\Omega]=1$, so does \eqref{e:dct-1}. This completes the proof.
\end{proof}

\subsection{Proofs for Section \ref{sec:trim-hill}}\label{sec:proofs-sec2}

\begin{proof}[ {\bf Proof of Proposition \ref{prop:xi-opt}}]
	Note that, if $X_i \sim {\rm Pareto}(\sigma,\xi)$, then it can be alternatively written as
	$$X_i=\sigma U_{i}^{-\xi}, \hspace{3mm}i=1, \cdots, n,$$
	where $U_i$'s are i.i.d. $U(0,1)$. Therefore by Relation \eqref{e:gam-u} , we have
	\begin{equation}
	\label{e:Pareto-order}
	(X_{(n,n)}, \cdots, X_{(1,n)})=\sigma(U_{(1,n)}^{-\xi},\cdots,U_{(n,n)}^{-\xi})\stackrel{d}{=}\sigma\Bigg( {\Big(\frac{\Gamma_1}{\Gamma_{n+1}}\Big)}^{-\xi},\cdots, {\Big(\frac{\Gamma_n}{\Gamma_{n+1}}\Big)}^{-\xi}\Bigg)
	\end{equation}
	where $X_{(n,n)}>\cdots>X_{(1,n)}$ are the order statistics for the $X_i$'s. Hence, for all $1\leq k \leq n-1$, we have
	\begin{eqnarray}
	\label{e:log-u}
	\Bigg(\log \Big(\frac{X_{(n,n)}}{X_{(n-k,n)}} \Big),\cdots,\log \Big(\frac{X_{(n-k+1,n)}}{X_{(n-k,n)}} \Big)\Bigg)  &\stackrel{d}{=}& -\xi\Bigg(\log \Big(\frac{\Gamma_1}{\Gamma_{k+1}} \Big),\cdots,\log \Big(\frac{\Gamma_k}{\Gamma_{k+1}} \Big)\Bigg)\\\nonumber
	&\stackrel{d}{=}&-\xi(\log U_{(1,k)}, \cdots, \log U_{(k,k)}),
	\end{eqnarray}
	where the $U_{(i,k)}$'s are the order statistics for a sample of $k$ i.i.d. $U(0,1)$ and the last equality in \eqref{e:log-u} follows from Relation \eqref{e:gam-u}. Since negative log transforms of $U(0,1)$ are standard exponentials, one can define $E_{(i,k)}$, $i=1,\cdots,k$ as
	\begin{equation}
	\label{e:log-exp}
	\Bigg(\log \Big(\frac{X_{(n,n)}}{X_{(n-k,n)}} \Big),\cdots,\log \Big(\frac{X_{(n-k+1,n)}}{X_{(n-k,n)}} \Big)\Bigg) =: (E_{(k,k)}, \cdots,  E_{(1,k)})
	\end{equation}
	such that the $E_{(i,k)}$'s are the  order statistics of $k$ i.i.d. exponentials with mean $\xi$.  
	
	Using \eqref{e:log-exp},  $\widehat{\xi}^{\rm trim}_{k_0,k}$ in \eqref{e:xi-trim-gen} is simplified as:
	\begin{equation}
	\label{e:blue-1}
	\widehat{\xi}^{\rm trim}_{k_0,k} =  \sum_{i=k_0+1}^k c_{k_0,k}(i)E_{(k-i+1,k)}=\sum_{i=1}^{k-k_0}\delta_i E_{(i,k)}
	\end{equation}
	where $\delta_i=c_{k_0,k}(k-i+1)$. The optimal choice of weights $\delta_i$'s  which produces the best linear unbiased estimator (BLUE) is obtained using Lemma \ref{lem:blue} as follows:
	\begin{equation}
	\label{e:blue-2}
	\delta^{\rm opt}_i=
	\begin{cases}
	\frac{1}{k-k_0}\hspace{5mm} i=1,\cdots, k-k_0-1\\
	\frac{k_0+1}{k-k_0}\hspace{5mm} i=k-k_0
	\end{cases}
	\end{equation}
	Rewriting $E_{(i,k)}$'s in terms of $X_{(n-i+1,n)}$'s as in \eqref{e:log-exp} completes the proof.
\end{proof}

\begin{lemma}
	\label{lem:blue} 
	If $E_i$, $i=1,\cdots,n$ are i.i.d. observations from ${\rm Exp}(\xi)$, the best linear unbiased estimator (BLUE) of $\xi$ based on the order statistics, $E_{(1,n)}<\cdots<E_{(r,n)}$ is given by
	\begin{equation*} 
	\widehat{\xi}=\frac{1}{r}\sum_{i=1}^{r-1} E_{(i,n)} +\frac{n-r+1}{r}E_{(r,n)}\end{equation*}
\end{lemma}

\begin{proof}
	Let $\widehat{\xi}=\sum_{i=1}^{r} \gamma_i E_{(i,n)}$ denote the BLUE of $\xi$. By Relation \eqref{e:renyi}, the BLUE can then be expressed as
	\begin{equation}
	\label{e:z-y-1}
	\hat{\xi}=\sum_{i=1}^{r} \gamma_i\sum_{j=1}^i \frac{E^*_j}{(n-j+1)}=\sum_{j=1}^{r} E^*_j \sum_{i=j}^{r} \frac{\gamma_i}{(n-j+1)}=:\sum_{j=1}^{r} E^*_j \delta_j
	\end{equation} 
	where the $E_j^*$ are i.i.d. from ${\rm Exp}(\xi)$ and $\delta_j=(n-j+1)\sum_{i=j}^r \gamma_i$ 
	
	For i.i.d. observations from ${\rm Exp}(\xi)$, the sample mean is the uniformly minimum variance unbiased estimator (UMVUE) for $\xi$ (see Lehmann Scheffe Theorem, Theorem 1.11, page 88 in~\cite{lehmann:casella}).
	
	Thus, $\delta_j=1/r$ yields the required best linear unbiased estimator and therefore, the weights $\gamma_i$'s have the form: $$\gamma_{i}=\begin{cases}
	\frac{n-{r}+1}{r} & \text{ }  i=r\\                                                                                                                                                                                                       \frac{1}{{r}}& \text{ } i<r
	\end{cases}$$
	This completes the proof.
\end{proof}

\begin{proof}[{\bf Proof of Theorem \ref{thm:umvue-p}}]
	Assume that $\sigma$ is known and consider the class of statistics:
	$$	{\cal{U}}^{\sigma}_{k_0}=\left\{T=T(X_{(n-k_0,n)},\cdots,X_{(1,n)}):\: \mathbb{E}(T)=\xi,\: X_1, \cdots, X_n \stackrel{i.i.d.}{\sim} {\rm Pareto}(\sigma,{\xi})\right\}.$$
	Since $\sigma$ is no longer a parameter, every statistic in ${\cal{U}}_{k_0}^{\sigma}$ can be equivalently written as a function of $\log(X_{(n-i+1,n)}/\sigma)$, $i=k_0+1, \cdots, n$ as follows:
	\begin{equation*}
	\label{e:u-sigma-0}
	{\cal{U}}_{k_0}^{\sigma}=\left\{S=S\left(\log\Big(\frac{X_{(n-k_0,n)}}{\sigma}\Big),\cdots,\log\Big(\frac{X_{(1,n)}}{\sigma}\Big)\right): \: \mathbb{E} (S)=\xi, \: X_1, \cdots, X_n \stackrel{i.i.d.}{\sim} {\rm Pareto}(\sigma,{\xi})\right\}.
	\end{equation*}
	Since $X_i$'s follow ${\rm Pareto(\sigma,\xi)}$,  $\log(X_i/\sigma) \sim {\rm Exp}(\xi)$ and therefore 
	\begin{equation*}
	\left(\log\Big(\frac{X_{(n-k_0,n)}}{\sigma}\Big),\cdots,\log\Big(\frac{X_{(1,n)}}{\sigma}\Big) \right) \stackrel{d}{=}\left(E_{(n-k_0,n)}, \cdots,E_{(1,n)}\right),
	\end{equation*} 
	where $E_{(1,n)} \leq \cdots \leq E_{(n,n)}$ are the order statistics of $n$ i.i.d. observations from ${\rm Exp}(\xi)$. Therefore
	\begin{equation}
	\label{e:u-sigma-1}
	{\cal{U}}_{k_0}^{\sigma}\stackrel{d}{=}\left\{S=S(E_{(n-k_0,n)}, \cdots,E_{(1,n)}):\: \mathbb{E}(S)=\xi, \:E_1, \cdots, E_n \stackrel{i.i.d.}{\sim} {\rm Exp}(\xi)\right\},
	\end{equation}
	where the $E_i$'s do not depend on $\sigma$. Next, using Relation \eqref{e:renyi} of Lemma \ref{lem:renyi}, we have
	\begin{eqnarray*}
		S(E_{(n-k_0,n)}, \cdots,E_{(1,n)})&=&S\Big(\sum_{j=1}^{n-k_0} \frac{E^*_j}{n-j+1}, \cdots, \sum_{j=1}^{n-k} \frac{E^*_j}{n-j+1}\Big)=R(E^*_1, \cdots,E^*_{n-k_0})	
	\end{eqnarray*}
	Using the above result with \eqref{e:u-sigma-1}, we get
	\begin{equation}
	\label{e:umvue-1}
	{\cal{U}}_{k_0}^{\sigma}
	\stackrel{d}{=}{\cal{V}}_{k_0}:=\left\{R=R(E^*_1, \cdots,E^*_{n-k_0}): \: \mathbb{E}(R)=\xi,\: E^*_{1}, \cdots, E^*_{n-k_0} \stackrel{i.i.d.}{\sim} {\rm Exp}(\xi)\right\}.
	\end{equation}		
	where the first equality is in the sense of finite dimensional distributions. 
	
	By \eqref{e:umvue-1}, we have $\inf_{T \in {\cal{U}}_{k_0}^{\sigma}} {\rm Var}(T)=\inf_{R \in {\cal{V}}_{k_0}} {\rm Var}(R):=L$. Since the sample mean, ${\overline{E}}^*_{n-k_0}=\sum_{i=1}^{n-k_0} E^*_i/(n-k_0)$ is uniformly the minimum variance estimator (UMVUE) of $\xi$ among the class described by ${\cal{V}}_{k_0}$,  $L$ can be easily obtained as
	\begin{equation}L={\rm Var}({\overline{E}}^*_{n-k_0})=\frac{\xi^2}{n-k_0}\end{equation}
	The fact that  ${\overline{E}}^*_{n-k_0}$ is the UMVUE follows because it is an  unbiased and complete sufficient statistic for $\xi$ (see Lehmann Scheffe Theorem, Theorem 1.11, page 88 in~\cite{lehmann:casella}).
	
	To complete the proof, observe that every statistic $T$ in ${\cal{U}}_{k_0}$ is an unbiased estimator of $\xi$ for any arbitrary choice of $\sigma$. This implies that for any $\sigma$, $T\in {\cal{U}}^\sigma_{k_0}$ and therefore $L \leq {\rm Var}(T)$. 
	Since this holds for all values of $T \in  {\cal{U}}_{k_0}$,
	the proof of the lower bound in  \eqref{e:opt-var} follows.
	
	For the upper bound in \eqref{e:opt-var}, we observe that $\widehat{\xi}_{k_0,n-1} \in {\cal{U}}_{k_0}$, which in view of Proposition \ref{prop:xi-exp} implies $$\inf_{T \in {\cal{U}}_{k_0}} {\rm{Var}}(T) \leq {\rm Var}(\widehat{\xi}_{k_0,n-1})= \frac{\xi^2}{n-k_0-1}.$$
	This completes the proof.
\end{proof}

\begin{proof}[{\bf Proof of Proposition \ref{prop:xi-exp}}]
	From Relations \eqref{e:blue-1} and \eqref{e:blue-2},  we have
	\begin{equation}
	\label{e:xi-jt}
	{\Big\{\widehat{\xi}_{k_0,k},\ k_0=0,\ldots,k-1 \Big\}}
	= {\Big\{\frac{1}{k-k_0}\sum_{i=1}^{k-k_0-1}E_{(i,k)}+\frac{k_0+1}{k-k_0}E_{(k-k_0,k)},\ k_0=0,\ldots,k-1 \Big\}}
	\end{equation}
	Using Relation \eqref{e:renyi}, for all $k_0=0, 1, \cdots, k-1$, we have
	\begin{equation}
	\label{e:xi-simp}
	\widehat{\xi}_{k_0,k}=\frac{1}{k-k_0}\sum_{i=1}^{k-k_0-1}\sum_{j=1}^{i}\frac{E^*_j}{(k-j+1)}+\frac{k_0+1}{k-k_0}\sum_{j=1}^{k-k_0}\frac{E^*_j}{(k-j+1)}\\
	\end{equation}
	Interchanging the order of summation in the first term in the right hand side of \eqref{e:xi-simp},  we obtain
	\begin{eqnarray*}
		\widehat{\xi}_{k_0,k}
		&=&\sum_{j=1}^{k-k_0-1}\frac{E^*_j}{k-j+1}\sum_{i=j}^{k-k_0-1}\frac{1}{k-k_0}+\frac{k_0+1}{k-k_0}\sum_{j=1}^{k-k_0}\frac{E^*_j}{(k-j+1)}\\
		&=&\sum_{j=1}^{k-k_0-1}\frac{E^*_j}{k-j+1}\left(\sum_{i=j}^{k-k_0-1}\frac{1}{k-k_0}+\frac{k_0+1}{k-k_0}\right)+\frac{E^*_{k-k_0}}{k-k_0}\\
		&=&\sum_{j=1}^{k-k_0-1}\frac{E^*_j}{k-j+1}\frac{(k-j+1)}{k-k_0}+\frac{E^*_{k-k_0}}{k-k_0}\\
		&=&\frac{1}{k-k_0}\sum_{j=1}^{k-k_0}E^*_j, \hspace{5mm} 
	\end{eqnarray*}
	Since $E_j^*$, $j=1, \cdots, k-k_0$ follow Exp($\xi$), $E_j^*$ are indeed $\xi$ times i.i.d. standard exponentials. This completes the proof of Relation \eqref{e:xi-jt-big}.
	
	The proof of Relation \eqref{e:xi-normal} is a direct application of central limit theorem to Relation \eqref{e:xi-jt-big}.
\end{proof}

\subsection{Proofs for Section \ref{sec:general-heavy}} \label{sec:proofs-sec3}

\subsubsection{Minimax Rate Optimality}\label{sec:proof-sec-minmax}
Our goal is to establish the uniform consistency in Relation \eqref{e:t:uniform-consistency}. To this end, recall the representation in Relation \eqref{e:tail-3}. For the Hall class of distributions in Relation \eqref{e:D-new-def}, it can be shown that $\sqrt{k}R_{k_0,k}$ is $O_{\mathbb{P}}(1)$ (see Lemma \ref{lem:Hall} below). With $\sqrt{k}|R_{k_0,k}|$ bounded away from infinity, it is easier to bound the quantity $\sqrt{k}|\widehat{\xi}_{k_0,k}-\xi|$ since by Relation \eqref{e:tail-3} 
\begin{equation}\label{e:xi-xi-star-r}
\sqrt{k}|\widehat{\xi}_{k_0,k}-\xi|\leq \sqrt{k}|R_{k_0,k}|+\sqrt{k}|\widehat{\xi^*}_{k_0,k}-\xi|.\end{equation} This shall form the basis of the proof for Theorem \ref{t:uniform-consistency} as shown next.

\begin{proof}[{\bf Proof of Theorem \ref{t:uniform-consistency}}] Let $P_n=\inf_{F \in {\cal{D}}_{\xi}(B,\rho)}\mathbb{P}_F\Big(\max_{0\leq k_0<h(k)}|\widehat{\xi}_{k_0,k}-\xi|\leq a(n)\Big)$. By Relation \eqref{e:xi-xi-star-r}, we have
	$$P_n=\inf_{{\cal{D}}_{\xi}(B,\rho)}\mathbb{P}_F\Big(\underbrace{\max_{0\leq k_0<h(k)}\sqrt{k}|R_{k_0,k}|\leq (\sqrt{k}a(n))/2}_{A_{1n}}\cap\underbrace{\max_{0\leq k_0<h(k)}\sqrt{k}|\widehat{\xi}^*_{k_0,k}-\xi|\leq (\sqrt{k}a(n))/2}_{A_{2n}}\Big).$$
	Since $\sqrt{k}a(n) \rightarrow \infty$, by Lemma \ref{lem:Hall}, $\inf_{F \in {\cal{D}}_{\xi}(B,\rho)}\mathbb{P}_F(A_{1n}) \rightarrow 1$. We also have that, $$\inf_{F \in {\cal{D}}_{\xi}(B,\rho)}\mathbb{P}_F(A_{2n})=\mathbb{P}\Big(\max_{0\leq k_0<h(k)} \sqrt{k}| \widehat {\xi}_{k_0,k}^* - \xi|\leq(\sqrt{k}a(n))/2\Big)$$ since $\widehat{\xi}^*_{k_0,k}$ does not depend on $F \in {\cal{D}}_{\xi}(B,\rho)$.

	By using Donsker's principle, we will show that 
	$$
	\max_{0\le k_0 <h(k)} | \widehat \xi_{k_0,k}^* - \xi|  =o_{\mathbb{P}}(a(n)),
	$$
	which will imply $\mathbb{P}_F(A_{2n}) \rightarrow 1$.  Indeed, without loss of generality, suppose $\xi=1$ and let 
	$E_i,\ i=1,2,\dots$ be independent standard exponential random variables. For every $\epsilon\in (0,1)$, we have that
	\begin{equation}\label{e:Wkprocess}
	W_k =\{W_k(t),\ t\in [\epsilon,1]\} :=  \left\{\frac{\sqrt{k}}{[kt]} \sum_{i=1}^{[kt]} (E_i -1),\ t \in [0,1] \right\} \stackrel{d}{\to } \{B(t)/t, \ t\in [\epsilon,1]\},
	\end{equation}
	as $k\to\infty$, where $B = \{B(t),\ t\in [0,1]\}$ is the standard Brownian motion, and where the last convergence is in 
	the space of cadlag functions ${\mathbb D}[\epsilon,1]$ equipped with the Skorokhod $J_1$-topology.  (In fact, since the limit
	has continuous paths, the convergence is also valid in the uniform norm.)
	
	Recall that by Relation \eqref{e:xi-jt-big}, we have 
	$$
	\{\widehat\xi_{k_0,k}^*(n),\ 0\le k_0 < k\}\stackrel{d}{=} \left\{\sum_{i=1}^{k-k_0} E_i/(k-k_0),\ 0\le k_0 < k\right\}.
	$$
	Thus, 
	\begin{equation}\label{e:t:uniform-consistency-1}
	\sqrt{k} \max_{0\le k_0<h(k)} | \widehat \xi_{k_0,k}^*(n) - \xi|  \stackrel{d}{= } \sup_{t \in [1-h(k)/k, 1]} |W_k(t)| \le
	\sup_{t \in [\epsilon,1]} |W_k(t)|, 
	\end{equation}
	where the last inequality holds for all sufficiently large $k$, since $1-h(k)/k\to1$, as $k\to\infty$.  
	Since the supremum is a continuous functional in $J_1$, the convergence in Relation \eqref{e:Wkprocess} implies that the 
	right--hand side of Relation \eqref{e:t:uniform-consistency-1} converges in distribution
	to $\sup_{t\in [\epsilon,1]} |B(t)/t| = O_\mathbb{P}(1)$, which is finite with probability one. This, since $a(n)\sqrt{k(n)} \to \infty$, completes the proof.
\end{proof}

\begin{lemma}
	\label{lem:Hall}
	Assumption \eqref{e:D-new-def} implies there exist $M>0$ such that
	\begin{equation}
	\label{e:rem-hall}
	\inf_{F \in {\cal{D}}_{\xi}(B,\rho)}\mathbb{P}_F\Big(\max_{0\leq k_0<h(k)}\sqrt{k}|R_{k_0,k}|\leq M\Big)\rightarrow 1 \textmd{ as }h(k) \rightarrow \infty
	\end{equation}
	where $R_{k_0,k}$ is defined as in Relation \eqref{e:tail-3} and $k=O(n^{2\rho/(1+2\rho)})$.
\end{lemma}

\begin{proof}
	By Relation \eqref{e:D-new-def}, we have $1-Bx^{-\rho}\leq L(x)\leq 1+Bx^{-\rho}$. Therefore,
	\begin{eqnarray}
	\label{e:r-ub}
	(k-k_0)R_{k_0,k}&\leq& (k_0+1) \log\frac{1+BY^{-\rho}_{(n-k_0,n)}}{1-BY^{-\rho}_{(n-k,n)}} +\sum_{i=k_0+2}^{k} \log\frac{1+BY^{-\rho}_{(n-i+1,n)}}{1-BY^{-\rho}_{(n-k,n)}} \\\nonumber
	&\leq& k\log\frac{1+BY^{-\rho}_{(n-k,n)}}{1-BY^{-\rho}_{(n-k,n)}},
	\end{eqnarray}
	since $Y^{-\rho}_{(n-k,n)}\geq Y^{-\rho}_{(n-i+1,n)}$ for $i=k_0+1, \cdots,k$.  Similarly, we also have \begin{equation}
	\label{e:r-lb}
	(k-k_0)R_{k_0,k}\geq k\log\frac{1-BY^{-\rho}_{(n-k,n)}}{1+BY^{-\rho}_{(n-k,n)}}=-k\log\frac{1+BY^{-\rho}_{(n-k,n)}}{1-BY^{-\rho}_{(n-k,n)}}.
	\end{equation}Thus, Relations \eqref{e:r-ub} and \eqref{e:r-lb} together imply
	\begin{equation}
	\label{e:r-mod-ub}
	\max_{0\leq k_0<h(k)}\sqrt{k}|R_{k_0,k}|\leq \frac{\sqrt{k}Y^{-\rho}_{(n-k,n)}}{1-h(k)/k}\max_{0\leq k_0<h(k)}\frac{1}{Y^{-\rho}_{(n-k,n)}}\log\frac{1+BY^{-\rho}_{(n-k,n)}}{1-BY^{-\rho}_{(n-k,n)}}\end{equation}
	Since $h(k)=o(k)$, $1-h(k)/k \to 1$ and 
	\begin{eqnarray*}
		\sqrt{k}Y^{-\rho}_{(n-k,n)}\frac{1}{Y^{-\rho}_{(n-k,n)}}\log\frac{1+BY^{-\rho}_{(n-k,n)}}{1-BY^{-\rho}_{(n-k,n)}}&
		\stackrel{d}{=}&\underbrace{\sqrt{k}(\Gamma_{k+1}/\Gamma_{n+1})^{\rho}}_{\Delta_{1k}}\underbrace{\frac{1}{(\Gamma_{k+1}/\Gamma_{n+1})^{\rho}}\log\frac{ 1+B(\Gamma_{k+1}/\Gamma_{n+1})^{\rho}}{1-B(\Gamma_{k+1}/\Gamma_{n+1})^{\rho}}}_{\Delta_{2k}}
	\end{eqnarray*}
	By Relation \eqref{e:gam-conv-0} in Lemma \ref{lem:ratio-conv}, we have  $\Gamma_{k+1}/\Gamma_{n+1} \stackrel{a.s.}{\sim} (k/n)^{\rho}$.  Therefore, for  $k=O(n^{2\rho/(1+2\rho)})$, $\Delta_{1k}$ is $O_\mathbb{P}(1)$.	Since $k/n \to 0$, therefore $(\Gamma_{k+1}/\Gamma_{n+1})^\rho  \stackrel{a.s.}{\longrightarrow} 0$ which implies $\Delta_{2k}\stackrel{a.s.}{\longrightarrow}2B$. 
	
	Thus, there exist $M$  such that 	$$\inf_{F \in {\cal{D}}_{\xi}(B,\rho)}\mathbb{P}_F\Big(\max_{0\leq k_0<k}\frac{k-k_0}{kY^{-\rho}_{(n-k,n)}}|R_{k_0,k}|\leq M\Big)\geq \mathbb{P}(\Delta_{1k}\Delta_{2k}\leq M) \rightarrow 1$$
	This completes the proof.
	
\end{proof}

\subsubsection{Asymptotic Normality}
\label{sec:proof-sec-normal}

\begin{proof}[{\bf Proof of Theorem \ref{prop:E-conv}}] To prove Relation \eqref{e:E-conv}, we observe that
	\begin{equation}
	\label{e:rs-decomp}
	k^\delta\Big|R_{k_0,k}-\frac{k^{-\delta}cA}{(1+\rho)}\Big|\leq k^{\delta}|R_{k_0,k}-S_{k_0,k}|+k^\delta\Big|S_{k_0,k}-\frac{k^{-\delta}cA}{ ({1+\rho})}\Big|
	\end{equation}
	with $S_{k_0,k}$ defined as
	\begin{equation}
	\label{e:s-def}
	S_{k_0,k}:=\frac{cg(Y_{(n-k,n)})}{k-k_0}\Big((k_0+1) \int_{1}^{Y_{(n-k_0,n)}/Y_{(n-k,n)}} \nu^{-\rho-1} d\nu +\sum_{i=k_0+2}^{k} \int_{1}^{Y_{(n-i+1,n)}/Y_{(n-k,n)}} \nu^{-\rho-1} d\nu\Big),
	\end{equation}
	where $Y_i$'s are  i.i.d observations from Pareto(1,1) as in \eqref{e:tail-3}.
	
	We will show that the right hand side of \eqref{e:rs-decomp} vanishes as $k \rightarrow \infty$. To this end, we first show that $k^\delta \max_{0\leq k_0<h(k)}|R_{k_0,k}-S_{k_0,k}|\stackrel{\mathbb{P}}{\longrightarrow}0$ as follows:
	\begin{eqnarray}
	\label{e:rk-sk}
	k^\delta \max_{0\leq k_0<h(k)}|R_{k_0,k}-S_{k_0,k}|&=&k^\delta \max_{0\leq k_0<h(k)}\frac{kg(Y_{n-k,n})}{k-k_0}\Big(\frac{k-k_0}{kg(Y_{(n-k,n)})}|R_{k_0,k}-S_{k_0,k}|\Big)\\\nonumber
	&\leq&\frac{k^\delta g(Y_{(n-k,n)})}{1-h(k)/k}\underbrace{\max_{0\leq k_0<h(k)}\Big(\frac{k-k_0}{kg(Y_{(n-k,n)})}|R_{k_0,k}-S_{k_0,k}|\Big)}_{\Delta_{2k}}
	\end{eqnarray}
	where $1-h(k)/k\rightarrow 1$ since $h(k)=o(k)$. Additionally,   
	$\Delta_{2k} \stackrel{\mathbb{P}}{\longrightarrow}0$ by Lemma \ref{lem:s-def} and  \begin{equation}
	\label{e:del-g}
	k^\delta g(Y_{(n-k,n)})\stackrel{ \mathbb{P}}{\longrightarrow}A
	\end{equation} follows from Relation \eqref{e:g-asy} and assumption \eqref{e:A-def}. Thus, the  bound  in \eqref{e:rk-sk} goes to 0 as $k\to \infty$. 
	
	Next we show that the second term in the right hand side of \eqref{e:rs-decomp} also vanishes. Indeed,
	\begin{eqnarray*}
		k^\delta \max_{0\leq k_0<h(k)}\Big|S_{k_0,k}-\frac{k^{-\delta}cA}{(1+\rho)}\Big|&=&k^\delta \max_{0\leq k_0<h(k)}\frac{kg(Y_{n-k,n})}{k-k_0}\Big|\frac{k-k_0}{kg(Y_{(n-k,n)})}S_{k_0,k}-\frac{cA(k-k_0)}{(1+\rho)k^\delta g(Y_{(n-k,n)})}\Big|\\\nonumber
		&\leq&\frac{k^\delta g(Y_{(n-k,n)})}{1-h(k)/k}\underbrace{\max_{0\leq k_0<h(k)}\Big|\frac{k-k_0}{kg(Y_{(n-k,n)})}S_{k_0,k}-\frac{cA(k-k_0)}{k(1+\rho)k^\delta g(Y_{(n-k,n)})}\Big|}_{\Delta_{3k}}
	\end{eqnarray*}
	where $k^\delta g(Y_{(n-k,n)})\stackrel{\mathbb{P}}{\longrightarrow} A$ as in \eqref{e:del-g} and $1-h(k)/k \rightarrow 1$.
	
	We next show that $\Delta_{3k} \stackrel{\mathbb{P}}{\longrightarrow}0$ as follows: 
	\begin{eqnarray*}
		\Delta_{3k}&\leq&\underbrace{\max_{0\leq k_0<h(k)}\Big|\frac{k-k_0}{kg(Y_{(n-k,n)})}S_{k_0,k}+c{\Big(\frac{k_0}{k}\Big)}^{1+\rho}-\frac{c}{1+\rho}\Big|}_{\Delta_{4k}}\\&+&\underbrace{\max_{0\leq k_0<h(k)}\Big|\frac{c}{1+\rho}-c{\Big(\frac{k_0}{k}\Big)}^{1+\rho}-\frac{cA(k-k_0)}{k(1+\rho)k^\delta g(Y_{(n-k,n)})}\Big|}_{\Delta_{5k}}
	\end{eqnarray*}
	\normalsize where
	$\Delta_{4k}\stackrel{\mathbb{P}}{\longrightarrow}0$ by Lemma \ref{lem:r-def}. Since $\max_{0\leq k_0<k}(k_0/k)^{1+\rho}\leq (h(k)/k)^{1+\rho}\rightarrow 0$, thus to prove $\Delta_{5k} \stackrel{\mathbb{P}}{\longrightarrow}0$, it suffices to show that
	$$\max_{0\leq k_0<h(k)}\Big|\frac{c}{1+\rho}-\frac{cA(k-k_0)}{k(1+\rho)k^\delta g(Y_{(n-k,n)})}\Big|\stackrel{\mathbb{P}}{\longrightarrow}0$$
	In this direction, we observe that
	\begin{eqnarray*}
		\max_{0\leq k_0 \leq h(k)}\Big|\frac{c}{1+\rho}-\frac{cA(k-k_0)}{k(1+\rho)k^\delta g(Y_{(n-k,n)})}\Big|&\leq& \frac{|c|}{1+\rho}\max_{0\leq k_0<h(k)}\Bigg(\Big|1-\frac{ A}{k^\delta g(Y_{(n-k,n)})}\Big|+\frac{ Ak_0}{k^{\delta+1} g(Y_{(n-k,n)})}\Bigg)\\
		&\leq&\frac{|c|}{1+\rho}\Bigg(\Big|1-\frac{ A}{k^\delta g(Y_{(n-k,n)})}\Big|+\frac{ Ah(k)}{k^{\delta+1}g(Y_{(n-k,n)})}\Bigg) \stackrel{\mathbb{P}}{\longrightarrow}0
	\end{eqnarray*}
	since $h(k)/k \rightarrow 0$ and by  Relation \eqref{e:del-g}, $A/k^\delta g(Y_{(n-k,n)}) \stackrel{\mathbb{P}}{\longrightarrow} 1$. This completes the proof.
\end{proof}

\begin{lemma}
	\label{lem:s-def}
	Assumption \eqref{e:L-behav} implies
	\begin{equation}
	\label{e:rem-1}
	\max_{0\leq k_0 \leq k}\Big(	\frac{k-k_0}{kg(Y_{(n-k,n)})}|R_{k_0,k}-S_{k_0,k}|\Big)\stackrel{\mathbb{P}}{\longrightarrow}0
	\end{equation}
	where $R_{k_0,k}$ and $S_{k_0,k}$ are defined in Relations \eqref{e:tail-3} and \eqref{e:s-def}, respectively.
\end{lemma}

\begin{proof}
	The proof of Relation \eqref{e:rem-1} involves two cases:  $\rho>0$ and $\rho=0$.\vspace{2mm}
	
	\noindent{\em Case $\rho>0$:} 	Since $Y_{(n-i+1,n)}/Y_{(n-k,n)}>1$, $i=1,\cdots, k$,  therefore, over the event $\{Y_{(n-k,n)}>t_\varepsilon\}$, by Relation \eqref{e:L-behav}, we have
	\vspace{-2mm}	
	\begin{eqnarray*}
		(k-k_0)|R_{k_0,k}-S_{k_0,k}|&\leq &(k_0+1)\left|\log \frac{L(Y_{(n-k_0,n)})}{L(Y_{(n-k,n)})}-cg(Y_{(n-k,n)}) \int_{1}^{Y_{(n-k_0,n)}/Y_{(n-k,n)}} \nu^{-\rho-1} d\nu \right|\\
		&+&\sum_{i=k_0+2}^{k}\Bigg|\log \frac{L(Y_{(n-i+1,n)})}{L(Y_{(n-k,n)})}-cg(Y_{(n-k,n)} \int_{1}^{Y_{(n-i+1,n)}/Y_{(n-k,n)}} \nu^{-\rho-1} d\nu \Bigg|\\	
		&\leq&(k_0+1)g(Y_{(n-k,n)})\varepsilon+\sum_{i=k_0+2}^{k}g(Y_{(n-k,n)})\varepsilon =g(Y_{(n-k,n)})k\varepsilon.\vspace{-5mm}
	\end{eqnarray*} 
	Therefore, over the event $\{Y_{(n-k,n)}>t_\varepsilon\}$  
	\begin{equation}\label{e:lem:s-def-1}
	\max_{0\leq k_0 \leq k}\Big(	\frac{k-k_0}{kg(Y_{(n-k,n)})}|R_{k_0,k}-S_{k_0,k}|\Big) \leq \varepsilon.
	\end{equation}
	From Relation \eqref{e:Pareto-order}, we have $Y_{(n-k,n)}\stackrel{d}{=}(\Gamma_{k+1}/\Gamma_{n+1})^{-1}$ where ${(\Gamma_{k+1}/\Gamma_{n+1})}^{-1}\stackrel{a.s.}{\sim}n/k$ by Lemma \ref{lem:ratio-conv}. Since $n/k \to \infty$, therefore 	$$\mathbb{P}(Y_{(n-k,n)} >t_{\varepsilon})\rightarrow 1$$ which completes the proof. \vspace{2mm}
	
	\noindent{\em Case $\rho=0$:} As in the previous case, over the event $\{Y_{(n-k,n)}>t_\varepsilon\}$, by Relation \eqref{e:L-behav} we have	\begin{eqnarray}
	\label{e:diff-0-1}
	(k-k_0)|R_{k_0,k}-S_{k_0,k}|&=&(k_0+1)\left|\log \frac{L(Y_{(n-k_0,n)})}{L(Y_{(n-k,n)})}-cg(Y_{(n-k,n)} \int_{1}^{Y_{(n-k_0,n)}/Y_{(n-k,n)}} \frac{d\nu}{\nu}  \right|\\\nonumber
	&+&\sum_{i=k_0+2}^{k}\left|\log \frac{L(Y_{(n-i+1,n)})}{L(Y_{(n-k,n)})}-cg(Y_{(n-k,n)} \int_{1}^{Y_{(n-i+1,n)}/Y_{(n-k,n)}} \frac{d\nu}{\nu} \right|\\\nonumber
	&\leq&\varepsilon\Bigg((k_0+1)g(Y_{(n-k,n)})\Big(\frac{Y_{(n-k_0,n)}}{Y_{(n-k,n)}}\Big)^{\varepsilon}+\sum_{i=k_0+2}^{k}g(Y_{(n-k,n)})\Big( \frac{Y_{(n-i+1,n)}}{Y_{(n-k,n)}}\Big)^{\varepsilon}\Bigg)\nonumber
	\end{eqnarray}
	Since $Y_{(n-i+1,n)}\geq Y_{(n-k_0,n)}$ for $i=1, \cdots, k_0+1$,  we further obtain
	\begin{equation}\label{e:rho-0-second-instance}
	\max_{0\leq k_0 \leq k}\Big(	\frac{(k-k_0)}{kg(Y_{(n-k,n)})}|R_{k_0,k}-S_{k_0,k}|\Big)\leq \frac{\varepsilon}{k} \sum_{i=1}^{k}\Big(\frac{Y_{(n-i+1,n)}}{Y_{(n-k,n)}}\Big)^{\varepsilon}
	\end{equation}
	over the event $\{Y_{(n-k,n)}>t_\varepsilon\}$. The upper bound in \eqref{e:rho-0-second-instance} can be bounded by $2\varepsilon$ over the event $\{(1/k) \sum_{i=1}^{k}(Y_{(n-i+1,n)}/Y_{(n-k,n)})^{\varepsilon}<2\}$. 
	
	We have already proved that $\mathbb{P}(Y_{(n-k,n)} >t_{\varepsilon})\rightarrow 1$. Thus, to complete the proof of Relation \eqref{e:rem-1}, it only remains to show that
	\begin{equation}
	\label{e:y-eps-con}
	\mathbb{P}\left(\Big\{\frac{1}{k} \sum_{i=1}^{k}(\frac{Y_{(n-i+1,n)}}{Y_{(n-k,n)}})^{\varepsilon}<2\Big\}\right) \rightarrow 1.
	\end{equation}
	In this direction, from Relation \eqref{e:Pareto-order}, we observe that  
	$$\frac{1}{k} \sum_{i=1}^{k}\Big(\frac{Y_{(n-i+1,n)}}{Y_{(n-k,n)}}\Big)^{\varepsilon}\stackrel{d}{=}\frac{1}{k}\sum_{i=1}^{k}\Big(\frac{\Gamma_{i+1}}{\Gamma_{k+1}}\Big)^{-\varepsilon}=\frac{1}{k}\sum_{i=1}^{k}U_{i,k}^{-\varepsilon}\stackrel{\mathbb{P}}{\longrightarrow}\frac{1}{1-\varepsilon}$$
	where the last convergence follows from weak law of large numbers. Thus, Relation \eqref{e:y-eps-con} holds as long as $\varepsilon<0.5$.
	
	This completes the proof for $\rho=0$.
\end{proof}

\begin{lemma}
	\label{lem:g-asy}
	Suppose $g$ is $-\rho$-varying for $\rho\geq 0$ and $Y_{(n-k,n)}$ is the $(k+1)^{th}$ order statistic for $n$ observations from ${\rm Pareto}(1,1)$, then 
	\begin{equation}
	\label{e:g-asy}
	\frac{g(Y_{(n-k,n)})}{g(n/k)}\stackrel{\mathbb{P}}{\longrightarrow}1
	\end{equation}
	provided $k\rightarrow \infty$, $n \rightarrow \infty$ and $k/n \rightarrow \infty$.
\end{lemma}

\begin{proof}
	Since $g$ is $-\rho$ varying, $g$ may be expressed as $g(t)=t^{-\rho}l(t)$, for some slowly varying function $l(\cdot)$. Thus, we have
	$$\frac{g(Y_{(n-k,n)})}{g(n/k)}={\Big(\frac{Y_{(n-k,n)}}{n/k}\Big)}^{-\rho}\frac{l(Y_{(n-k,n)})}{l(n/k)}$$
	From Relation \eqref{e:Pareto-order}, we have $Y_{(n-k,n)}\stackrel{d}{=}\Gamma_{n+1}/\Gamma_{k+1}$ and therefore, by weak law of large numbers, we have $Y_{(n-k,n)}/(n/k)\stackrel{\mathbb{P}}{\longrightarrow}1$.
	
	Thus to prove Relation \eqref{e:g-asy}, it suffices to show $l(Y_{(n-k,n)})/l(n/k)\stackrel{\mathbb{P}}{\longrightarrow}1$. In this direction, observe that for any $\delta>0$, we have
	\begin{eqnarray*}
		\label{e:conv-p}
		\mathbb{P}\Big(\Big|\frac{l(Y_{(n-k,n)})}{l(n/k)}-1\Big|>\varepsilon\Big)&\leq&\mathbb{P}\Big(\Big|\frac{l(Y_{(n-k,n)})}{l(n/k)}-1\Big|>\varepsilon, \Big|\frac{Y_{(n-k,n)}}{n/k}-1\Big|\leq \delta\Big)+\mathbb{P}\Big(\Big|\frac{Y_{(n-k,n)}}{n/k}-1\Big|>\delta\Big)\\
		&\leq& \mathbb{P}\Big( \sup_{\lambda \in [1-\delta,1+\delta)}\Big|\frac{l(\lambda n/k)}{l(n/k)}-1\Big|>\varepsilon\Big)+\mathbb{P}\Big(\Big|\frac{Y_{(n-k,n)}}{n/k}-1\Big|>\delta\Big)
	\end{eqnarray*}
	For $\delta$ small enough, the first term on the right hand side  goes to 0 by Theorem 1.5.2 on page 22 in \cite{bingham1989regular}. Also, for $\delta$ small enough, the second term goes to 0 since $Y_{(n-k,n)}/(n/k)\stackrel{\mathbb{P}}{\longrightarrow}1$.	
	
	This completes the proof.
\end{proof}

\begin{lemma}
	\label{lem:r-def}
	\begin{equation}
	\label{e:R-conv-2}
	\max_{0 \leq k_0 <k}\Big|\frac{k-k_0}{kg(Y_{(n-k,n)})}S_{k_0,k}+\frac{c}{1+\rho}\left(\frac{k_0}{k}\right)^{1+\rho}-\frac{c}{1+\rho}\Big|\stackrel{\mathbb{P}}{\longrightarrow}0.
	\end{equation}
	where $S_{k_0,k}$ is defined in Relation \eqref{e:s-def}.
\end{lemma}

\begin{proof}
	The proof of Relation \eqref{e:R-conv-2} involves two cases:  $\rho>0$ and $\rho=0$.\vspace{2mm}
	
	\noindent{\em Case $\rho>0$:} Using the expression of $S_{k_0,k}$ in Relation \eqref{e:s-def}, we get
	\begin{eqnarray}\label{e:lem:r-def-1}
	\frac{k-k_0}{kg(Y_{(n-k,n)})} S_{k_0,k} &=&-\frac{c}{k\rho}\Bigg((k_0+1)\Big(\frac{Y_{(n-k_0,n)}}{Y_{(n-k,n)}}\Big)^{-\rho}+\sum_{i=k_0+2}^{k}\Big(\frac{Y_{(n-i+1,n)}}{Y_{(n-k,n)}}\Big)^{-\rho}-k\Bigg) \\
	&=&\frac{c}{k\rho}\sum_{i=1}^{k_0}
	\left\{\Big(\frac{Y_{(n-i+1,n)}}{Y_{(n-k,n)}}\Big)^{-\rho}-\Big(\frac{Y_{(n-k_0,n)}}{Y_{(n-k,n)}}\Big)^{-\rho}\right\} -\frac{c}{k\rho}\sum_{i=1}^{k}\left\{\Big(\frac{Y_{(n-i+1,n)}}{Y_{(n-k,n)}}\Big)^{-\rho}-1\right\}\nonumber
	\end{eqnarray}
	Expressing the order statistics of Pareto in terms of Gamma random variables as in Relation \eqref{e:Pareto-order}, we get
	\
	\begin{eqnarray*}
		\frac{k-k_0}{kg(Y_{(n-k,n)})} S_{k_0,k}+\frac{c}{1+\rho}{\left(\frac{k_0}{k}\right)}^{1+\rho} -\frac{c}{1+\rho}&\stackrel{d}{=}&
		\underbrace{\frac{c}{1+\rho}{\left(\frac{k_0}{k}\right)}^{1+\rho}+\frac{c}{k\rho}\sum_{i=1}^{k_0}\Big\{\Big(\frac{\Gamma_{i+1}}{\Gamma_{k+1}}\Big)^{\rho}-\Big(\frac{\Gamma_{k_0+1}}{\Gamma_{k+1}}\Big)^{\rho}\Big\}}_{B_{k_0,k}}\\&-&\underbrace{\left(\frac{c}{k\rho}\sum_{i=1}^{k}\Big\{\Big(\frac{\Gamma_{i+1}}{\Gamma_{k+1}}\Big)^{\rho}-1\Big\}+\frac{c}{1+\rho}\right)}_{A_k}
	\end{eqnarray*}
	\noindent In view of the above result, to prove \eqref{e:R-conv-2}, we first show that $\max_{0 \leq k_0 <k}|A_k| \stackrel{a.s.}{\longrightarrow}0.$ 
	
	Note that, by Relation \eqref{e:sum-rho-conv}, we have $|(1/k)\sum_{i=1}^{k}(\Gamma_{i+1}/\Gamma_{k+1})^{\rho}-1/(1+\rho)| \stackrel{a.s.}{\longrightarrow} 0 $. This implies that there  exists  $\Omega$ with $\mathbb{P}(\Omega)=1$ such that for any $\omega \in \Omega$,   
	$$\Big|A_k(\omega)+\frac{c}{1+\rho}\Big|=\Big|\frac{c}{\rho k}\sum_{i=1}^k\Big(\frac{\Gamma_{i+1}}{\Gamma_{k+1}}\Big)^{\rho}(\omega)-\frac{c}{\rho}+\frac{c}{1+\rho}\Big|=\Big|\frac{c_{\rho}}{k}\sum_{i=1}^k\Big(\frac{\Gamma_{i+1}}{\Gamma_{k+1}}\Big)^{\rho}(\omega)-\frac{c_\rho}{1+\rho}\Big|{\rightarrow} 0$$
	We next show that $\max_{0 \leq k_0 <k}B_{k_0,k} \stackrel{a.s.}{\longrightarrow}0$.  For this observe that for any $\omega \in \Omega$,
	\begin{eqnarray}
	\label{e:b-1}
	\max_{0 \leq k_0 <M}B_{k_0,k}(\omega)&\leq& \max_{0 \leq k_0 <M} \left\{\frac{c}{1+\rho}{\left(\frac{k_0}{k}\right)}^{1+\rho}+\frac{c}{k\rho}\sum_{i=1}^{k_0}\Big|\Big(\frac{\Gamma_{i+1}}{\Gamma_{k+1}}\Big)^{\rho}(\omega)-\Big(\frac{\Gamma_{k_0+1}}{\Gamma_{k+1}}\Big)^{\rho}(\omega)\Big|\right\}\\\nonumber
	&\leq&  \max_{0 \leq k_0 <M}\left\{\frac{c}{1+\rho}{\left(\frac{k_0}{k}\right)}^{1+\rho}+\frac{2ck_0}{k\rho}\right\}\hspace{2mm}(\textmd{since ${(\Gamma_i /\Gamma_{k+1})}^{\rho} \leq 1$, $1\leq i\leq k$,  $\rho>0$})\\\nonumber
	&\leq& \frac{cM^{1+\rho}/(1+\rho)+2cM/\rho}{k} =\frac{B_{0M}}{k}.
	\end{eqnarray}
	Additionally, we have
	\begin{eqnarray}
	\label{e:b-2}
	\max_{M \leq k_0 <k} B_{k_0,k}(\omega)&\leq& \max_{M \leq k_0 <k}  \left|c_{1+\rho}{\left(\frac{k_0}{k}\right)}^{1+\rho}+\frac{c_\rho}{k}\sum_{i=1}^{k_0}\Big\{\Big(\frac{\Gamma_{i+1}}{\Gamma_{k+1}}\Big)^{\rho}(\omega)-\Big(\frac{\Gamma_{k_0+1}}{\Gamma_{k+1}}\Big)^{\rho}(\omega)\Big\}\right|\\\nonumber
	&\leq&  \max_{M \leq k_0 <k}{\left(\frac{k_0}{k}\right)}^{1+\rho}\left|\frac{c_{1+\rho}}{c_\rho}+{\left(\frac{k_0}{k}\right)}^{\rho}\frac{1}{k_0}\sum_{i=1}^{k_0}\Big\{\Big(\frac{\Gamma_{i+1}}{\Gamma_{k+1}}\Big)^{\rho}(\omega)-\Big(\frac{\Gamma_{k_0+1}}{\Gamma_{k+1}}\Big)^{\rho}(\omega)\Big\}\right|\\\nonumber
	&\leq & \max_{M \leq k_0 <k}\Bigg|\frac{\rho}{1+\rho}+\Big(\frac{\Gamma_{k_0+1}/k_0}{\Gamma_{k+1}/k}\Big)^{\rho}(\omega)\Big\{\underbrace{\frac{1}{k_0}\sum_{i=1}^{k_0}\Big(\frac{\Gamma_{i+1}}{\Gamma_{k_0+1}}\Big)^{\rho}(\omega)}_{C_{k_0}(\omega)}-1\Big\}\Bigg|\\\nonumber
	&=&\max_{M \leq k_0 <k}\left|\frac{\rho}{1+\rho}+(C_{k_0}(\omega)-1)+(C_{k_0}(\omega)-1)\Big\{\Big(\frac{\Gamma_{k_0+1}/k_0}{\Gamma_{k+1}/k}\Big)^{\rho}-1\Big\}\right|.\nonumber
	\end{eqnarray}

	Since $\Gamma_{i+1}<\Gamma_{k_0+1}$ and $\rho>0$, thereby $|C_{k_0}|<1$. This allows us to simplify Relation \eqref{e:b-2} as
	\begin{eqnarray*}
		\max_{M \leq k_0 <k} B_{k_0,k}(\omega)&\leq&\underbrace{ \sup_{M \leq k_0 }\left|C_{k_0}(\omega)-\frac{1}{1+\rho}\right|}_{B_{1M}(\omega)}+2\underbrace{\sup_{M \leq k_0,k}\left|\Big(\frac{\Gamma_{k_0+1}/k_0}{\Gamma_{k+1}/k}\Big)^{\rho}(\omega)-1\right|}_{B_{2M}(\omega)}.
	\end{eqnarray*}
	Thus, we obtain $$\max_{0 \leq k_0 <k}B_{k_0,k}(\omega) \leq \frac{B_{0M}}{k}+B_{1M}(\omega)+B_{2M}(\omega).$$
	Taking $\limsup$ w.r.t to $k$ on both sides, we get
	\begin{equation}
	\label{e:k-sup}
	\limsup_{k \rightarrow \infty}\max_{0 \leq k_0 <k}B_{k_0,k}(\omega)\leq B_{1M}(\omega)+B_{2M}(\omega).
	\end{equation}
	Using Lemmas  \ref{lem:sum-conv} and \ref{lem:ratio-conv} shows that $B_{1M}(\omega) \rightarrow 0$  and $B_{2M}(\omega) \rightarrow 0$ for all $\omega \in \Omega$ with $\mathbb{P}(\Omega)=1$.
	
	Thus, taking  $\limsup$ w.r.t $M$ on both sides of Relation \eqref{e:k-sup} completes the proof for $\rho<0$.
	\vspace{2mm}
	
	\noindent {\em Case $\rho=0$:}  Using the expression of $S_{k_0,k}$ in Relation \eqref{e:s-def}, we get
	\begin{eqnarray}\label{e:rho-0-case}
	\frac{k-k_0}{kg(Y_{(n-k,n)})} S_{k_0,k}+\frac{ck_0}{k}-c &=&\frac{c}{k}\Bigg((k_0+1)\log \Big(\frac{Y_{(n-k_0,n)}}{Y_{(n-k,n)}}\Big)+\sum_{i=k_0+2}^{k}\log\Big(\frac{Y_{(n-i+1,n)}}{Y_{(n-k,n)}}\Big)\Bigg)-\frac{c(k-k_0)}{k}\nonumber\\
	&\stackrel{d}{=}&\frac{c(k-k_0)}{k}\widehat{\xi}^{**}_{k_0,k}-\frac{c(k-k_0)}{k}\nonumber\\
	&\stackrel{d}{=}&c\Big(\frac{\Gamma_{k-k_0}}{k}-\frac{k-k_0}{k}\Big),
	\end{eqnarray}
	where $\widehat{\xi}^{**}_{k_0,k}$ is the trimmed Hill estimator in Relation \eqref{e:xi-opt} with $X_i$'s replaced by the i.i.d. ${\rm Pareto}(1,1)$. The last distribution equality in Relation \eqref{e:rho-0-case} follows from Relation \eqref{e:xi-jt-big}.
	
	Thus, to prove Relation \eqref{e:R-conv-2}, we shall next show $\max_{0\leq k_0<k}|\Gamma_{k-k_0}-(k-k_0)|/k\stackrel{a.s.}{\longrightarrow}0$. In this direction, we have
	\begin{eqnarray}
	\label{e:max-gam}
	\max_{0\leq k_0<k}\frac{|\Gamma_{k-k_0}(\omega)-(k-k_0)|}{k}&=&\max_{0\leq k_0<k}\frac{(k-k_0)}{k}\Big|\frac{\Gamma_{k-k_0}}{k-k_0}(\omega)-1\Big|\\\nonumber
	&\leq&\frac{M}{k}\max_{0\leq k-k_0<M}\Big|\frac{\Gamma_{k-k_0}}{k-k_0}(\omega)-1\Big|+\sup_{k-k_0\geq M}\Big|\frac{\Gamma_{k-k_0}}{k-k_0}(\omega)-1\Big|\\
	&\leq&\frac{M}{k}\underbrace{\sup_n\Big|\frac{\Gamma_{n}}{n}(\omega)-1\Big|}_{B_0(\omega)}+\underbrace{\sup_{n\geq M} \Big|\frac{\Gamma_{n}}{n}(\omega)-1\Big|}_{B_{1M}(\omega)}
	\end{eqnarray}
	By SLLN, there exists $\Omega$ with $\mathbb{P}(\Omega)=1$ such that for every $\omega \in \Omega$,  $|\Gamma_n(\omega)/n-1| \to 0$ as $n \rightarrow \infty$. This also implies that $\sup_{n\geq M}|\Gamma_n(\omega)/n-1| \to 0$ as $M \to \infty$. Therefore, $B_0(\omega)$ is bounded and $B_{1M}(\omega)$ converges to 0.
	
	Thus, first taking $\limsup$ with respect to $k$  followed by $\lim$ with respect to $M$ on both sides of Relation \eqref{e:max-gam}, the proof follows.
\end{proof}

\subsubsection{Consistency of the Weighted Sequential Testing}
\label{sec:proof-sec-cons}

\begin{proof}[{\bf Proof of Theorem \ref{prop:T-conv}}]
	From Relation \eqref{e:T-i-k}, we have
	\begin{eqnarray*}
		\label{e:T-conv-mod}
		k^{\delta}\max_{0\leq k_0<h(k)}|T_{k_0,k}-T^*_{k_0,k}|&=&k^\delta \max_{0\leq k_0 <h(k)}\frac{k-k_0-1}{k-k_0}\Big|\frac{\widehat{\xi}_{k_0+1,k}}{\widehat{\xi}_{k_0,k}}-\frac{\widehat{\xi}^*_{k_0+1,k}}{\widehat{\xi}^*_{k_0,k}}\Big|\\
		&\leq&\frac{k^{\delta}}{1-h(k)/k}\max_{0\leq k_0 <h(k)}\underbrace{\Big|\frac{\widehat{\xi}_{k_0+1,k}}{\widehat{\xi}_{k_0,k}}-\frac{\widehat{\xi}^*_{k_0+1,k}}{\widehat{\xi}^*_{k_0,k}}\Big|}_{W_{k_0,k}}
	\end{eqnarray*}
	We shall show that the upper bound in the above relation converges to 0. To see this, note that $h(k)=o(k)$ which implies $1-h(k)/k \to 1$. Thus, it only remains to prove $
	k^\delta \max_{0\leq k_0 <h(k)}W_{k_0,k} \stackrel{\mathbb{P}}{\longrightarrow}0$. To this end, we observe that
	\begin{eqnarray*}
		W_{k_0,k} &\leq&\Big|\frac{\widehat{\xi}_{k_0+1,k}}{\widehat{\xi}_{k_0,k}}-\frac{\widehat{\xi}^*_{k_0+1,k}}{\widehat{\xi}_{k_0,k}}-\frac{cAk^{-\delta}}{(1+\rho)\widehat{\xi}_{k_0,k}}\Big|+\frac{|c|Ak^{-\delta}}{(1+\rho)\widehat{\xi}_{k_0,k}}\Big|1-\frac{\widehat{\xi}^*_{k_0+1,k}}{\widehat{\xi}^*_{k_0,k}}\Big|\\
		&+&\Big|\frac{\widehat{\xi}^*_{k_0+1,k}}{\widehat{\xi}_{k_0,k}}-\frac{\widehat{\xi}^*_{k_0+1,k}}{\widehat{\xi}^*_{k_0,k}}+\frac{cAk^{-\delta}}{(1+\rho)\widehat{\xi}_{k_0,k}}\frac{\widehat{\xi}^*_{k_0+1,k}}{\widehat{\xi}^*_{k_0,k}}\Big|\\
		&=&\frac{1}{\widehat{\xi}_{k_0,k}}\Bigg(\Big|R_{k_0,k}-\frac{cAk^{-\delta}}{1+\rho}\Big|+\frac{|c|Ak^{-\delta}}{(1+\rho)}\Big|1-\frac{\widehat{\xi}^*_{k_0+1,k}}{\widehat{\xi}^*_{k_0,k}}\Big|+\frac{\widehat{\xi}^*_{k_0+1,k}}{\widehat{\xi}^*_{k_0,k}}\Big|\frac{cAk^{-\delta}}{(1+\rho)}-R_{k_0+1,k}\Big|\Bigg)\\
	\end{eqnarray*}
	\normalsize
	where $R_{k_0,k}$ is defined in Relation \eqref{e:tail-3}. Thus, the quantity $
	k^\delta \max_{0\leq k_0 <h(k)}W_{k_0,k}$ is bounded above as follows
	\begin{eqnarray}
	\label{e:w-m-b}
	\max_{0\leq k_0 <h(k)}k^\delta W_{k_0,k}&\leq& \Big(M_{1k}+\frac{|c|A}{(1+\rho)}\max_{0\leq k_0  h(k)}|1-B_{k_0,k}|+M_{1k}\max_{0\leq k_0 \leq h(k)}B_{k_0,k}\Big)\max_{0\leq k_0 <h(k)} \frac{1}{\widehat{\xi}_{k_0,k}}\nonumber\\
	&=&\Big(M_{1k}\max_{0\leq k_0 \leq h(k)}(1+B_{k_0,k})+\frac{|c|A}{(1+\rho)}\max_{0\leq k_0 \leq h(k)}|1-B_{k_0,k}|\Big)\max_{0\leq k_0 <h(k)} \frac{1}{\widehat{\xi}_{k_0,k}}\nonumber\\
	\end{eqnarray}
	\vspace{-3mm}
	where $$M_{1k}:=k^\delta \max_{0\leq k_0<h(k)}\Big|R_{k_0,k}-\frac{cAk^{-\delta}}{1+\rho}\Big|\quad \mbox{and}\quad B_{k_0,k}:=\frac{\widehat{\xi}^*_{k_0+1,k}}{\widehat{\xi}^*_{k_0,k}}.$$
	Theorem \ref{prop:E-conv} implies $M_{1k}\stackrel{\mathbb{P}}{\longrightarrow}0$ as $k \to \infty$. On the other hand, by Relation \eqref{e:xi-jt}, 
	\begin{eqnarray*}
		\label{e:B-conv}
		\max_{0\leq k_0 \leq h(k)}|1-B_{k_0,k}|&\stackrel{d}{=}&\max_{0\leq k_0 \leq h(k)}\Big|1-\frac{\Gamma_{k-k_0-1}/(k-k_0-1)}{\Gamma_{k-k_0}/(k-k_0)}\Big|\\\nonumber
		&\leq &\frac{1}{1-h(k)/k}\max_{k-h(k)\leq i \leq k}\Big|\frac{\Gamma_{i}/i}{\Gamma_{i+1}/(i+1)}-1\Big|\stackrel{a.s.}{\longrightarrow}0,
	\end{eqnarray*} where the last convergence is a consequence of   \eqref{e:gam-conv-1}. Using a similar argument, one can also show that $\max_{0\leq k_0 \leq h(k)}(1+B_{k_0,k})=O_{\mathbb{P}}(1)$.
	
	We shall end the proof by showing that that $\min_{0\leq k_0 <h(k) }|\widehat{\xi}_{k_0,k}|$ is bounded away from 0 in probability and thus in view of Relation \eqref{e:w-m-b}, the convergence of 	$\max_{0\leq k_0 <h(k)}k^\delta W_{k_0,k}$ to 0 in probability shall follow. To this end, we have,
	\begin{eqnarray}
	\label{e:xi-bound}
	\min_{0\leq k_0 <h(k) }\widehat{\xi}_{k_0,k}\geq \min_{0\leq k_0 <h(k) }\widehat{\xi}^*_{k_0,k}-\max_{0\leq k_0 <h(k) }|\widehat{\xi}_{k_0,k}-\widehat{\xi}^*_{k_0,k}|
	\end{eqnarray}
	For $\delta>0$,  Theorem \ref{prop:E-conv} implies $\max_{0\leq k_0 <h(k) }|\widehat{\xi}_{k_0,k}-\widehat{\xi}^*_{k_0,k}| \stackrel{\mathbb{P}}{\longrightarrow}0$. Therefore $\min_{0 \leq k_0 <h(k)}\widehat{\xi}_{k_0,k}$ is bounded away from 0 as long as  $\min_{0 \leq k_0 <h(k)}\widehat{\xi}^*_{k_0,k}$ is bounded away from 0. This is easy to show because
	$$\min_{0 \leq k_0 <h(k)}\widehat{\xi}^*_{k_0,k}\stackrel{d}{=}\min_{0 \leq k_0 <h(k)}\frac{\Gamma_{k-k_0}}{k-k_0}\geq 1-\max_{k-h(k) \leq i <k}\Big|\frac{\Gamma_{i}}{i}-1\Big|\stackrel{a.s.}{\longrightarrow}1$$
	where the last convergence is a direct consequence of Relation \eqref{e:gam-conv-0}. This completes the proof.
\end{proof}
\vspace{-1mm}

\begin{proof}[{\bf Proof of Theorem \ref{prop:U-conv}}]
	{\em Proof of Relation \eqref{e:U-conv}}:	By Relation \eqref{e:U-i-k}, we have
	\begin{eqnarray}
	\label{e:u-conv-1}
	k^{\delta-1} \max_{0\leq k_0<h(k)}|U_{k_0,k}-U^*_{k_0,k}|&=&	2k^{\delta-1}\max_{0\leq k_0<h(k)}\Big||{(T_{k_0,k})}^{k-k_0-1}-0.5|-|{(T^*_{k_0,k})}^{k-k_0-1}-0.5|\Big|\nonumber\\
	&\leq&	2k^{\delta-1} \max_{0\leq k_0<h(k)}\Big|{(T_{k_0,k})}^{k-k_0-1}-{(T^*_{k_0,k})}^{k-k_0-1}\Big|\\
	&\leq&	2k^{\delta-1}\max_{0\leq k_0<h(k)}\Big|\Big(\frac{T_{k_0,k}}{T^*_{k_0,k}}\Big)^{k-k_0-1}-1\Big|\nonumber
	\end{eqnarray}
	where the last bound follows $T^*_{k_0,k}\leq 1$ (see Proposition \ref{T-def}). In view of Relation \eqref{e:u-conv-1}, to prove Relation \eqref{e:U-conv}, it suffices to show
	\begin{equation}
	\label{e:U-conv-mod}
	k^{\delta-1}\max_{0\leq k_0<h(k)}\Big|\Big(\frac{T_{k_0,k}}{T^*_{k_0,k}}\Big)^{k-k_0-1}-1\Big|\stackrel{\mathbb{P}}{\longrightarrow}0.
	\end{equation}
	To this end, we begin by showing  \begin{equation} \label{e:TU-conv-mod}
	k^{\delta}\max_{0\leq k_0<h(k)}\Big|\frac{T_{k_0,k}}{T^*_{k_0,k}}-1\Big|\stackrel{\mathbb{P}}{\longrightarrow}0.\end{equation}
	In this direction, observe that
	$$k^{\delta}\max_{0\leq k_0<h(k)}\Big|\frac{T_{k_0,k}}{T^*_{k_0,k}}-1\Big|\leq \frac{1}{\min_{0\leq k_0< h(k)}T^*_{k_0,k}}\max_{0\leq k_0<h(k)}\underbrace{k^{\delta}|T_{k_0,k}-T^*_{k_0,k}|}_{\Delta_{k}},$$
	where $\Delta_{k}\stackrel{\mathbb{P}}{\longrightarrow}0$ by Relation  \eqref{e:T-conv}. Thus, Relation \eqref{e:TU-conv-mod} holds as long as $\min_{0\leq k_0< h(k)}T^*_{k_0,k}$ is bounded away from 0 in probability as shown next.
	$$\min_{0 \leq k_0 <h(k)}T^*_{k_0,k}\stackrel{d}{=}\min_{0 \leq k_0 <h(k)}\frac{\Gamma_{k-k_0-1}/(k-k_0-1)}{\Gamma_{k-k_0}/(k-k_0)}\geq 1-\max_{k-h(k) \leq i <k}\Big|\frac{\Gamma_{i}/i}{\Gamma_{i+1}/(i+1)}-1\Big|\stackrel{a.s.}{\longrightarrow}1,$$
	where the last convergence is a direct consequence of Relation \eqref{e:gam-conv-1}.  
	
	Finally to prove Relation \eqref{e:U-conv-mod}, we shall equivalently show that  for every subsequence $\{k_l\}$, there exists a further subsequence ${\widetilde{k}}$ such that
	\begin{equation}
	\label{e:U-conv-as}
	{\widetilde{k}}^{\delta-1}\max_{0\leq k_0<h({\widetilde{k}})}\Big|\Big(\frac{T_{k_0,{\widetilde{k}}}}{T^*_{k_0,{\widetilde{k}}}}\Big)^{{\widetilde{k}}-k_0-1}-1\Big|\stackrel{a.s.}{\longrightarrow}0.
	\end{equation}
	This is shown next as follows.	In view of Relation \eqref{e:TU-conv-mod}, for every subsequence $\{k_l\}$, there exists a further subsequence ${\widetilde{k}}$ such that
	$$\widetilde{k}^{\delta}\max_{0\leq k_0<h(\widetilde{k})}\Big|\frac{T_{k_0,\widetilde{k}}}{T^*_{k_0,\widetilde{k}}}-1\Big|\stackrel{a.s.}{\longrightarrow}0.$$
	Hence, there is event an $\Omega$ with $\mathbb{P}(\Omega)=1$ such that for every $\epsilon>0$, there exist a $M=M(\omega, \epsilon)$
	\begin{equation} \label{e:ineq-1}
	1-\frac{\epsilon}{\widetilde{k}^\delta}\leq\Big(\frac{T_{k_0,\widetilde{k}}}{T^*_{k_0,\widetilde{k}}}\Big)(\omega)\leq 1+\frac{\epsilon}{\widetilde{k}^\delta},\hspace{5mm} \textmd{for all} \:\:  \widetilde{k}\geq M,\: 0\leq k_0<h(\widetilde{k}) \:\mbox{and}\: \omega \in \Omega\end{equation}
	Therefore,
	\begin{eqnarray*}
		\underbrace{\widetilde{k}^{\delta-1}\Big (\Big(1-\frac{\epsilon}{\widetilde{k}^\delta}\Big)^{\widetilde{k}-h(\widetilde{k})-1}-1\Big)}_{-a_{\widetilde{k}}}\leq& \widetilde{k}^{\delta-1}\Big({\Big(\frac{T_{k_0,\widetilde{k}}}{T^*_{k_0,\widetilde{k}}}\Big)}^{\widetilde{k}-k_0-1}(\omega)-1\Big)&\leq \underbrace{\widetilde{k}^{\delta-1}\Big (\Big(1+\frac{\epsilon}{\widetilde{k}^\delta}\Big)^{\widetilde{k}-1}-1\Big)}_{b_{\widetilde{k}}}\end{eqnarray*}
	which equivalently implies
	\begin{equation}
	\label{e:T-lim-1}
	\widetilde{k}^{{\delta-1}}\max_{0\leq k_0<h(\widetilde{k})}\Big|\Big(\frac{T_{k_0,\widetilde{k}}}{T^*_{k_0,\widetilde{k}}}\Big)^{\widetilde{k}-k_0-1}(w)-1\Big|\leq a_{\widetilde{k}}\vee b_{\widetilde{k}}
	\end{equation}
	Note that both the sequences $a_{\widetilde{k}}$ and $b_{\widetilde{k}}$ converge to $\epsilon$ as $\widetilde{k}\rightarrow \infty$. Thereby, taking limsup w.r.t $\widetilde{k}$ on both sides of Relation \eqref{e:T-lim-1}, we get
	\begin{equation}
	\label{e:T-lim-2}\limsup_{\widetilde{k}\rightarrow \infty}\widetilde{k}^{{\delta-1}}\max_{0\leq k_0<h(\widetilde{k})}\Big|\Big(\frac{T_{k_0,\widetilde{k}}}{T^*_{k_0,\widetilde{k}}}\Big)^{\widetilde{k}-k_0-1}(w)-1\Big| \leq \epsilon
	\end{equation}
	Since Relation \eqref{e:T-lim-2} holds for all $\epsilon>0$ and $\omega \in \Omega$ with $\mathbb{P}(\Omega)=1$, we have
	$$\widetilde{k}^{{\delta-1}}\max_{0\leq k_0<h(\widetilde{k})}\Big|\Big(\frac{T_{k_0,\widetilde{k}}}{T^*_{k_0,\widetilde{k}}}\Big)^{\widetilde{k}-k_0-1}(w)-1\Big|\rightarrow 0$$
	This entails the proof of the convergence in probability of Relation  \eqref{e:U-conv-mod}.
	
	\noindent{\em Proof of Relation \eqref{e:type-I-gen}}. To this end, we show that  $\mathbb{P}_{\mathcal{H}_0}(\widehat{k}_0=0)\rightarrow 1-q$. We first provide an upper bound on $\mathbb{P}_{\mathcal{H}_0}(\widehat{k}_0=0)$ as follows.
	\begin{eqnarray*}
		\mathbb{P}_{\mathcal{H}_0}(\widehat{k}_0=0)&=&\mathbb{P}_{\mathcal{H}_0}\Big(\underbrace{\bigcap_{i=0}^{k-2}\{U_{i,k}<(1-q)^{ca^{k-i-1}}\}}_{A_k}\Big)\\
		&\leq & \mathbb{P}_{\mathcal{H}_0}\Big(A_k \cap \underbrace{ \{k^{\delta-1}\max_{0\leq i\leq k-2}(U_{i,k}-U^*_{i,k})<\epsilon\}}_{B_{1k}}\Big)+\mathbb{P}_{\mathcal{H}_0}( \{k^{\delta-1}\max_{0\leq i\leq k-2}(U_{i,k}-U^*_{i,k})>\epsilon\})\\
		&\leq& \mathbb{P}_{\mathcal{H}_0}\Big(\underbrace{\bigcap_{i=0}^{k-2}\{U^*_{i,k}<(1-q)^{ca^{k-i-1}}+\epsilon k^{1-\delta}\}}_{A^*_{1k}}\Big)+\mathbb{P}_{\mathcal{H}_0}(B_{1k}^c) \hspace{4mm}(\textmd{since $A_k \cap B_{1k} \implies A^*_{1k}$})
	\end{eqnarray*}
	By Relation \eqref{e:U-conv}, $\mathbb{P}_{\mathcal{H}_0}(B_{1k}^c) \rightarrow 0$. Additionally, we have shown below that $\limsup_{k \to \infty} \mathbb{P}_{\mathcal{H}_0}(A^*_{1k}) \leq 1-q$ which implies $\limsup_{k \to \infty}\mathbb{P}(A_k) \leq 1-q$.
	
	Since $U_{i,k}^*$ are i.i.d. $U(0,1)$, therefore
	\begin{eqnarray*}
		\mathbb{P}_{\mathcal{H}_0}(A^*_{1k}) &=& (1-q)\prod_{i=0}^{k-2}\Big(1+\frac{\epsilon k^{1-\delta}}{(1-q)^{ca^{k-i-1}}}\Big) \\
		&\leq& (1-q)\underbrace{\Big(1+\frac{\epsilon k^{1-\delta}}{(1-q) }\Big)^{k-2}}_{c_{1k}} \quad  (\mbox{	since $(1-q)^{ca^{k-i-1}}\geq (1-q)$
		})	\end{eqnarray*}
	
	For $\delta\geq 2$, $k-2=O(k^{\delta-1})$ which implies  $\limsup_{k \rightarrow \infty}c_{1k} \leq(1+ M\epsilon/(1-q))$ for some $M>0$. Since this holds for all $\epsilon>0$, we have
	$$\limsup_{k \rightarrow \infty}\mathbb{P}_{\mathcal{H}_0}(A^*_{1k}) \leq (1-q)$$
	Finally, we provide a lower bound on $\mathbb{P}_{\mathcal{H}_0}(\widehat{k}_0=0)=\mathbb{P}_{\mathcal{H}_0}(A_k)$ as follows:
	\begin{eqnarray*}
		\mathbb{P}_{\mathcal{H}_0}(A_k)&\geq& \mathbb{P}_{\mathcal{H}_0}\Big(A_k \cap \underbrace{\{k^{\delta-1}\max_{0\leq i< k}(U_{i,k}-U^*_{i,k})>-\epsilon\}}_{B_{2k}}\Big)\\
		&\geq& \mathbb{P}_{\mathcal{H}_0}\Big(\underbrace{\bigcap_{i=0}^{k-2}\{U^*_{i,k}<(1-q)^{ca^{k-i-1}}-\epsilon k^{1-\delta}\}}_{A^*_{2k}} \cap B_{2k} \Big)=\mathbb{P}_{\mathcal{H}_0}(A_{2k}^*)-\mathbb{P}_{\mathcal{H}_0}(B_{2k}^c),\hspace{4mm}
	\end{eqnarray*}
	since $A^*_{2k}\cap B_{2k} \implies A_k\cap B_{2k} $.  By Relation \eqref{e:U-conv},  $\mathbb{P}(B_{2k}^c) \rightarrow 0$. Additionally, it has been shown below that $\liminf_{k \to \infty} \mathbb{P}_{\mathcal{H}_0}(A^*_{2k}) \geq 1-q$ which implies $\liminf_{k \to \infty} \mathbb{P}_{\mathcal{H}_0}\geq 1-q$.
	\begin{eqnarray*}
		\mathbb{P}_{\mathcal{H}_0}(A^*_{2k})&=&(1-q)\Big(1-\frac{\epsilon k^{1-\delta}}{(1-q)^{ca^{k-i-1}}}\Big)\\
		&\geq& (1-q)\underbrace{\Big(1-\frac{\epsilon}{(1-q) k^{\delta-1}}\Big)^{k-2}}_{c_{2k}}  \hspace{2mm}(\textmd{since $(1-q)^{ca^{k-i-1}}\geq (1-q)$})
	\end{eqnarray*}	
	For $\delta\geq 2$,  $k=O(k^{\delta-1})$ which implies $\liminf_{k \rightarrow \infty}c_{1k} \geq(1- M\epsilon/(1-q))$ for some $M>0$. Since this 	holds for every $\epsilon >0$, we have
	$$
	\liminf_{k \rightarrow \infty}\mathbb{P}_{\mathcal{H}_0}(A^*_{2k}) \geq (1-q)
	$$ 
	This completes the proof.
\end{proof}

\end{document}